\documentclass[11pt]{article}
\usepackage[letterpaper, margin=1in]{geometry}
\usepackage{authblk}

\title{Suffix Random Access via Function Inversion:\\
A Key for Asymmetric Streaming String Algorithms}

\date{\vspace{-1cm}}

\author[1]{Panagiotis Charalampopoulos}
\author[2]{Taha El Ghazi}
\author[3]{Jonas Ellert}
\author[4]{Pawe{\l} Gawrychowski}
\author[2]{Tatiana Starikovskaya} 

\affil[1]{King's College London}
\affil[2]{ENS Paris}
\affil[3]{CWI Amsterdam}
\affil[4]{Institute of Computer Science, University of Wrocław}

\usepackage{tikz}
\usetikzlibrary{patterns}
\usetikzlibrary{calc} 
\usetikzlibrary{decorations.markings,decorations.pathreplacing}
\usetikzlibrary{math}\tikzset{crossmark/.style={thick,inner sep=1.5pt}}

\usepackage{makecell}
\usepackage{lmodern}
\usepackage{xspace}
\usepackage[T1]{fontenc}
\usepackage{amsmath}
\usepackage{amssymb}
\usepackage{amsthm}
\usepackage{mathtools}
\usepackage{pifont}
\usepackage{hyperref}
\hypersetup{
      colorlinks=false,
      linkcolor=blue,
      urlcolor=cyan,
}
\usepackage[noadjust,sort]{cite}
\usepackage[capitalise]{cleveref}
\usepackage{thmtools}
\usepackage{thm-restate}
\usepackage{algorithm}
\usepackage[noend]{algpseudocode}
\usepackage{subcaption}
\usepackage{comment}
\usepackage{pgfplots}

\usepackage{xcolor}
\definecolor{my-green}{HTML}{66c2a5}
\definecolor{my-red}{HTML}{fc8d62}
\definecolor{my-blue}{HTML}{8da0cb}
\definecolor{my-purple}{HTML}{e78ac3}

\usepackage{contour}
\usepackage{tikz}
\usetikzlibrary{positioning,calc,fit,arrows.meta,patterns}

\newcommand{\prob}{{\textnormal{{Pr}}}}

\newcommand{\cO}{{\mathcal O}}
\newcommand{\tO}{{\tilde\cO}}
\let\cOtilde\tO

\newcommand{\rlz}{\mathsf{rlz}}
\newcommand{\per}{{\mathsf{per}}}
\newcommand{\ident}{{\mathsf{sync}}}

\newcommand{\rev}{{\mathsf{rev}}}
\newcommand{\rot}{\mathrm{rot}}

\let\unvarphi\phi
\let\phi\varphi
\let\eps\epsilon
\let\err\lambda

\def\tauconst{{\chi}}

\newcommand{\N}{\mathcal{N}}

\def\dd{\mathinner{.\,.}}

\newcommand{\ceil}[1]{{\lceil#1\rceil}}
\newcommand{\floor}[1]{{\lfloor#1\rfloor}}
\newcommand{\absolute}[1]{{\lvert#1\rvert}}
\def\poly{\textnormal{poly}}
\def\polylog{\textnormal{polylog}}

\newtheorem{theorem}{Theorem}[section]
\newtheorem{corollary}[theorem]{Corollary}
\newtheorem{lemma}[theorem]{Lemma}
\newtheorem{fact}[theorem]{Fact}
\newtheorem{observation}[theorem]{Observation}
\newtheorem{claim}[theorem]{Claim}

\newtheorem{remark}[theorem]{Remark}
\newtheorem{assumption}[theorem]{Assumption}

\theoremstyle{definition}
\newtheorem{definition}[theorem]{Definition}

\def\defaultproblemfont{\sffamily}

\newlength{\probboxminipagewidth}
\newlength{\fboxsepbak}
\newcommand{\defdetailedcustomproblem}[4][\defaultproblemfont]{%
\setlength{\fboxsepbak}{\fboxsep}
\setlength{\fboxsep}{.5em}
\setlength{\probboxminipagewidth}{\textwidth - 2\fboxsep - 2\fboxrule}%
\def\tmp{#3}
\par\vspace{\topsep}\noindent\fbox{%
  \begin{minipage}{\probboxminipagewidth}
    {#1{{\large#2}}} \par\smallskip
    \ifx\tmp\relax \else#3\par\smallskip\fi
    \foreach \fieldname/\fieldcontent in {#4} {
        {\bf{\fieldname:}} \fieldcontent \par\smallskip
    }
  \end{minipage}%
  }%
\par\vspace{\topsep}%
\setlength{\fboxsep}{\fboxsepbak}%
}

\newcommand{\defcustomproblem}[3][\defaultproblemfont]{%
\defdetailedcustomproblem[#1]{#2}{}{#3}%
}

\usepackage{enumitem}
\usepackage{lipsum}

\setlist[itemize]{wide=\parindent}

\begin{document}

\maketitle

\thispagestyle{empty}

\begin{abstract}
Many string processing problems can be phrased in the streaming setting,
where the input arrives symbol by symbol and we have sublinear working space.
The area of streaming algorithms for string processing has flourished since the seminal work of Porat and Porat [FOCS 2009].

Unfortunately, problems with efficient solutions in the classical setting often do not admit efficient solutions in the streaming setting.
As a bridge between these two settings,
Saks and Seshadhri [SODA 2013]
introduced the \emph{asymmetric streaming model} (see also [Andoni, Krauthgamer, and Onak; FOCS 2010]).
Here, one is given read-only access to
a (typically short) reference string $R$ of length $m$, while a (typically long) text~$T$ arrives as a stream.

We provide a generic technique to reduce fundamental string problems in the asymmetric streaming model to the online read-only model, lifting several existing algorithms and generally improving upon the state of the art.
Most notably, we obtain asymmetric streaming algorithms for \emph{exact} and \emph{approximate pattern matching} (under both the Hamming and edit distances), and for \emph{relative Lempel--Ziv compression}, a popular scheme for measuring and exploiting redundancy in repetitive text collections.

\smallskip

At the heart of our approach lies a novel tool that facilitates efficient computation in
the asymmetric streaming model: the \emph{suffix random access data structure}.
In its simplest variant, it maintains constant-time random access to the longest suffix of (the seen prefix
of) $T$ that occurs in~$R$.
Let $\tau$ be a parameter that denotes the size of the data structure.
A straightforward approach maintains the data structure in $\cO(m/\tau)$ time per arriving symbol of~$T$.

We drastically improve this tradeoff and reveal fundamental barriers
via a bidirectional reduction between suffix random access and \emph{function inversion}, a central problem in cryptography:
\begin{itemize}
\item By leveraging Fiat and Naor's function inversion data structure [SIAM J.~Comput.~2000],  we achieve $\tO(1+m^3 / \tau^6)$ update time.\footnote{The $\tO(\cdot)$ and $\tilde{\Omega}(\cdot)$ notations suppress, respectively, $\log^{\cO(1)}N$ and $1 / \log^{\cO(1)}N$ factors, where $N$ is the input-size.}
In particular, for $\tau = \sqrt{m}$, we obtain $\tO(1)$ update time, improving over the $\Omega(\sqrt{m})$ bound of the straightforward solution.
\item We establish an unconditional $\tilde{\Omega}(m / \tau^3)$ lower bound on the update time.
Additionally, we show that achieving update time $o(m^3 / \tau^7)$ would imply a breakthrough in function inversion.
\end{itemize}

On the way to our upper bound, we propose a variant of the
string synchronizing sets ([Kempa and Kociumaka; STOC 2019]) with a local sparsity condition that, as we show, admits an efficient streaming construction algorithm.
We believe that our framework and techniques will find broad applications in the development of small-space string algorithms.
\end{abstract}

\clearpage

\thispagestyle{empty}
\tableofcontents
\clearpage

\pagenumbering{arabic}
\setcounter{page}{1}

\section{Introduction}
The streaming model is designed to capture problems in which we need to efficiently process very large amounts of incoming data.
The common assumption is that we do not have enough space to store the input. 
Ideally, we would like to
use space polylogarithmic in the size of the data, perhaps at the expense of introducing randomization or returning approximate
results (especially if it is provably impossible to return exact results).
Many problems concerning strings (sequences of symbols) can be naturally expressed in the streaming model.
While string algorithms in the classical setting have been studied intensively for many decades, a systematic study of algorithmic problems on strings in the streaming setting was initiated relatively recently by a seminal work of Porat and Porat~\cite{porat2009optimal} on exact and approximate pattern matching.
It was followed by a 
flurry of results for pattern matching~\cite{DBLP:conf/esa/CliffordFPSS15,DBLP:conf/esa/GolanP17, DBLP:conf/soda/CliffordFPSS16, starikovskaya:LIPIcs:2017:7320, DBLP:conf/icalp/GolanKP18,DBLP:journals/algorithmica/GolanKP19,10.1145/3357713.3384266,DBLP:journals/algorithmica/GawrychowskiS22,clifford2018streaming,DBLP:conf/cpm/GolanKKP20, DBLP:journals/iandc/RadoszewskiS20,DBLP:conf/focs/KociumakaPS21,DBLP:conf/soda/DudekGGS22,DBLP:conf/icalp/Bhattacharya023}, repetition detection~\cite{DBLP:conf/isaac/GhaziS25,Ergun:10,stream-periodicity-mismatches, stream-periodicity-wildcards,DBLP:journals/algorithmica/GawrychowskiMSU19, DBLP:conf/cpm/MerkurevS19, DBLP:conf/spire/MerkurevS19, DBLP:conf/cpm/GawrychowskiRS19}, and formal language recognition~\cite{DBLP:journals/siamcomp/MagniezMN14,ganardi_et_al:LIPIcs:2018:9131, DBLP:conf/lata/GanardiHL18, ganardi_et_al:LIPIcs:2018:8485, ganardi_et_al:LIPIcs:2016:6853, DBLP:conf/mfcs/GanardiJL18, DBLP:journals/tcs/BabuLRV13, franois_et_al:LIPIcs:2016:6355,DBLP:conf/icalp/BathieS21,DBLP:conf/isaac/BathieKS23,ganardi_et_al:LIPIcs:2019:11502,DBLP:journals/theoretics/GanardiHLMS25} in the streaming model.

However, the streaming model is very restrictive, and some natural string processing problems do not admit sublinear-space streaming algorithms.
For example, Bathie, Charalampopoulos, and Starikovskaya~\cite{DBLP:conf/cpm/BathieCS24} showed that this is the case for longest common substring and circular pattern matching.
On the other hand, both of these problems do admit efficient solutions in the read-only setting~\mbox{\cite{DBLP:conf/cpm/BathieCS24,DBLP:conf/esa/KociumakaSV14}}, where one is given constant-time random access to the input and does not need to account for the space required to store it.

This suggests that one should seek a computational model (i) relevant to applications in which we need to process long sequences of incoming
symbols, and (ii) powerful enough to allow for non-trivial space-efficient solutions.
A promising option is the \emph{asymmetric streaming} model proposed by Saks and Seshadhri~\cite{doi:10.1137/1.9781611973105.122} (see also the work of Andoni, Krauthgamer, and Onak~\cite{asymED}). 
In this model, the input typically consists of a short read-only string that we are allowed to preprocess and a long streaming string that arrives symbol by symbol.
The algorithm must only account for any extra space used while processing $T$, that is, the space required to store the read-only string is not counted toward the space complexity. 
So far, existing results in this model consider edit distance and longest common subsequence computation~\cite{doi:10.1137/1.9781611973105.122,7354391,li_et_al:LIPIcs.FSTTCS.2021.27,cheng_et_al:LIPIcs.ICALP.2021.54}.
We thus ask:

\begin{center}
\textit{Which other fundamental string processing problems admit\\ efficient solutions in the asymmetric streaming model?}
\end{center}

Towards answering this question, we introduce a novel tool for string algorithms in the asymmetric streaming model:
the \emph{suffix random access data structure}.
We demonstrate that this data structure serves as a unifying primitive for asymmetric streaming string processing.
In particular, we show that:
\begin{enumerate}[label=(\alph*)]
\item It enables a generic reduction that transforms online read-only algorithms to asymmetric streaming algorithms for a wide class of pattern matching problems.
In particular, this reduction yields the first \emph{deterministic online} algorithms for exact and approximate pattern matching (under both the Hamming and edit distances)
with strongly sublinear space usage on top of read-only access to $P$ (without random access to $T$).
\item It leads to the first efficient asymmetric streaming algorithm for relative Lempel--Ziv factorization~\cite{10.1007/978-3-642-16321-0_20}, a popular compression scheme (especially for repetitive text collections).
\end{enumerate}
This allows deterministically processing a single reference string (for pattern matching or compression) and $d$ concurrently arriving text streams in $o(md)$ space, which was previously not possible.

\subsection{Suffix Random Access Data Structures}
In asymmetric streaming, we have read-only access to a \emph{reference} string $R$ of length $m$.
Then, a \emph{text} string $T$ of length $n$ arrives as a stream. 
Ideally, we would like to provide access to the last $\Theta(m)$ symbols of the text: this would immediately facilitate the design of asymmetric streaming
algorithms for an abundance of problems.
For example, if we want to locate exact occurrences of the reference in the text, then we indeed only need to access the last $m$ symbols to determine if $R$ occurs as a suffix of (the seen prefix of) $T$.
However, maintaining $m$ arbitrary symbols is impossible in $o(m)$ space.
To overcome this, we observe that, in many
applications, it suffices to have access to the longest suffix of the text that is ``similar'' to the reference string.
A trivial example is the case in which the alphabets of the reference and the text are disjoint: for exact pattern matching, it would then suffice to store the empty suffix.  

Formally, we introduce the notion of a \emph{$(k - 1)$-error suffix random access data structure}, parametrized by an integer $k\geq 1$, that supports random access to a suffix of~$T$ at least as long as the longest suffix of $T$
that can be decomposed into $k - 1$ single symbols and $k$ substrings of~$R$.

\paragraph*{Upper bounds for suffix random access.}
As a warm-up, we use read-only constant-space pattern matching~\cite{CROCHEMORE199233} to obtain a simple $0$-error suffix random access data structure that uses $\cO(\tau)$ space, has update time\footnote{Throughout this work, we use \emph{update time} for the time spent per arriving symbol of $T$ and \emph{query time} for the time required for randomly accessing symbols in the supported suffix of the text.} $\cO(m/\tau)$, and has constant query time (\cref{lm:random-warmup}).
Note that the product of the space and the update time is linear in $m$. This raises a natural question about $0$-error suffix random access data structures:

\begin{center}
\textit{Is there a data structure with sublinear product of space and update time\\that has (near-)constant query time?}
\end{center}

We answer this question affirmatively: for every constant $\epsilon>0$ and for every positive integer $\tau \geq m^{2/5+\epsilon}$, we design a data structure with space usage $\cOtilde(\tau)$ and a strongly sublinear product of space and update time.
Notably, for $\tau = \floor{\sqrt{m}}$, we achieve near-constant update time, and hence the product is $\tO(\sqrt{m})$.
Our main result for maintaining a suffix random access data structure in asymmetric streaming is as follows:

\def\versionofstreamingkerror{simplified}
\def\contentofstreamingkerror{%
There is an algorithm that, for any reference $R \in \Sigma^m$ and positive integer parameters~$\tau$ and $k = \tO(\tau)$, maintains a $(k - 1)$-error streaming suffix random access data structure with space complexity $\tO(\tau)$, worst-case update time $\tO(1 + k^3m^{3} / \tau^6)$, and query time $\cO(1 + \log k)$. The preprocessing succeeds without error in $\cO(\poly(m))$ expected time.
}

\begin{restatable*}[\versionofstreamingkerror]{theorem}{streamingkerror}
\label{cor:streaming_kerror}%
\contentofstreamingkerror
\end{restatable*}

\newcommand{\minipar}[1]{\noindent\emph{\underline{#1}}\ }
\minipar{A pattern-matching task.} A natural approach to designing a suffix random access data structure would be to
construct a text indexing data structure for $R$, which then allows locating an occurrence of any given query string in $R$.
By periodically locating recently received substrings of~$T$ within $R$, we can represent the maintained suffix of~$T$ as a list of fragments of $R$.

A textbook index is the \emph{suffix array} of~$R$, which requires $\cO(m)$ space and allows searching for any length-$\ell$
string in $\cO(\ell+\log m)$ time.
To decrease space usage, one can store the \emph{sparse suffix array} (see, e.g., \cite{DBLP:conf/latin/AyadLPV24} and references therein), 
which can efficiently locate every string of length at least $\ell$ and uses only roughly $\cO(m/\ell)$ space, for a parameter $\ell$ of our choice. 
However, this requires fixing the value of $\ell$, and it provides a small-space data
structure only when $\ell$ is large.

Interestingly, preprocessing a string of length $m$ in small space for detecting an occurrence
of a query string of fixed length $\ell = \Theta(\log m)$ has been studied in another context under the name of \emph{systematic substring 
search}~\cite{DBLP:journals/tcs/GalM07,DBLP:conf/tcc/Corrigan-GibbsK19}. In particular, Corrigan-Gibbs and
Kogan~\cite{DBLP:conf/tcc/Corrigan-GibbsK19} showed that systematic substring search is equivalent to a well-known problem in cryptography: function inversion. The idea of their reduction is to define a function that maps a position of~$R$ to a length-$\ell$ substring starting at that position.
By combining this reduction with the seminal data structure for function inversion devised by Fiat and Naor~\cite{doi:10.1137/S0097539795280512},
they achieve a trade-off $\tau^{3}q=\tO(m^{3})$, where~$\tau$ and~$q$ respectively denote the index size and the query time.
However, we would like to emphasize that systematic substring search (as well as the reduction to function inversion from \cite{DBLP:conf/tcc/Corrigan-GibbsK19}) only works for $\ell = \Theta(\log m)$, when the substrings are essentially integers from a polynomial domain.

The construction of a random access data structure may require computing occurrences of substrings of $T$ of arbitrary lengths (in $[1 \dd m]$) in $R$.
Hence, both existing approaches described above are inadequate for the task at hand.

\medskip

\minipar{Our approach.}
We obtain \cref{cor:streaming_kerror} via a reduction to function inversion.
From the discussion above, one might assume that the reduction is an adaptation of the one used for systematic substring search in \cite{DBLP:conf/tcc/Corrigan-GibbsK19}.
While our problem bears a conceptual resemblance to systematic substring search, the definition of systematic substring search is fairly restrictive, and their straightforward reduction is not applicable in our case. 
To nonetheless use function inversion, we present two novel and highly non-trivial ideas of independent interest:
\begin{enumerate}
\item We design a data structure that can be seen as a relaxation of text indexing on $R$. Given a query string $Q[1\dd 3q]$, it either locates an occurrence of the central fragment $Q(q\dd 2q]$ in $R$, or reports that the entire $Q$ has no occurrence in $R$; we call this query a \emph{core-matching query}.
\item We enrich \emph{string synchronizing sets}~\cite{DBLP:conf/stoc/KempaK19}, a tool for locally consistent sampling of substrings with numerous applications in stringology~(e.g.,~\cite{DBLP:conf/spire/Ellert23,DBLP:conf/focs/KempaK24,DBLP:conf/esa/Charalampopoulos21,DBLP:conf/soda/KempaK23,DBLP:journals/siamcomp/KociumakaRRW24,DBLP:conf/cpm/Charalampopoulos22,DBLP:conf/mfcs/Charalampopoulos25,DBLP:conf/spire/RadoszewskiZ24}), with a local sparsity condition. We then show a space-efficient algorithm for constructing such sets.\footnote{A different locally sparse variant of string synchronizing set has been proposed by Kempa and Kociumaka in \cite{DBLP:conf/stoc/KempaK22}; however the authors of this work do not provide a space-efficient construction algorithm.}
\end{enumerate}

By combining these ideas, we ultimately reduce the maintenance of a suffix random access data structure to inverting a function over a domain of size $\tO(m / \tau)$.
We next provide further details on each of these two ideas and clarify the challenges they address.

Our data structure for core-matching queries is of size $\tO(\tau)$ and supports queries for strings of the form $Q[1\dd 3q]$ with $q = 2^j \cdot \tau$, for integer parameters $j \geq 0$ and $\tau \geq 1$.
It is based on the following strategy.
Consider covering the reference $R$ with blocks of length $2q$, starting at positions that are multiples of $q$.
Then, every occurrence of $Q[1\dd 3q]$ in $R$ fully contains one of the blocks, and this block, in turn, contains the central fragment $Q(q\dd 2q]$ of the occurrence.
Therefore, to answer the query, it suffices to (attempt to) locate every length-$2q$ substring of $Q$ among the $\cO(m/q)$ blocks.
Hence, rather than considering every position of the reference as a potential occurrence, we only have to consider $\cO(m/q)$ positions.
In our core-matching solution using function inversion, this translates to reducing the domain of the function to $\cO(m/ q) \subseteq \cO(m / \tau)$.

However, by designing a core-matching data structure this way, we inadvertently create another challenge: to answer a query, in the worst case, we have to try to locate \emph{every} length-$2q$ substring of the query string among the blocks.
To avoid this, we use a variant of $q$-synchronizing sets~\cite{DBLP:conf/stoc/KempaK19}. 
Introduced by Kempa and Kociumaka~\cite{DBLP:conf/stoc/KempaK19}, they provide a consistent sampling mechanism for length-$2q$ substrings of $R$ and $Q$ such that, broadly speaking, within every length-$3q$ substring of $R$ and $Q$, at least one length-$2q$ substring is sampled.
Crucially, the total number of sampled substrings in $R$ is $\cO(m / \tau)$.
Hence, instead of the evenly-spaced blocks of length $2q$, we can use the sampled substrings.
Thus, rather than having to locate all length-$2q$ substrings of $Q$, we only have to locate the first sampled length-$2q$ substring of $Q$, accelerating queries by a factor~$q$.

To actually implement this idea, we crucially need to control the local sparsity of the sampling scheme, i.e., we not only need the total number of sampled substrings in $R$ to be $\cO(m / \tau)$, but we also need the number of samples within any length-$\cO(q)$ substring to be $\tO(1)$. (Otherwise, in the reduction to function inversion, there will be a blow-up in the size of the domain.)
Thus, we design a method for sampling a synchronizing set with this local sparsity property that, in addition, can be implemented efficiently in the streaming model, using $\tO(q)$ space.

\paragraph{Lower bounds for suffix random access.} We complement our main result with a reverse reduction from the function
inversion problem to the problem of maintaining a suffix random access data structure. This implies each of the following results (for a suffix random access data structure whose query time does not exceed its update time):%
\begin{enumerate}
\item An $\tO(\tau)$-space suffix random access data structure with update time faster than ours by a factor $\gg \tau$ would imply a breakthrough in function inversion, improving over Fiat and Naor's solution~\cite{doi:10.1137/S0097539795280512}. This gives a conditional update time lower bound of $\Omega(m^3/\tau^7)$ (\cref{lem:ra_offline_lower_cond}). 
\item An $\tO(\tau)$-space suffix random access data structure cannot have update time better than $\tO(m/\tau^{3})$, up to subpolynomial factors (\cref{lem:ra_streaming_lower}). We stress that this is an \emph{unconditional lower bound}.
\item An unconditional lower bound higher than ours by a $\tau$ factor would imply a breakthrough in the long-standing unconditional lower bound for function inversion~\cite{DBLP:conf/stoc/Yao90,DBLP:journals/eccc/DeTT09}. Hence, there is little hope of turning our conditional lower bound into an unconditional one (\cref{lem:nobetterlowercond}).
\end{enumerate}%
See \cref{fig:tradeoffs} for an illustration of our upper and lower bounds for suffix random access data structures.

\definecolor{cbBlue}{HTML}{0072B2}
\definecolor{cbOrange}{HTML}{E69F00}
\definecolor{cbSkyBlue}{HTML}{56B4E9}
\definecolor{cbGreen}{HTML}{009E73}
\definecolor{cbYellow}{HTML}{F0E442}
\definecolor{cbRed}{HTML}{D55E00}
\definecolor{cbPink}{HTML}{CC79A7}
\definecolor{cbBlack}{HTML}{000000}

\renewcommand{\floatpagefraction}{.8}

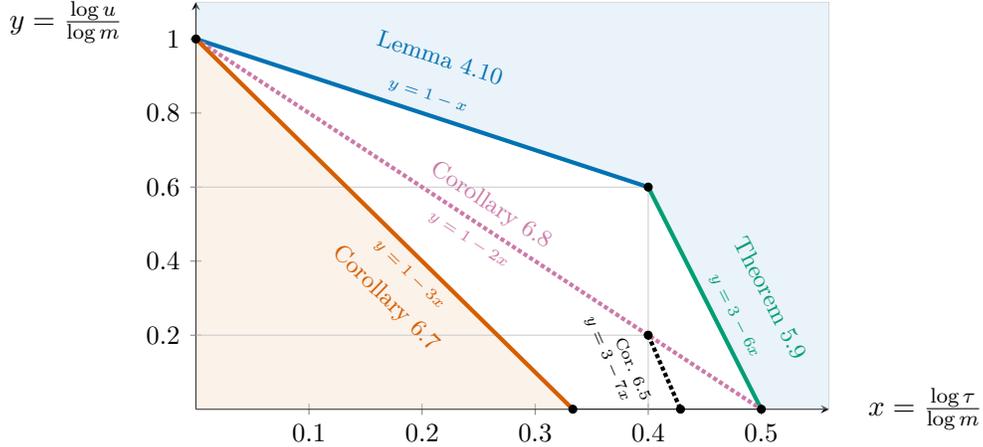
\begin{figure}
\centering
\begingroup
\NoHyper
\begin{tikzpicture}
    \begin{axis}[
        domain=0:0.5,
        ymin=0, ymax=1.1,
        xmin=0, xmax=0.56,
        axis lines=middle,
        xlabel={{$\quad x = \frac{\log \tau}{\log m}$}},
        xlabel style={
            at={(axis description cs:1,0)},
            anchor=west
        },
        ylabel={{$y = \frac{\log u}{\log m}$}},
		ylabel style={
		    at={(axis description cs:-0.1,.95)},
		    anchor=east
		},
    	xticklabel style={font=\small},
		yticklabel style={font=\small},
        samples=200,
        width=10cm,
        height=7cm,
        every axis plot/.append style={mark=none}
    ]
    
\fill[cbRed!20, opacity=0.4] 
        (axis cs:0,0) -- (axis cs:0,1) -- (axis cs:0.33333,0) -- cycle;
        
\fill[cbBlue!20, opacity=0.4] 
        (axis cs:0,1.1) -- (axis cs:0,1) -- (axis cs:0.4,0.6) -- (axis cs:0.5,0) -- (axis cs:0.56,0) -- (axis cs:0.56,1.1) -- cycle;
        
\draw[black!20] 
        (axis cs:0.4,0) -- (axis cs:0.4,0.6) -- (axis cs:0,0.6);
        
\draw[black!20] 
        (axis cs:0,1.1) -- (axis cs:0.56,1.1) -- (axis cs:0.56,0);
        
\draw[black!20] 
        (axis cs:0.4,0.2) -- (axis cs:0,0.2);
        
        \addplot[cbBlue, ultra thick, domain=0:0.4] {1 - x}
    	node[pos=0.5, sloped, above, align=center] {\footnotesize{\cref{lm:random-warmup}}\\\tiny{$y = 1 - x$}};

\addplot[cbGreen, ultra thick, domain=0.4:0.5] {3 - 6*x}
    node[pos=0.3, sloped, above right, align=center] {\footnotesize{\cref{cor:streaming_kerror}}\\\tiny{$y = 3 - 6x$}};

\addplot[cbBlack, ultra thick, densely dotted, domain=0.4:0.5] {3 - 7*x}
    node[pos=0.0875, sloped, below=.33em, align=center] {\shortstack[c]{\tiny{Cor.~6.5}\\\tiny{$y = 3 - 7x\quad\quad$}}};

\addplot[cbPink, ultra thick, densely dotted, domain=0:0.5] {1 - 2*x}
    node[pos=0.5, sloped, above, align=center] {\footnotesize{\cref{lem:nobetterlowercond}}}
    node[pos=0.5, sloped, below, align=center] {\tiny{$y = 1 - 2x$}};

\addplot[cbRed, ultra thick] {1 - 3*x}
    node[pos=0.4, sloped, below, align=center] {\tiny{$y = 1 - 3x$}\\\footnotesize{\cref{lem:ra_streaming_lower}}};
    
    \addplot[mark=*, only marks, mark size=1.5pt, black] coordinates {(0, 1)};
    \addplot[mark=*, only marks, mark size=1.5pt, black] coordinates {(0.5, 0)};
    \addplot[mark=*, only marks, mark size=1.5pt, black] coordinates {(0.4, 0.6)};
    \addplot[mark=*, only marks, mark size=1.5pt, black] coordinates {(0.4, 0.2)};
    \addplot[mark=*, only marks, mark size=1.5pt, black] coordinates {(0.333333333, 0)};
    \addplot[mark=*, only marks, mark size=1.5pt, black] coordinates {(.42857143, 0)};
    \end{axis}
\end{tikzpicture}
\endNoHyper
\endgroup
\caption{Trade-offs for suffix random access data structures in a doubly-logarithmic scale.
Each line segment corresponds to a relation of the exponents of the space $\tau$ and the update time $u$ of such a data structure as functions of $m$, that is, $\tau = m^x$ and $u = m^y$.
We consider only data structures where the query time is at most equal to the update time.
The {\color{cbBlue}blue} line segment corresponds to the upper bound of the warm-up data structure (see \cref{lm:random-warmup}).
The {\color{cbGreen}green} line segment corresponds to the upper bound of our main suffix access random data structure (see \cref{cor:streaming_kerror}).
The {\color{cbRed}orange} line segment corresponds to our unconditional lower bound (see \cref{lem:ra_streaming_lower}).
Any improvement of the lower bound beyond the {\color{cbPink}pink} dotted line segment would imply a breakthrough in the long-standing unconditional lower bound for function inversion \cite{DBLP:conf/stoc/Yao90,DBLP:journals/eccc/DeTT09} (see \cref{lem:nobetterlowercond}).
Finally, any improvement in the upper bound beyond the {\color{cbBlack}black} dotted line segment would yield an improvement over the Fiat--Naor function inversion data structure~\cite{doi:10.1137/S0097539795280512} (see \cref{lem:ra_offline_lower_cond}).}
\label{fig:tradeoffs}
\end{figure}

\subsection{Applications} 
We demonstrate that our suffix random data structure is a powerful tool by using it to obtain asymmetric streaming algorithms for pattern matching and compression.

\paragraph*{Generic pattern matching.}
We show that our suffix random data structure allows a unifying reduction for a large family of pattern matching problems.
Broadly speaking, if such a problem has an online algorithm in the read-only setting, then our data structure immediately yields an algorithm with similar guarantees in the asymmetric streaming model. 
In a pattern matching problem, we are given a pattern $P[1\dd m]$ and a text $T[1\dd n]$, and we are looking for (the ending positions of) fragments of~$T$, called \emph{occurrences}, that \emph{match} $P$ under some matching relation.
We introduce the following class of \emph{generic pattern matching} problems: 

\begin{restatable*}[Generic pattern matching]{definition}{defgenericpatternmatching}
For an integer $z \geq 0$, a pattern matching problem is \emph{$z$-generic} if it satisfies the following property. For a pattern $P[1\dd m]$ and a text $T[1\dd n]$, every occurrence of $P$ in $T$ is a fragment of $T$ that equals a concatenation of at most $z + 1$ substrings of~$P$ and at most $z$ symbols.
\end{restatable*}

The generic pattern matching framework captures several by now classical pattern matching problems.
To name just a few examples,
standard pattern matching, circular pattern matching~\cite{10.1007/978-3-319-02309-0_59,FREDRIKSSON2009579,10.1093/comjnl/bxt023,Iliopoulos2017,10.1007/978-3-642-25591-5_69,10.1007/978-3-642-38905-4_15,DBLP:journals/jcss/Charalampopoulos21,DBLP:conf/esa/Charalampopoulos22,keditCPM,DBLP:conf/cpm/BathieCS24}, a variant of elastic-degenerate string matching~\cite{grossi_et_al:LIPIcs.CPM.2017.9,aoyama_et_al:LIPIcs.CPM.2018.9,doi:10.1137/20M1368033,BERNARDINI2020109,DBLP:conf/biostec/ProchazkaCK021,ILIOPOULOS2021104616,DBLP:journals/mst/BernardiniGPSSZ24,pissis_et_al:LIPIcs.CPM.2025.28,gawrychowski_et_al:LIPIcs.CPM.2025.29,GABORY2025105296}, and string matching with variable-length gaps~\cite{Lee03,DBLP:journals/jcb/MorgantePVZ05,DBLP:conf/cocoon/RahmanILMS06,DBLP:journals/ir/FredrikssonG08,DBLP:conf/lata/FredrikssonG09,10.5555/643002.643006,BILLE201225,AMIR201534,
10.1007/978-3-642-20662-7_7,DBLP:journals/algorithmica/AmirKLPPS19} are all $\cO(1)$-generic.
Additionally, their $k$-approximate variants under each of the Hamming distance and the edit distance are $\cO(k)$-generic pattern matching problems.
In \cref{th:blackbox_ro_to_as}, we show that a suffix random access data structure can be used to transform any online read-only algorithm for a $k$-generic pattern matching problem into an asymmetric streaming algorithm at a small overhead cost. In particular, if one allows $\tau = \tO(\sqrt{mk})$ extra space, then an online read-only algorithm implies an asymmetric streaming algorithm without asymptotic time overhead. Despite its generality, our reduction is conceptually simple: we maintain a suffix random access structure and use it to simulate the online read-only algorithm through symbol queries. The definition of the class of generic pattern matching problems guarantees that all necessary symbols are available. 

\paragraph*{Relative Lempel--Ziv factorization.}
Relative Lempel--Ziv factorization is a compression scheme in which a text $T$ of length $n$ is greedily parsed into lengthwise maximal substrings of a reference~$R$ of length $m$. This method is especially effective for compressing collections of strings that are highly similar to a reference, and it has inspired a small-space text index for databases containing full genome sequences of individuals of the same species~\cite{10.1093/bioinformatics/btr505,10.1007/978-3-642-16321-0_20,DBLP:journals/bioinformatics/DeorowiczDG13,DBLP:journals/bmcgenomics/ValenzuelaNVPM18,DBLP:journals/bioinformatics/NavarroSMG19}.
Nevertheless, no small-space algorithm for computing this factorization is known. 

We address this gap by leveraging our suffix random access data structure to design the first asymmetric streaming algorithm for relative Lempel--Ziv factorization (\cref{thm:approx_rlz_ub}). If the size of the Lempel--Ziv factorization of $T$ relative to $R$ is~$z$, then our algorithm outputs a factorization of size $(1+\eps) \cdot z$, where each factor is either a symbol or a substring of $R$. The algorithm uses $\tO(\tau)$ space and $\tO(\eps^{-1} z (m/\tau)^{5/3})$ time (ignoring the time needed for maintaining the suffix random access data structure). 
The algorithm starts by constructing a factorization of size $\cO(\log^2 m) \cdot z$ using core-matching queries.
The factorization is then refined using ideas similar to those in a work of Fischer, Gagie, Gawrychowski, and Kociumaka~\cite{DBLP:conf/esa/0001GGK15}, who showed an efficient read-only algorithm for constructing the \emph{standard} Lempel--Ziv factorization of a text in small space. 
For both steps, we exploit random access on the relevant suffix of $T$.

On the lower bound side, we show that any $\tau$-space $\tO(1)$-approximation asymmetric streaming algorithm for relative Lempel--Ziv factorization must use $\tilde\Omega(z (m/\tau))$ time (\cref{cor:rlz_lb}). We obtain this lower bound via a reduction from (a version of) the set-complementation problem, for which we show a new lower bound using $R$-way branching programs.

\subsection{Discussion and Open Problems}
The main open problem that stems from our work is closing the gap between the obtained upper and lower bounds for suffix random access data structures.
Let us briefly discuss some of the main obstacles toward this goal.

The main challenge in improving our lower bound from \cref{lem:ra_streaming_lower} to match the $\tilde{\Omega}(m/\tau)$ space-time product of \cref{lem:nobetterlowercond} is the following difference in the nature of the two considered problems: a function inversion data structure becomes active just after construction, while a random access data structure can essentially afford to wait for $\tau$ symbols (as it can simply store them explicitly).
In our lower bound construction (see \cref{sec:overview}), we naturally encode a function $f$ into the reference and use the encoding of the image $f(x)$ of a domain element $x$ as the text.
If this text is shorter than $\tau$, then a (naive) random access data structure can explicitly store a copy.
Therefore, we enforce the text length to be significantly greater than $\tau$ by encoding $f(x)$ in $\gg \tau$ symbols, for each $x$ in the domain.
Consequently, the size of the domain is limited by $\cO(m / \tau)$.
While we can indeed show that processing the encoding of $f(x)$ as the text allows us to invert $f$ at position $f(x)$,
we only obtain a lower bound $\Omega(m/\tau^3)$ on the update time: we have to divide the $\Omega(m/\tau^2)$ lower bound for inverting a function with domain of size $m/\tau$ by the length of the text.

The main reason for the extra $\tau$ factor in the upper bound of \cref{cor:streaming_kerror} compared to \cref{lem:ra_offline_lower_cond} is that the function we invert (Karp--Rabin fingerprints) cannot be evaluated in constant time. Essentially, we can afford to wait for at most $\tau$ symbols to arrive, but then we intuitively have to perform a task of a pattern-matching flavor for a length-$\tau$ substring of $T$.
Using string synchronizing sets, we reduce this task to inverting a function over a domain of size roughly $m / \tau$; this can be done in $\tO(m^3/\tau^6)$ time using the Fiat--Naor function inversion data structure~\cite{doi:10.1137/S0097539795280512}.
We only need to invert once every $\cO(\tau)$ symbols arrive, so we need to pay for $\tO(m^3 / \tau^7)$ function evaluations per symbol (as part of the inversions).
However, evaluating Karp--Rabin fingerprints of fragments of $T$ requires $\Omega(\tau)$ time (given our space constraints).
We believe that overcoming this challenge will require a fundamentally new idea.

\section{Technical Overview of the Suffix Random Access Data Structure}
\label{sec:overview}
In this section, we give a technical overview of the main results of this work related to the suffix random access data structure. (Below, we often refer to it as random access data structure for brevity.) Intuitively, we are given a short read-only string~$R$ and a long streaming string $T$, and would like to maintain random access to the longest suffix of~$T$ that is close to a substring of~$R$ using little extra space. We study upper and lower bounds for random access data structures via bi-directional reductions to function inversion. 

We assume the word RAM model with words of width $w$, i.e., we can perform basic arithmetic operations on integers from $[0\dd 2^{\cO(w)}]$ in constant time. We further assume that the input strings are over a finite integer alphabet $\Sigma = [1\dd \sigma]$ with $\sigma \in n^{\cO(1)}$, and make the common assumption that $\log n = \cO(w)$. We develop classic offline read-only algorithms and asymmetric streaming algorithms, introduced by Andoni et al.~\cite{asymED} and Saks and Seshadhri~\cite{doi:10.1137/1.9781611973105.122}. In the former setting, one can access symbols of the input in constant time, and does not account for the space required to store the input in the space complexity of an algorithm. In the asymmetric streaming setting, the input consists of a read-only string $R \in \Sigma^\ast$, called a \emph{reference}, and a streaming string ${T \in \Sigma^\ast}$. Here, we receive~$T$ symbol by symbol and must account for the space required to store any information about~$T$ in the space complexity of an algorithm (whereas the space required to store~$R$ is not accounted for). 
Unless explicitly said otherwise, we measure the space complexity in words (including the working space used while answering queries). 


\subsection{Upper Bounds for Random Access Data Structures}\label{sec:tov_ub_ra}

We formalize the problem of maintaining a random access data structure as follows.

\begin{restatable}[$k$-Error Random Access Data Structure]{definition}{defdatastructure}\label{def:ra_data_structure}
For an integer $k \geq 0$, a string $R \in \Sigma^m$ (called the \emph{reference}), and a string~$T \in \Sigma^n$ (called the \emph{text}), a \emph{$k$-error random access data structure}
\begin{itemize}
	\item explicitly stores its \emph{support-length} $h \in [0\dd n]$, which satisfies the following property: if $h < n$, then $T[n - h \dd n]$ cannot be factorized into~$k$ symbols and $k + 1$ substrings of~$R$, and
	\item can return $T[i]$ upon receiving a \emph{random access query} $i \in (n-h\dd n]$, possibly accessing the reference $R$, but not accessing the text~$T$.
\end{itemize}
\end{restatable}


We introduce different algorithms for constructing such data structures in the read-only model and in the asymmetric streaming model.
In both models, we assume that the data structure is defined by a given parameter $\tau$.
Also in both models, we have read-only access to the reference, and we may preprocess it into a data structure of size $\tO(\tau)$.
Then, in the read-only model, the text is given as a read-only string, and we have to construct the random access data structure using the precomputed information and random access on $R$ and $T$. We call this version of the data structure \emph{offline}. 
In the asymmetric streaming model, we instead receive the text as a stream. Whenever we receive some symbol $T[i]$, we must update the maintained data structure such that it is a random access data structure for the reference $R$ and a text $T[1\dd i]$. We call this version of the data structure \emph{streaming}.
In either model, the total space occupied at any time after preprocessing (apart from $R$) may not exceed $\tO(\tau)$.
To summarize, apart from the $\tO(\tau)$ space complexity of the data structure, we are interested in the preprocessing time (spent processing~$R$ before the first access to~$T$), the query time for random access queries, and either the construction time after preprocessing (for the read-only model) or the update time per arriving symbol (for the asymmetric streaming model).


\paragraph{Reduction to the $0$-error offline setting.} 
We start by showing that an efficient offline construction algorithm immediately implies an efficient streaming algorithm. This is done by cutting~$T$ into blocks of geometrically increasing sizes on $\cO(\log m)$ levels, and then constructing the offline data structure for each block. The cost of running the offline construction algorithm can be de-amortized over the arriving symbols.
This is similar to the black box offline to online reduction for approximate pattern matching, see \cite{DBLP:journals/iandc/CliffordEPP11}.
Unlike~\cite{DBLP:journals/iandc/CliffordEPP11}, we cannot afford $\cO(m)$ additional space, and we have to ensure that we already support random access to a block while constructing its data structure.
While constructing the offline data structure for a block of size $2^\ell$, we use the random access data structures already constructed for the two blocks of size $2^{\ell - 1}$ that make up the larger block. 
For the offline construction, we further observe that, instead of constructing a $k$-error data structure, we can chain together $\cO(k)$ $0$-error data structures. A random access query can then be answered by binary searching in $\cO(\log k)$ time for the data structure responsible for the query.

\begin{restatable*}{corollary}{lemreducekerrorstreamingtonoerroroffline}\label{lem:random_access_reduce_kerrorstreaming_to_noerroroffline}
    Assume that there is an offline $0$-error random access data structure with preprocessing time $\cO(p(m))$, construction time $\cO(n \cdot c(m))$, query time $\cO(q(m))$, and space complexity $\cO(s(m))$.
    For every positive integer $k$, there is a streaming $k$-error random access data structure with preprocessing time $\cO(p(m))$, worst-case update time $\cO(c(m) \cdot (q(m) + \log k) \cdot \log m)$, query time $\cO(q(m) +\log k)$, and space complexity $\cO(k \cdot s(m) \cdot \log m)$. 
\end{restatable*}

\paragraph{Reduction to core-matching queries.} 
In the $0$-error offline setting, we further isolate the main computational challenge by showing that a simple data structure of size $\tO(\tau)$ with constant query time can be constructed by answering a series of $\tO(1)$ so-called \emph{core-matching queries} on~$R$. Each query consists of a string $T'[1\dd 3n']$ (a~substring of $T$) with $n' = 2^\ell \cdot \tau$ for some integer $\ell \geq 0$, and the result of the query is some position~$i$ such that $R[i\dd i + n') = T'(n'\dd 2n']$, i.e., we must locate the central length-$n'$ fragment of $T'$ in~$R$. We are allowed to fail if $T'$ does not occur in~$R$.

To understand the idea of the reduction, consider a particularly simple case when~$T$ is of length $n = 3\cdot 2^L \cdot \tau$ and the entire~$T$ is a substring of~$R$. In this case, the random access data structure consists of a copy of the length-$\tau$ prefix and suffix of $T$, as well as pointers to $\cO(\log n)$ fragments $r_1, \dots, r_{\cO(\log n)}$ of~$R$ obtained by answering core-matching queries for all prefixes and suffixes of~$T$ that are of length $3 \cdot 2^\ell \cdot \tau$ for integer $\ell \geq 0$. 
This is visualized below, where
\raisebox{-.1em}{\begin{tikzpicture}[every node/.style={inner sep=0pt}]
    \node[minimum width=1.1em, minimum height=.9em, draw, pattern=north east lines] {};
\end{tikzpicture}}
indicates the length-$\tau$ prefix and suffix of $T$, and\enskip
\raisebox{-.25em}{\begin{tikzpicture}[every node/.style={inner sep=0pt}]
    \node[minimum width=3.6em, minimum height=.6em] (box) {};
    \draw[thick] (box.north west) to (box.south west) (box.north east) to (box.south east) (box.west) to (box.east);
    \node[minimum width=1.2em, minimum height=1em, draw, fill=black!10!white] (core) at (box) {};
    \node at (core) {$r_i$};
\end{tikzpicture}}\enskip
indicates the core-matching query used to find $r_i$.

\medskip

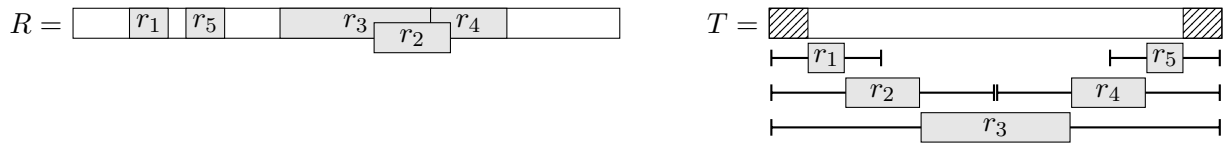
\begin{center}%
\newlength{\tikzxx}%
\setlength{\tikzxx}{.65em}%
\newlength{\tikzyy}%
\setlength{\tikzyy}{1.2em}%
\begin{tikzpicture}[x = \tikzxx, y = 1em, every node/.style={inner sep=0pt}]
    \foreach[count=\ri from 0, count=\ti from -37] \i in {0,...,60} {
        \node[minimum width=\tikzxx, minimum height=1em] (r\ri) at (\i, 0) {};
        \node[minimum width=\tikzxx, minimum height=1em] (t\ti) at (\i, 0) {};
        \node[minimum width=\tikzxx, minimum height=1em] (lowr\ri) at (\i, -.5em) {};
    }
    \node[draw, fit=(r0)(r28)] (R) {};
    \node[draw, fit=(t0)(t23)] (T) {};
    \node[left=0 of R] {$R =\ $};
    \node[left=0 of T] {$T =\ $};
    
    \node[draw, fit=(t0)(t1), pattern=north east lines] {};
    \node[draw, fit=(t22)(t23), pattern=north east lines] {};

    \foreach[count=\q from 1, evaluate=\sideid as \sidemult using int(1-(\sideid*2))] \yoff/\shortlen/\sideid in {%
        1/2/0,%
        2/4/0,%
        3/8/0,%
        2/4/1,%
        1/2/1} {
            \node (v) at (0, -\yoff\tikzyy) {};
            \foreach[
                count=\i from 0, 
                evaluate=\x as \nid using int(\sideid*23 + \sidemult * \i * \shortlen),
                evaluate=\x as \xp using int(\x + 1),
                evaluate=\nid as \nidp using int(\nid + \sidemult * (\shortlen - 1))] \x in {0, 2, 4} {
                \node[minimum width=\tikzxx, minimum height=1em] (n\x) at (t\nid |- v) {};
                \node[minimum width=\tikzxx, minimum height=1em] (n\xp) at (t\nidp |- v) {};
            }
            \node[fit=(n0)(n5), inner ysep = -.25em, inner xsep=-1pt] (box) {};
            \draw[thick] (box.north west) to (box.south west) (box.north east) to (box.south east) (box.west) to (box.east);
            \node[fit=(n2)(n3), inner ysep = -.01em, inner xsep=-.5pt, draw, fill=black!10!white] (core) {};
            \node at (core) {$r_\q$};
        }

    \foreach[evaluate=\pos as \epos using int(\pos+\len-1)] \len/\pos/\q in {2/3/1,8/11/3,4/19/4,2/6/5} {
        \node[draw, fit=(r\pos)(r\epos), fill=black!10!white] (tmp) {};
        \node at (tmp) {$r_\q$};
    }
    \node[draw, fit=(lowr16)(lowr19), fill=black!10!white] (tmp) {};
    \node at (tmp) {$r_2$};
\end{tikzpicture}
\end{center}

\medskip

With simple techniques to generalize for a string~$T$ of arbitrary length that is not a substring of $R$, this leads to the following result.

\def\versionoflemreduceratocorematching{simplified}
\def\addtolemreduceratocorematching{}
\def\addtolemreduceratocorematchingtext{}
\begin{restatable*}[\versionoflemreduceratocorematching]{lemma}{lemreduceratocorematching}\label{lem:reduce_offline_ra_to_core_matching}
    Let $\tau > 0$ be an integer parameter. Computing an offline $0$-error random access data structure with constant query time and space complexity $\cO(\tau + \log m)$ for reference $R \in \Sigma^m$ \addtolemreduceratocorematchingtext{}can be reduced to $\cO(\log^2 m)$ core-matching queries on~$R$.\addtolemreduceratocorematching{}
The reduction takes $\cO(\tau +  \log^2 m)$ time and $\cO(\tau + \log m)$ additional working space.
\end{restatable*}
\def\addtolemreduceratocorematching{
For every ${\ell \in [0\dd \floor{\log_2(m/3\tau)}]}$ there are $\cO(\log m)$ queries of length $3 \cdot 2^\ell \cdot \tau$, and each query is (a~reference~to) a substring of~$T$. 
The reduction algorithm is adaptive, i.e., it waits for each query to be answered prior to asking the next one.}
\def\addtolemreduceratocorematchingtext{and text $T \in \Sigma^n$\ }
\def\versionoflemreduceratocorematching{full version}

\paragraph{Core-matching queries via function inversion.} 
As a warm-up, in \cref{lem:naive_core_matching} we show a simple solution to core-matching queries based on the following observation: If $T' = R[i\dd i+3n')$ for some $i \in [1\dd m - 3n' + 1]$, then this occurrence of $T'$ fully contains the fragment $R(bn' - 2n' \dd bn']$ with $b = \floor{\frac{i+3n'-1}{n'}}$. We precompute the Karp--Rabin fingerprints~\cite{karp1987efficient} $\phi(U)$ of such substrings $U$ in $\cO(m/n')$ space and store them in a balanced binary search tree. At query time, we compute the fingerprints of all length-$2n'$ substrings of $T'$, and use them to find~$j$ such that $T'[j \dd j + 2n') = R(bn' - 2n'\dd bn']$. We can then simply output $R(bn' - n' - j + 1\dd bn' - j + 1] = T(n' \dd 2n']$. Combined with the reductions above, this already gives a $k$-error random access data structure with a non-trivial trade-off:  

\def\versionoflemnaivekerror{simplified}
\def\contentoflemnaivekerror{%
There is a streaming $k$-error random access data structure with space complexity $\tO((k + 1) \cdot \sqrt{m})$, worst-case update time $\tO(1)$, and query time $\cO(1 + {\log (k + 1)})$. The preprocessing takes $\cO(\poly(m))$ time and $\tO(\sqrt{m})$ space deterministically.%
}

\begin{restatable*}[\versionoflemnaivekerror]{corollary}{cornaivekerror}
\label{cor:naivekerror}%
\contentoflemnaivekerror
\end{restatable*}

\def\contentoflemnaivekerror{%
There is streaming $k$-error random access data structure with space complexity $\tO((k + 1) \cdot \sqrt{m})$, worst-case update time $\tO(1)$, and query time $\cO(1 + \log (k + 1))$. The preprocessing takes $\cO(\poly(m))$ time and $\tO(\sqrt{m})$ space deterministically, or $\tO(m + \poly(\err))$ time and $\tO(\sqrt{m} + \poly(\err))$ space with success probability at least $1 - 2^{-\err}$, for any integer parameter $\err \geq \log_2 (m\sigma)$.%
}
\def\versionoflemnaivekerror{full version}

In a pursuit of a better trade-off, we replace the binary search tree with a function inversion data structure.
The idea is to use a function $f(x) = \phi(R(xn'-2n'\dd xn'])$ with domain ${[N]}$, where $N = \floor{m/n'} - 1$. 
Rather than locating the fingerprint~$y$ of every length-$2n'$ substring $T'[j\dd j + 2n')$ of $T'$ in the binary search tree, 
we instead invert the function and obtain $x \in f^{-1}(y)$, which implies $T'[j\dd j + 2n) = R(xn'-2n'\dd xn']$.
The best-known trade-off for function inversion is due to Fiat and Naor~\cite{doi:10.1137/S0097539795280512}, who showed that $\tO(N^3 / \tau^3)$ inversion time can be achieved in $\tO(\tau)$ space:

\begin{restatable*}[\cite{doi:10.1137/S0097539795280512}]{fact}{inversion}\label{fact:inversion}
Let $f : [N] \rightarrow [N]$ be a function that can be
evaluated at any point in constant time. For any integer parameter $\tau > 0$, we can construct in $\tO(N)$ time and with success probability $1- 1/N$ a data structure that is capable of inverting~$f$ at any point in $\tO(N^3/\tau^3)$ time. The data structure can be constructed, stored, and queried in $\tO(\tau)$ space.
\end{restatable*}

For core-matching queries of length $3n' = 3\tau$, we can achieve $\tO(\tau)$ space and $\tO(m^3/\tau^4)$ time: the domain of~$f$ is of size $N = \cO(m/\tau)$, but we have to multiply the $\tO(N^3 / \tau^3)$ inversion time of \cref{fact:inversion} with the number of inversion queries and the evaluation time of $f$, both of which are roughly $\tau$. 
To reduce the time, we use a variant of the string synchronizing sets \cite{DBLP:conf/stoc/KempaK19} to consistently and regularly select synchronizing positions in both~$R$ and~$T$.
This allows us to obtain a better function~$f$ that retains the size of the domain and the evaluation time, but only needs to be inverted (close to) a constant number of times in order to answer a core-matching query.
In the definition below, we use the original consistency and density conditions from \cite{DBLP:conf/stoc/KempaK19}, but also add a sparsity condition, ensuring that, within any length-$\tau$ fragment of $R$, we select a small number of positions.
For a string $T$, we denote by $\per(T)$ the smallest period of $T$, that is, the smallest integer $p$ such that $T[i] = T[i+p]$ for all $i \in [1 \dd n-p]$.

\begin{restatable*}{definition}{syncsets}
    Let $k,\tau$ be positive integers. For a string $R \in \Sigma^m$, a set $S \subseteq [1\dd m - 2\tau + 1]$ is a $k$-sparse $\tau$-synchronizing set if it satisfies the following conditions:
    \begin{description}
        \item[(Consistency)] For $i,j \in [1\dd m - 2\tau + 1]$, if $R[i\dd i + 2\tau) = R[j\dd j+2\tau)$ then $i \in S \iff j \in S$.
        \item[(Density)] 
        For $i \in [1\dd m-3\tau + 2]$, it holds $S \cap [i\dd i + \tau) = \emptyset$ if and only if $\per(R[i\dd i + 3\tau - 2]) \leq \frac\tau3$.
        \item[(Sparsity)] For $i \in [1\dd m-3\tau + 2]$, it holds $\absolute{S \cap [i\dd i + \tau)} \leq k$.
    \end{description}
\end{restatable*}

During construction, we enforce the sparsity condition by sampling every distinct substring with probability around $k/\tau$. Namely, we define a function $\ident : \Sigma^{2\tau} \rightarrow \{0,1\}$ that identifies synchronizing substrings. Crucially, the function can be stored in $\tO(1)$ space, and can be evaluated efficiently in a rolling manner, similarly to Karp--Rabin fingerprints. This way, we can produce the elements of a string synchronizing set as a stream. 

\begin{restatable*}{theorem}{sss}\label{lm:sss}
    Let $R \in \Sigma^m$ be a string, and let $\tau \in [1\dd \floor{\frac{m}2}]$ and $\err \geq \log_2 (m\sigma)$ be integers. In $\cO(\poly(\err))$ time and space, and with success probability $1 - 2^{-\err}$, we can construct a function $\ident : \Sigma^{2\tau} \rightarrow \{0,1\}$ encoded in $\cO(\err)$ space with the following properties: 
    \begin{enumerate}
        \item $S = \{ i \in [1\dd m - 2\tau + 1] \mid \ident(R[i\dd i + 2\tau)) = 1 \}$ is an $\cO(\err)$-sparse $\tau$-synchronizing set.
        \item Given $R' \in \Sigma^{m'}$ of arbitrary length, the set $\{ i \in [1\dd m' - 2\tau + 1] \mid \ident(R'[i\dd i + 2\tau)) = 1 \}$ can be produced as a stream in increasing order, in $\cO(m'\err)$ time and $\cO(\err)$ working space.
    \end{enumerate}
\end{restatable*}

We are now ready to explain our final solution for core-matching queries. To solve a query $T'[1 \dd 3n']$, we use an $\tO(1)$-sparse $n'$-synchronizing set~$S$ of~$R$. We find the minimal $j \in [1 \dd n' + 1]$ such that $\ident(T'[j\dd j + 2n')) = 1$, which takes $\cO(n')$ time with the algorithm from \cref{lm:sss}(2). For intuition, assume that~$j$ is well-defined. (If~$j$ does not exist, then $T'$ is periodic and the solution only becomes simpler.) We compute $y = \phi(T'[j\dd j + 2n'))$ and use \cref{fact:inversion} to obtain some $x \in f^{-1}(y)$, where this time $f(x) = \phi(R[x\dd x+2n'))$ is defined not for domain $[\floor{\frac m{n'}} - 1]$ but for domain $S$, which is still of size $\tO(\frac m{n'})$. While \cref{fact:inversion} requires that the domain is an integer range, the sparsity of~$S$ allows us to obtain an efficient mapping between $[\absolute{S}]$ and~$S$.
Hence we can still invert~$f$, and, if we find $x \in f^{-1}(y)$, then it holds $T'[j\dd j + 2n) = R[x\dd x + 2n')$.

For core-matching queries of length $3n' = 3\tau$, using synchronizing sets improves the time complexity to $\tO(m^3 / \tau^5)$. 
The full construction of the function~$f$ and the reduction of core-matching to function inversion is given in \cref{lem:reduce_core_matching_to_fi}. 
We plug this solution for core-matching queries into \cref{lem:reduce_offline_ra_to_core_matching} to obtain a $0$-error offline data structure, and then into \cref{lem:random_access_reduce_kerrorstreaming_to_noerroroffline} (with a slight modification that improves the dependency on $k$) to obtain our main result:

\streamingkerror*

\subsection{Lower Bound for Random Access Data Structures}
In computational complexity, a fundamental difference arises between \emph{uniform} and \emph{non-uniform} computation. 
In the former case, we assume that a single algorithm is used to solve input instances of all possible sizes (i.e., the size is part of the input).
In the latter case, we are allowed to design a family of algorithms, using a different algorithm for each input size. 
We say \emph{non-uniform algorithm} to refer to the entire family, and \emph{non-uniform data structure} whenever a data structure is constructed, updated, or queried using non-uniform algorithms.
Crucially, a non-uniform algorithm may embed an arbitrary amount of \emph{advice} that depends only on the size of the input.

To show a lower bound for random access data structures, we show how to obtain a non-uniform function inversion protocol from any offline random access data structure. Assume that the text~$T$ is a substring of the reference~$R$. Intuitively, the construction algorithm for a random access data structure must find an occurrence of~$T$ in~$R$. Using this observation, we show that a random access data structure implies a function inversion protocol: we encode the image of a function in a reference $R$, and to invert an element of the image, we encode and feed it as a text~$T$. As a result, we obtain the following:

\begin{restatable*}{corollary}{randomaccessimpliesnonuniforminversion}\label{lem:random_access_implies_non_uniform_inversion}
\let\tau\tauconst
Let $\tau, \delta > 0$ with $\tau < 1$ be constant.
Assume that there is a (possibly non-uniform) offline random access data structure with 
space complexity $\cO(m^\tau)$ bits, amortized query time $\tO(m^\delta)$, and construction time $\tO(m^{\tau + \delta})$ for a text of length at most $m^\tau \log m$.
Then there is a non-uniform function inversion data structure that, for functions with domain and co-domain~$[N]$, has query time $\tO(N^{(\tau + \delta) / (1- \tau)})$ and can be stored and queried in $\tO(N^{{\tau}/{(1 - \tau)}})$ space.
\end{restatable*} 

This reduction directly yields a conditional lower bound for random access data structures:

\begin{restatable*}{corollary}{raofflinelowercond}\label{lem:ra_offline_lower_cond}
\let\tau\tauconst
Let $\tau, \delta > 0$ with $2/5 \le \tau < 1$ be constant.
Assume that there is a (possibly non-uniform) offline random access data structure with 
space complexity $\cO(m^\tau)$ bits, amortized query time $\tO(m^\delta)$, and construction time $\tO(m^{\tau + \delta})$ for a text of length at most $m^\tau \log m$.
If $\delta < 3 - 7\tau$, then there is a non-uniform function inversion data structure that improves over the time-space trade-off of Fiat and Naor.
\end{restatable*}

Using well-known lower bounds on the time-space trade-off of permutation inversion~\cite{DBLP:conf/stoc/Yao90,DBLP:journals/eccc/DeTT09}, we also derive the following unconditional lower bound:

\begin{restatable*}{corollary}{raofflinelower}\label{lem:ra_offline_lower}
\let\tau\tauconst
Let $\tau, \delta > 0$ with $\tau < 1$ be constant.
Assume that there is a (possibly non-uniform) offline random access data structure with 
space complexity $\cO(m^\tau)$ bits, amortized query time $\tO(m^\delta)$, and construction time $\tO(m^{\tau + \delta})$ for a text of length at most $m^\tau \log m$.
Then $\delta \geq 1 - 3\tau$.
\end{restatable*}

In \cref{lem:ra_streaming_lower}, we extend both the conditional and the unconditional lower bounds to a streaming text. 
Finally, via our reduction from function inversion to suffix random access, we show that an unconditional lower bound higher than ours by a factor of $\tau$ would imply a breakthrough in the long-standing unconditional lower bound for function inversion~\cite{DBLP:conf/stoc/Yao90,DBLP:journals/eccc/DeTT09} (\cref{lem:nobetterlowercond}).

\section{Preliminaries}
\label{sec:prelim}
For integers $i,j\in \mathbb{Z}$, we write $[i\dd j] = [i\dd j + 1) = (i - 1\dd j] = (i - 1\dd j + 1)$ to denote the integer range
$\{k \in \mathbb{Z} : i \le k \le j\}$, and $[i]$ to denote $[1\dd i]$.
For a set $X\subseteq \mathbb{Z}$ and an integer $s \in \mathbb{Z}$, we write
$s+X$ and $X+s$ for the set $\{s+x : x\in X\}$ of all elements of~$X$ incremented by~$s$.

An \emph{alphabet} $\Sigma$ is a finite set, and its members are called \emph{symbols}.
A length-$n$ string~$T$ over alphabet $\Sigma$ is a sequence of~$n$ symbols.
We write $T[1\dd n]$ to denote that~$T$ is a string of length $\absolute{T} = n$.
For $i, j \in [1\dd n]$, we write $T[i]$ to denote the $i$-th symbol in $T$, 
and $T[i\dd j] = T[i\dd j + 1) = T(i - 1\dd j] = T(i - 1\dd j + 1)$ to denote the sequence $T[i]T[i + 1] \dots T[j]$ if $i \leq j$, and the \emph{empty string} $\varepsilon$ otherwise.
We say that another string $P[1\dd m]$ is a \emph{substring} of $T$, or equivalently that it \emph{occurs} in $T$, if there is $i \in [1 \dd n-m+1]$ such that $\forall j \in [0 \dd m) : T[i + j] = P[1 + j]$.
Any such fragment $T[i \dd i+m)$ is an \emph{occurrence} of~$P$ in~$T$.
We sometimes refer to the position $i$ itself as an occurrence, too. 
For $i \in [1\dd n]$ and $j \in [i\dd n]$, we differentiate between the substring $T[i \dd j]$ and the \emph{fragment} $T[i \dd j]$, where in the latter case we mean the specific occurrence of $T[i \dd j]$ at position~$i$.
\emph{Suffixes} and \emph{prefixes} of~$T$ respectively have the form $T[i\dd n]$ and $T[1\dd i]$, and we may simply write $T[i\dd ]$ and $T[\dd i]$ instead.
The \emph{reverse} of~$T$ is $\rev(T) = T[n]T[n-1] \cdots T[2]T[1]$.

For strings $S[1\dd m],T[1\dd n]$, their \emph{concatenation} $S[1]S[2] \cdots S[m]\, T[1]T[2] \cdots T[n]$ is denoted by $S \cdot T$ or simply $ST$.
We denote the concatenation $S_1 \cdot S_2 \cdots S_k$ of strings $S_1, S_2, \ldots, S_k$ by $\bigodot_{i=1}^k S_i$.
The concatenation of $k\in \mathbb{Z}_{>0}$ copies of a string~$S$ is denoted by $S^k$, and $S^k$ is called a \emph{power} of~$S$ with \emph{exponent} $k$, while the infinite string obtained by concatenating infinitely many copies of~$S$ is denoted by $S^\infty$.

A string~$T$ is \emph{primitive} if it cannot be expressed as $S^k$
for any string~$S$ and any integer~$k > 1$.
The \emph{rotation} operation $\rot(\star)$ takes a string as input and moves its last symbol to the
front; that is,~$\rot(T) := T[|T|]T[1 \dd |T|-1]$.
We say that a string $\rot^i(T)$, where $i \in \mathbb{Z}$, is a \emph{rotation} of~$T$.
Note that a primitive string does not match any of its non-trivial rotations, i.e., $T = \rot^i(T)$ if and only if $i \equiv 0 \pmod{n}$.
A positive integer~$p$ is a \emph{period} of~$T$ if $T[i] = T[i+p]$ for all $i \in [1 \dd n-p]$.
The smallest period of~$T$ is referred to as \emph{the period} of~$T$ and is denoted by $\per(T)$. A string~$T$ is called \emph{periodic} if $\per(T) \leq \absolute{T}/2$. 

\begin{fact}[Corollary of the Fine--Wilf periodicity lemma~\cite{fine1965uniqueness}]\label{cor:progression}
The starting positions of the occurrences of a pattern~$P$ in a text~$T$ form $\cO(|T|/|P|)$ arithmetic progressions with difference $\per(P)$.
\end{fact}

\begin{fact}[{\cite[Section~5]{DBLP:journals/jcss/FredmanW93}}]\label{lem:computelog}
    Let~$n$ be a positive integer. After an $\cO(\log n)$ time and space preprocessing, computing $\ceil{\log_2 x}$ or $\floor{\log_2 x}$ for any $x \in [1\dd n]$ takes $\cO(1)$ time.
\end{fact}

\subsection{Karp--Rabin Fingerprinting}
\begin{definition}[Karp--Rabin fingerprint~\cite{karp1987efficient,porat2009optimal}]
For a prime number~$p$ and an integer $r \in [0\dd p)$, the Karp--Rabin fingerprint of $X \in \Sigma^n$ is $\varphi_{p,r}(X) = {\sum_{i = 1}^n X[i] \cdot r^{n - i} \pmod p}$.
\end{definition}

\begin{fact}\label{obs:linearity_fingerprints}
Let $X, Y$ be strings of length at most~$n$, let $p$ be a prime number, and let $r$ be an integer in $[0\dd p)$.
Given two of the fingerprints $\varphi_{p,r}(X), \varphi_{p,r}(Y), \varphi_{p,r}(XY)$, the remaining one can be computed in $\cO(\log n)$ time and $\cO(1)$ space.
\end{fact}

For example, we can compute $\varphi_{p,r}(XY)$ as $(\varphi_{p,r}(X) \cdot r^{\absolute{Y}} + \varphi_{p,r}(Y)) \bmod p$, where $r^{\absolute{Y}}$ can be computed in $\cO(\log \absolute{Y})$ time using exponentiation by squaring.

By precomputing $r^\ell$ in $\cO(\log \ell)$ time, we can compute fingerprints of length-$\ell$ substrings in a \emph{rolling} manner.  

\begin{fact}\label{obs:rollfp}
Let $p$ be a prime number, let $r$ be an integer in $[0 \dd p)$, and let $\ell$ be a positive integer.
After an $\cO(\log \ell)$-time and $\cO(1)$-space preprocessing, the following holds. For string $X[1\dd n]$, consider $i \in [1\dd n - \ell]$. If $\phi_{p,r}(X[i\dd i + \ell))$ is given, then 
$\phi_{p,r}(X[i\dd i + \ell])$ and $\phi_{p,r}(X(i\dd i + \ell])$ can be computed in $\cO(1)$ time.
\end{fact}

If distinct length-$n$ strings $X_1,X_2$ have the same fingerprint, then~$r$ is a root of the polynomial $\varphi_{p,x}(X_1) - \varphi_{p,x}(X_2)$ of degree at most $n - 1$ over the field $\mathbb{F}_p$.
There are at most $n - 1$ roots, and hence using large~$p$ and choosing~$r$ at random implies that distinct strings have distinct fingerprints with high probability (w.h.p.).

\begin{fact}[see, e.g., \cite{porat2009optimal}]\label{fact:collision_probability}
Let~$p>\sigma$ be a prime number and let $r \in [0\dd p)$ be drawn uniformly at random. For any two strings $X_1, X_2 \in [0\dd\sigma)^n$, the following holds:
\begin{itemize}
\item If $X_1 = X_2$, then $\varphi_{p,r}(X_1) = \varphi_{p,r}(X_2)$;
\item If $X_1 \neq X_2$, then $\varphi_{p,r}(X_1) = \varphi_{p,r}(X_2)$ with probability at most $\frac{n-1}{p}$.
\end{itemize}
\end{fact}

The following fact follows from a direct application of the union bound.

\begin{fact}\label{fact:collision_probability:allsubstrings}
Let~$p > \sigma$ be a prime number, let $r \in [0\dd p)$ be drawn uniformly at random, and let $\ell \in [n]$.
Any string $X \in [0\dd\sigma)^n$ has two distinct length-$\ell$ substrings $X_1,X_2$ with $\varphi_{p,r}(X_1) = \varphi_{p,r}(X_2)$ with probability at most $\frac{\ell \cdot (n - \ell + 1)^2}{p} \leq \frac{n^3}{p}$.
\end{fact}

In the remainder of the paper, we use Karp--Rabin fingerprints with a large prime number~$p$ and a random base~$r$.
We often omit the subscript $p,r$ and simply write $\varphi(X)$ instead. We then provide more details on the choice of~$p$ when analyzing the success probability of our algorithms.
Either way, we use the following result to find a prime in a given range.

\begin{fact}\label{fact:findprime_in_range}
    Let $\err, n > 1$ be integers. A prime $p \in [n\dd 2n)$ can be found in $\cO(\err \cdot \polylog(n))$ time and $\cO(\polylog(n))$ space with success probability at least $1 - 2^{-\err}$, and in $\cO(n \cdot \polylog(n))$ time and $\cO(\polylog(n))$ space deterministically.
\end{fact}

\begin{proof}
    The AKS primality test verifies if a given number~$p$ is prime in $\cO(\polylog(p))$ time~\cite{AKS}. The deterministic algorithm applies the test to all $p \in [n\dd 2n)$ until a prime is found. (At least one such prime exists due to the Bertrand–Chebyshev theorem.) The probabilistic algorithm exploits that, once~$n$ exceeds some constant, $[n\dd2n)$ contains at least $n / (c  \ln n)$ primes for another constant $c > 0$ (which is a consequence of the prime number theorem). If we try $\ceil{\err  \log^2 n} > c  \err  \ln n$ values $p \in [n\dd 2n)$ drawn uniformly at random, then the verification takes $\cO(\err \cdot \polylog(n))$ time, and the probability that none of the tried values is prime is $(1 - 1 /(c \ln n))^{(c \cdot \ln n) \cdot \err} < e^{-\err} < 2^{-\err}$.
\end{proof}

\begin{fact}\label{fact:construct_deterministic_collision_free_fp}
    For any string $R \in [0\dd \sigma)^m$, we can find a prime $p \in \Theta(\sigma + m^4)$ and $r \in [0\dd p)$ such that, for every $\ell \in [1\dd m]$, there are no distinct length-$\ell$ substrings $X_1,X_2$ of~$R$ such that $\varphi_{p,r}(X_1) = \varphi_{p,r}(X_2)$ in $\cO(\poly(m + \sigma))$ time and $\cO(\polylog(m + \sigma))$ space deterministically.
\end{fact}

\begin{proof}
    If distinct length-$\ell$ strings $X_1,X_2$ have the same fingerprint, then~$r$ is a root of polynomial $\varphi_{p,r}(X_1) - \varphi_{p,r}(X_2)$ of degree at most $\ell - 1$. There are at most $\ell - 1$ roots.
    Over all the different lengths $\ell \in [1\dd m]$, and all the pairs of length-$\ell$ substrings, there are less than $\sum_{\ell = 1}^m {\ell \cdot(m - \ell + 1)^2} \leq m^4$ values $r \in [0\dd p)$ that cause distinct equal-length substrings to have the same fingerprint. We find a prime  $p \in [\sigma + m^4\dd2\sigma + 2m^4)$ with \cref{fact:findprime_in_range}, and then systematically try all the $r \in [0\dd p)$ in $\cO(\poly(m + \sigma))$ time and constant space, until we find one that causes no collision.
\end{proof}

\begin{fact}\label{fact:construct_probabilistic_collision_free_fp}
   For any string $R \in [0\dd \sigma)^m$ and integer parameter $\err \geq \log_2 (m\sigma)$, we can find a prime $p \in [2^{5\err}\dd 2^{5\err + 1})$ and a base $r \in [0\dd p)$ such that, for every $\ell \in [1\dd m]$, there are no distinct length-$\ell$ substrings $X_1,X_2$ of~$R$ such that $\varphi_{p,r}(X_1) = \varphi_{p,r}(X_2)$. This takes $\cO(\poly(\err))$ time and succeeds with probability at least $1 - 2^{-\err}$.
\end{fact}

\begin{proof}
    We compute the prime with \cref{fact:findprime_in_range}, and draw $r \in [0\dd p)$ uniformly at random. As observed in the proof of \cref{fact:construct_deterministic_collision_free_fp}, less than $m^4$ values of~$r$ cause a collision, and thus the success probability is at least $1 - \frac{m^4}{p} \geq 1 - 2^{-\err}$. The inequality follows from $p \geq 2^{5\err} \geq m^4 \cdot 2^\err$.
\end{proof}

\subsection{Fiat--Naor Function Inversion Data Structure}
In the function inversion problem, we are given a function $f : [N] \rightarrow [N]$ and a value $y \in [N]$, and must either report some element of $f^{-1}(x) = \{ x \in [N] \mid f(x) = y\}$, or report that $f^{-1}(x) = \emptyset$.
Assume that $f(x)$ can be computed for any $x \in [N]$ in constant time.
A fundamental problem is to design an inversion data structure with good space-time trade-off; if we are allowed to preprocess~$f$ in order to construct a data structure of size $\tO(\tau)$ (for some parameter $\tau$), what is the fastest achievable inversion time?
A well-known lower bound states that the query time is $\tilde\Omega(N / \tau)$ \cite{DBLP:conf/stoc/Yao90} (also see~\cite{DBLP:journals/eccc/DeTT09} for an alternative proof).
The best known upper bound is due to Fiat and Naor, who showed that $\tO(N^3 / \tau^3)$ inversion time can be achieved~\cite{doi:10.1137/S0097539795280512}:

\inversion

Note that we state an upper bound $\tO(\tau)$ for the working space needed by the inversion data structure at query time, which was not explicitly stated (but can be readily verified) in \cite{doi:10.1137/S0097539795280512}.
Below, we show that inverting a function $f : [N] \rightarrow [M]$ with $M > N$ is not significantly harder than inverting a function $g : [N] \rightarrow [N]$.

\begin{lemma} \label{lem:inversion_reduce_codomain}
    Given a function $f: [N] \rightarrow [M]$ and a positive integer parameter $\err \geq \log_2 N$,
    we can in $\cO(\poly(\err))$ time, and with success probability at least $1 - 2^{-\err}$, construct $k = \cO(\err)$ functions $f_1, \dots, f_k : [N] \rightarrow[N]$ with the following property.
    For any $i \in [k]$ and $x \in [N]$, evaluating $f_i(x)$ can be reduced in $\cO(1)$ time to evaluating $f(x)$.
    An inversion query for~$f$ can be reduced in $\cO(\err)$ time to one inversion query for each of $f_1, \dots, f_k$.
\end{lemma}

\begin{proof}
    The problem is trivial for $M \leq N$ (using a single function $f_1 = f$).
    An intuitive idea would be to use a hash function $h : [M] \rightarrow [N]$.
    If there are no $x_1,x_2 \in [N]$ such that $f(x_1) \neq f(x_2)$ but $h(f(x_1)) = h(f(x_2))$,
    then inverting $y \in [M]$ with respect to~$f$ is equivalent to inverting $h(y)$ with respect to $h \circ f$.
    Since we cannot afford to store a perfect hash function, we use a slightly more involved probabilistic approach.
    We compute a prime $p \in [2N\dd 4N)$ using \cref{fact:findprime_in_range} with parameter $\err+1$ in $\cO(\poly(\err))$ time, failing with probability at most $2^{-\err}/2$.
    Then, we use $k = 2\err + 1$ hash functions $h_1, \dots, h_k : [M] \rightarrow [1 \dd p]$ drawn independently and uniformly at random from a family of pairwise independent hash functions.
    This can be done in $\cO(\err)$ time by defining $\forall i \in [k] : h_i(x) = 1 + ((a_i \cdot x + b_i) \bmod p)$, where $a_i, b_i \in [0\dd p)$ are drawn uniformly at random.
    For $i \in [k]$, let
    \begin{align*}
    C_i &=\{ x_1 \in [N] \mid \exists x_2 \in [N] \textnormal{\ such that\ } f(x_1) \neq f(x_2)\textnormal{\ and\ }h_i(f(x_1)) = h_i(f(x_2)) \quad \textnormal{\ and \ }\\
    C &= \bigcap\nolimits_{i \in [k]} C_i.
    \end{align*}
    Since $h_i$ is pairwise independent, for any $x_1, x_2 \in [N]$ with $f(x_1) \neq f(x_2)$, we have $\prob[h_i(f(x_1)) = h_i(f(x_2))] = 1/p$.
    We apply the union bound over all the $x_2 \in [N]$ with $f(x_1) \neq f(x_2)$ to obtain $\prob[x_1 \in C_i] \leq N/p \leq 1/2$.
    We then apply the union bound first over all $i \in [k]$ and then over all $x_1 \in [N]$, resulting in $\prob[x_1 \in C] \leq 2^{-(2\err + 1)} \leq 2^{-\err} / (2N)$ and $\prob[C \neq \emptyset] \leq 2^{-\err} / 2$.
    
    For $i \in [k]$, we define $f_i : [N] \rightarrow [p]$ with $f_i(x) = (h_i \circ f)(x)$, where $f_i(x)$ can be computed in constant time given $f(x)$. (We reduce the co-domain to $[N]$ later.)
    For some $x \in [N]$, let $y = f(x)$.
    It clearly holds ${f^{-1}(y) \subseteq f_i^{-1}(h_i(y))}$.
    If $C = \emptyset$ then there is $i \in [k]$ with $x \notin C_i$.
    This implies that, for all $x' \in f_i^{-1}(h_i(y))$, we have $f(x') = y$, that is, inverting $f_i$ at point $h_i(y)$ has the same effect as inverting~$f$ at point~$y$.
    Therefore, we can use the following query procedure.
    Given an inversion query $y \in [M]$, for each $i \in [k]$, we invert $f_i$ to obtain $x_i \in f_i^{-1}(h_i(y))$.
    For each $x_i$, we check whether $f(x_i) = y$.
    If this does not hold for any $x_i$, then we return that $f^{-1}(y) = \emptyset$.
    The construction fails if we fail to compute a prime~$p$ or if $C \neq \emptyset$, and thus with probability at most~$2^{-\err}$.

    Finally, each $f_i$ has co-domain $[p] \subset [4N]$. For an aesthetic improvement, we split each function~$f_i$ into five functions $f_{i,1}, \dots, f_{i,5} : [N] \rightarrow [N]$ with $f_{i,j}(x) = N$ if $f(x) \notin ((j - 1) \cdot (N-1)\dd j \cdot (N - 1)]$, and otherwise $f_{i,j}(x) = f_i(x) - (j - 1) \cdot (N-1)$.
    Hence, inverting $f_i$ at position~$y$ is the same as inverting $f_{i,j}$ for $j = \ceil{y / (N - 1)}$ at position $f_i(x) - (j - 1) \cdot (N-1)$.
\end{proof}

\section{Suffix Random Access Data Structures}
\label{sec:ra_data_structures}

In this section, we consider the problem of maintaining a $k$-error random access data structure formalized in the beginning of \cref{sec:tov_ub_ra}.

We consider two computational models: in the first, the text is given offline; in the second, it arrives as a stream.
In both settings, the data structure is computed in two stages. First, during a preprocessing stage, we only have access to the reference (and not to the text).
Then, in the second stage, we are allowed to access both the text and the reference to complete the construction of the data structure.

We measure the performance of the data structure in terms of
the preprocessing time (i.e., the time spent preprocessing only the reference, independently of the text),
the query time (i.e., the time needed to answer a random access query),
and the space complexity (i.e., the worst-case size of the data structure excluding the space required to store $R$ at any moment in time).
We stress that the space complexity includes the additional working space used during preprocessing and construction.
Unless explicitly stated otherwise, we measure the space usage in memory words.
For the aforementioned measures, we usually express the worst-case performance as a function of the reference length $m$, i.e., we state the worst-case time or space required for an arbitrary text.

If the text is given offline, then we also measure the construction time \emph{after preprocessing}, i.e., the time needed for completing the construction of the data structure after the first access to~$T$.
If the text arrives as a stream, then we have to maintain the random access data structure for the currently known prefix of the text.
Immediately before $T[i]$ arrives, we have already constructed the random access data structure for text $T[1\dd i)$. When $T[i]$ arrives, we have to update the data structure so that it works for text $T[1\dd i]$.
In this setting, we measure the update time needed for each arriving symbol. In the streaming setting, the construction time is then the product of the update time and the length of the stream.
Note that every data structure for the streaming setting yields a data structure for the offline setting.

\subsection{Auxiliary Lemmas}

During the preprocessing stage, we compute some information that depends only on the reference.
At construction time, i.e., once the text is first accessed, we may still modify (e.g., overwrite) this precomputed information.
Hence, if the goal is to construct the data structure for multiple texts, the preprocessing has to be repeated for each text. 
Below, we show that, for offline random access data structures, we can assume without loss of generality that the precomputed information is not modified at construction time, and hence it suffices to run the preprocessing once.

\begin{lemma}\label{lem:preprocess_once_use_many_times}
    Assume that there is an offline $k$-error random access data structure with preprocessing time $\cO(p(m))$, construction time $\cO(n \cdot c(m))$, query time $\cO(q(m))$, and space complexity $\cO(s(m))$.
    There is an offline $k$-error random access data structure with the same asymptotic preprocessing time, construction time, query time, and space complexity that, after preprocessing, never modifies the information computed during preprocessing.
\end{lemma}

\begin{proof}
    Let $\mathcal A$ be the data structure of size $\cO(s(m))$ computed during the preprocessing stage.
    If the text is of length $n \leq s(m)$, then we can use a copy of the text as a random access data structure with construction time $\cO(n)$, query time $\cO(1)$, and space complexity $\cO(s(m))$.
    If $n > s(m)$, then we first create a copy of $\mathcal A$ in $\cO(s(m)) \subseteq \cO(n)$ time.
    Then, we run the construction algorithm using this copy, resulting in overall $\cO(n \cdot c(m))$ construction time.
\end{proof}

\begin{remark}\label{rem:preprocess_once_use_many_times}
	Henceforth, we assume that \cref{lem:preprocess_once_use_many_times} is applied in the construction of any offline $k$-error random access data structure.
\end{remark}

Another useful observation is that, by using exponential search for finding an approximation of the support-length, we can reduce the construction time so that, instead of depending on the text length $n$, it depends on the actual support-length~$h$ of the computed data structure.

\begin{lemma}\label{lem:construct_in_h_time}
    If there is an offline $k$-error random access data structure with preprocessing time $\cO(p(m))$, construction time $\cO(n \cdot c(m))$, query time $\cO(q(m))$, and space complexity $\cO(s(m))$, then there is an offline $k$-error random access data structure with construction time $\cO(h \cdot c(m))$, where~$h >0$ is the support-length of the data structure. The asymptotic preprocessing time, query time, and space complexity remain unchanged.
\end{lemma}

\begin{proof}
We will construct an offline $k$-error random access data structure for multiple suffixes of~$T$.
By \cref{lem:preprocess_once_use_many_times}, we have to perform the preprocessing only once.
We compute the data structure at most $\ceil{\log_2 n} + 1$ times. 
In step $i = 0$, we store a copy of $T[n]$ as a random access data structure with support-length~$1$ for text $T(n - 2^0\dd n] = T[n]$.
In step $i > 0$, we construct the $k$-error random access data structure for text $T(n - \min(n, 2^i)\dd n]$. 
If $i < \ceil{\log_2 n}$ and the support-length of the data structure is $2^i$, then we continue with the next step.
Otherwise, if the support-length lies in the range $[2^{i - 1}\dd 2^i)$, then we return the data structure constructed in step~$i$.
Else, the support-length is less than $2^{i - 1}$, and we instead return the data structure constructed during (the previous) step $i - 1$, which has support-length $2^{i - 1}$. 
Either way, if we terminate in step $i$, then the support-length of the returned data structure is $h \geq 2^{i - 1}$, and the overall construction time is $\cO(\sum_{j=0}^i 2^j \cdot c(m)) = \cO(h \cdot c(m))$.
The construction space is $\cO(s(m))$ because we only have to store the data structures computed in the two most recent levels.
\end{proof}

\subsection{Reduction from Streaming to Offline Data Structures}

Every streaming algorithm for constructing a random access data structure implies an offline algorithm by simply treating the offline text as a stream.

\begin{observation}
If there is a streaming $k$-error random access data structure with amortized update time $\cO(u(m))$, then there is an offline $k$-error random access data structure with construction time $\cO(n \cdot u(m))$.
The space complexity, preprocessing time, and query time are identical.
\end{observation}

In the remainder of the section, we show that every offline construction algorithm also implies a streaming construction algorithm.
The general idea is to decompose the stream into blocks, and then compute the offline data structure of each block (carefully de-amortizing the construction time over the arriving symbols, similarly to what was done in \cite{DBLP:journals/iandc/CliffordEPP11}). We will periodically merge the data structures of adjacent blocks using the following auxiliary \lcnamecref{lem:mergeoffline}.

\begin{lemma}\label{lem:mergeoffline}
Assume that there is an offline $k$-error random access data structure with 
construction time $\cO(n \cdot c(m))$, query time $\cO(q(m))$, and space complexity~$\cO(s(m))$.
If this data structure is given for texts $T_1[1\dd n_1]$ and $T_2[1\dd n_2]$, then a $k$-error random access data structure for $T_1 \cdot T_2$ can be constructed in $\cO((n_1 + n_2) \cdot q(m) \cdot c(m))$ time.
\end{lemma}

\begin{proof}
Let $T = T_1 \cdot T_2$ and $n = n_1 + n_2$. 
Let $h_1,h_2$ be respectively the support-lengths of the data structures for $T_1$ and $T_2$.
Let $h = h_2$ if $h_2 < n_2$ and $h = h_1 + n_2$ otherwise.

By our assumption from \cref{rem:preprocess_once_use_many_times}, we assume that we have access to the information computed during the preprocessing of $R$ (see \cref{lem:preprocess_once_use_many_times}).
Using the data structures for $T_1$ and $T_2$, we can simulate random access to $T(n - h \dd n]$ in $\cO(q(m))$ time.
Hence we can run the construction algorithm for $T(n - h \dd n]$ in $\cO(h \cdot q(m) \cdot c(m)) \subseteq \cO(n \cdot q(m) \cdot c(m))$ time using $\cO(s(m))$ space (using \cref{lem:construct_in_h_time}).

It remains to be shown that the constructed data structure works correctly.
For $x \in \{1, 2\}$, if $h_x < n_x$ then $T_x[n_x - h_x\dd n_x]$ cannot be written as a factorization of~$k$ symbols and $k + 1$ substrings of~$R$ (by the definition of $k$-error random access data structures).
If $h < n$, then $T[n - h\dd n]$ either has prefix $T_1[n_1 - h_1 \dd n_1]$, or it has prefix (and is equal to) $T_2[n_2 - h_2 \dd n_2]$.
Hence $T[n - h\dd n]$ cannot be written as a factorization of~$k$ symbols and $k + 1$ substrings of $R$ either, and it thus suffices to consider $T(n - h\dd n]$.
\end{proof}

\newcommand{\z}{z}

\begin{observation}\label{obs:datastructure_of_substring}
	Consider a text $T \in \Sigma^n$, a reference $R \in \Sigma^m$, and $1 \le i \le \z \le j \le n$.
	If $T[\z\dd n]$ can be factorized into~$k$ symbols and $k+1$ substrings of $R$, then so can $T[\z \dd j]$.
	If $\z > i$ and $T[\z - 1\dd j]$ cannot be factorized in such a way, then neither can $T[\z - 1 \dd n]$.
\end{observation}

\begin{restatable}{theorem}{lemreducestreamingtooffline}\label{lem:random_access_reduce_streaming_to_offline}
If there is an offline $k$-error random access data structure with preprocessing time $\cO(p(m))$, construction time $\cO(n \cdot c(m))$, query time $\cO(q(m))$, and space complexity $\cO(s(m))$, then there is a streaming $k$-error random access data structure with preprocessing time $\cO(p(m))$, worst-case update time $\cO(q(m) \cdot c(m) \cdot \log m)$, query time $\cO(q(m))$, and space complexity $\cO(s(m) \cdot \log m)$.
\end{restatable}

\begin{proof}
During the lifetime of the streaming data structure, we will construct the offline data structure for multiple substrings of~$T$.
By \cref{lem:preprocess_once_use_many_times}, we only have to perform the $\cO(p(m))$ time preprocessing once. 
The streaming data structure does not require any further preprocessing.

We first describe the streaming data structure without explaining its construction.
We conceptually divide the text into blocks on $\ceil{\log_2m} + 1$ levels.
On level $\ell \in [0\dd \ceil{\log_2m}]$, the blocks are $T_{\ell,\z} = T((\z - 1) \cdot 2^\ell \dd \z \cdot 2^\ell]$ for $\z \in [1\dd \floor{\absolute{T} / 2^\ell}]$.
After receiving and processing~$T[i]$, for each level, we store the offline data structures of the rightmost three blocks fully contained in $T[1\dd i - 2^\ell]$.
More precisely, let~$\z$ be the maximal integer such that $\z \cdot 2^\ell \leq i$, then we store the offline data structures of $T_{\ell, \z - 1}$, $T_{\ell, \z - 2}$, and $T_{\ell, \z - 3}$ (respectively assuming $\z > 1$, $\z > 2$, and $\z > 3$).
See \cref{fig:ra_reduce_streaming_to_offline} for a visualization. We explicitly store the most recently received symbol $T[i]$.

\begin{figure}
\def\lenm{32}
\def\leni{37}
\def\lenn{53}
\def\ygap{1}
\input{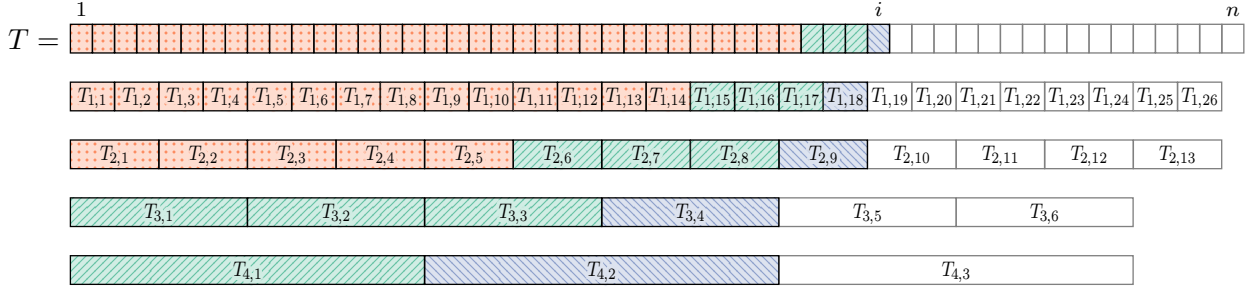}
\caption{%
Text partitioning for the proof of \cref{lem:random_access_reduce_streaming_to_offline}. 
After $T[i]$ arrives, the streaming data structure consists of the offline data structures of the blocks marked \insertedblock. 
The data structures of previous blocks, marked \deletedblock, have already been deleted. 
Blocks marked \underconstructionblock{} have already been received, but their offline data structures have not been contructed yet.
The offline data structure of block $T_{3,4}$ will be constructed while $T_{2,10}$ arrives. 
As soon as the last symbol of $T_{2, 10}$ arrives, the offline data structure of $T_{3,1}$ will be deleted.
}\label{fig:ra_reduce_streaming_to_offline}
\end{figure}

\paragraph{Maintaining the data structure.} Consider a block $T_{\ell,\z}$ on any level.
We finish the construction of its offline data structure as soon as we receive the last symbol of $T_{\ell,\z + 1}$, and we have to delete its data structure as soon as we receive the last symbol of $T_{\ell,\z + 4}$. 

We show how to construct the offline data  structures by induction over the levels.
For a block $T_{0, \z} = T[\z]$ of the lowest level, the offline data structure is merely a copy of the symbol, which we immediately create on its arrival.
Now inductively assume that we have already shown how to maintain the offline data structures for some level $\ell - 1 \geq 0$.
Next, we show how to maintain the offline data structures for level $\ell$.
Note that $T_{\ell,\z} = T_{\ell - 1, 2\z - 1} \cdot T_{\ell - 1, 2\z}$.
We will compute the offline data structure of $T_{\ell,\z}$ while the second half $T_{\ell - 1, 2\z + 2}$ of $T_{\ell,\z + 1}$ arrives.
By the inductive assumption, during this time, the offline random access data structures of $T_{\ell - 1, 2\z - 1}$ and $T_{\ell - 1, 2\z}$ are available.
Hence we can use \cref{lem:mergeoffline} to construct the offline random access data structure of $T_{\ell, \z}$ in $\cO(2^\ell \cdot c(m) \cdot q(m))$ time, or $\cO(c(m) \cdot q(m))$ time for each of the $2^{\ell - 1}$ symbols of $T_{\ell - 1, 2\z + 2}$.

At any moment in time, we run the offline construction algorithm for $\cO(\log m)$ blocks.
As each algorithm requires $\cO(c(m) \cdot q(m))$ time per arriving symbol,
the update time is as claimed.

\paragraph{Maintaining the support-length.}
We maintain the leftmost supported position~$j$ as follows.
When~$T[1]$ arrives, we assign $j \gets 1$. Then, whenever we finish constructing the offline data structure for some block $T_{\ell,\z}$, we obtain its support-length~$h$. If $h < 2^\ell$, then $T_{\ell,\z}[2^\ell - h\dd2^\ell] = T[\z \cdot 2^\ell - h \dd \z \cdot 2^\ell]$ cannot be factorized into~$k$ symbols and $k+1$ substrings of $R$, and, by \cref{obs:datastructure_of_substring}, neither can $T[\z \cdot 2^\ell - h\dd n]$.
Hence we assign $j \gets \max(j, \z \cdot 2^\ell - h + 1)$.
The support-length $i - j + 1$ can be returned in constant time.

\paragraph{Answering queries.}
We claim that a query $x \in [j \dd i)$ can be answered using the offline $k$-error random access data structure of block $T_{\ell, \z}$ with $\ell = \floor{\log_2 (i - x)} - 1$ and $\z = \ceil{x / 2^\ell}$.
It clearly holds that $x \in ((\z - 1) \cdot 2^\ell \dd \z \cdot 2^{\ell}]$, and thus $T[x] = T_{\ell, \z}[2^\ell - \z \cdot 2^\ell + x]$.
It remains to be shown that we have already computed and not yet deleted the random access data structure of $T_{\ell, \z}$, i.e., $(\z + 1) \cdot 2^\ell \leq i$ and $(\z + 4) \cdot 2^\ell > i$.
This holds due to
\begin{alignat*}{2}
(\z + 1) \cdot 2^\ell \leq (x/2^\ell + 2) \cdot 2^\ell = {}&x + 2^{\floor{\log_2 (i - x)}} &&{} \leq i\text{, and}\\
(\z + 4) \cdot 2^\ell \geq (x/2^\ell + 4) \cdot 2^\ell = {}&x + 2^{\floor{\log_2 (i - x)} + 1} &&{} > 
i.
\end{alignat*}
Computing $\ell$ takes constant time after $\cO(\log m)$-time preprocessing (\cref{lem:computelog}). We perform this preprocessing when $T[1]$ arrives, without asymptotically increasing the update time.
\end{proof}

\subsection[Reduction from k-Error to Exact Data Structures]{{\boldmath Reduction from $k$-Error to Exact Data Structures\unboldmath}}

\begin{restatable}{theorem}{lemreducekerrortonoerror}\label{lem:random_access_reduce_kerror_to_noerror}
Assume that there is an offline $0$-error random access data structure with preprocessing time $\cO(p(m))$, construction time $\cO(n \cdot c(m))$, query time $\cO(q(m))$, and space complexity $\cO(s(m))$.
For every positive integer $k$, there is an offline $k$-error random access data structure with preprocessing time $\cO(p(m))$, construction time $\cO(n \cdot c(m))$, query time $\cO(q(m) +\log k)$, and space complexity $\cO(k \cdot s(m))$. 
\end{restatable}

\begin{proof}
    We will construct the $0$-error data structure for multiple substrings of~$T$.
    By \cref{lem:preprocess_once_use_many_times}, it suffices to perform the preprocessing only once.
    The $k$-error data structure is defined by a sequence of $k' \leq 2k + 3$ integers $n + 1 = n_1 > n_2 > \dots > n_{k'} \geq 1$.
    For every $i \in [1 \dd k')$, we store a random access data structure with support-length $n_i - n_{i + 1}$ for $T[1 \dd n_i)$.
    The sequence of positions and the data structures are obtained as follows. 
    The initial position is $n_1 = n + 1$. 
    Inductively assume that we have already computed $n_i$ for some $i \geq 1$. 
    If $n_i = 1$ or $i = 2k + 3$, then we terminate.
    Otherwise, we construct the $0$-error data structure for text $T[1\dd n_{i})$ using \cref{lem:construct_in_h_time}.
    This always results in a non-zero support-length $h_i \in [1\dd n_i)$, and we assign $n_{i + 1} \gets n_{i} - h_i$.

    Since the data structure computed for $T[1\dd n_{i})$ has support-length $h_i = n_{i} - n_{i + 1}$, it supports access to $T[n_{i + 1}\dd n_i)$.
    Hence, for every position $x \in [n_{k'} \dd n]$, we can access $T[x]$ using one of the data structures, and the overall support-length is $n - n_{k'} + 1$.
    The correct data structure can be found in $\cO(\log k)$ time by binary searching for the unique $i \in [1\dd k')$ such that $x \in [n_{i + 1}\dd n_{i})$.
    Constructing the $0$-error data structure for $T[1\dd n_i)$ takes $\cO(h_i \cdot c(m)) = \cO((n_i - n_{i + 1}) \cdot c(m))$ time. Over all $i$, the total construction time is thus $\cO(n \cdot c(m))$. The space complexity is $\cO(k' + k' \cdot s(m)) = \cO(k \cdot s(m))$.
    
    If $n_{k'} > 1$, then we have to show that $T[n_{k'} - 1 \dd n]$ cannot be written as a factorization of at most~$k$ symbols and~$k + 1$ substrings of~$R$.
    Due to $n_{k'} > 1$, we must have terminated the construction with $k' = 2k + 3$. Since the support-length of the data structure computed for $T[1\dd n_i)$ is $h_i = n_i - n_{i + 1}$, it cannot be that $T[n_i - h_i - 1\dd n_i) = T[n_{i + 1} - 1 \dd n_i)$ is a substring of~$R$.
    Now consider any factorization of $T[n_{k'} - 1 \dd n]$ into at most $2k + 1$ factors.
    There must be some $i \in [1\dd 2k + 2]$ such that fragment $T[n_{i + 1} \dd n_i)$ does not contain the initial position of any factor. This implies that $T[n_{i + 1} - 1 \dd n_i)$ is fully contained in one of the factors, and this factor is neither a single symbol nor a substring of~$R$. 
\end{proof}

By combining \cref{lem:random_access_reduce_streaming_to_offline,lem:random_access_reduce_kerror_to_noerror}, we have:

\lemreducekerrorstreamingtonoerroroffline

\subsection{Warm-up: Simple Random Access Data Structure via Pattern Matching}
We start with a very simple data structure with a linear product of space and update time: 

\begin{lemma}\label{lm:random-warmup}
Let $\tau \in [1\dd m]$ be a parameter. There is a streaming $0$-error random access data structure with space complexity $\cO(\tau)$, worst-case update time $\cO(m/\tau)$, and query time $\cO(1)$. The data structure is deterministic and does not require any preprocessing.
\end{lemma}
\begin{proof}
At all times, we naively buffer the most recent $2\tau$ symbols of~$T$. Using $\cO(1)$ space, we maintain variables $\ell$ and~$i$ that satisfy the following invariant.
While $T((b + 1) \cdot \tau\dd (b + 2) \cdot \tau]$ arrives for any non-negative integer~$b$, $\ell$ is the maximal value in $[0, \min\{b\tau, m\}]$ such that $T' = T(b\tau - \ell\dd b\tau]$ has an occurrence in $R$, and $R[i\dd i + \ell)$ is an arbitrary such occurrence.
(For $b = 0$, we trivially have $\ell = 0$.) If we maintain this invariant, then we have random access in constant time by either looking up one of the most recent $2\tau$ symbols in the buffer or by using the stored reference $R[i\dd i + \ell)$.

Now we show how to maintain the invariant. While $T((b + 1) \cdot \tau\dd (b + 2) \cdot \tau]$ arrives, we compute the longest suffix of $T' \cdot T(b\tau\dd (b + 1) \cdot \tau]$ that has an occurrence in~$R$. This can be done using exact pattern matching, where the pattern is the reverse of $T' \cdot T(b\tau\dd (b + 1) \cdot \tau]$ and the text is the reverse of~$R$.
The pattern matching algorithm from \cite{CROCHEMORE199233} runs in $\cO(m)$ time and uses constant additional working space.
If the pattern $\rev(T' \cdot T(b\tau\dd (b + 1) \cdot \tau])$ does not occur in $\rev(R)$, then the algorithm reveals the length and an occurrence of the longest pattern prefix that occurs in $\rev(R)$, which in our case is the reverse of the longest suffix of $T' \cdot T(b\tau\dd (b + 1) \cdot \tau]$ that has an occurrence in~$\rev(R)$.
We update $\ell$ and~$i$ accordingly.

The pattern matching algorithm takes at most $cm$ time for some constant $c > 0$.
We use the so-called time-slicing technique to de-amortize the computation.
Specifically, while $T((b + 1) \cdot \tau\dd (b + 2) \cdot \tau]$ arrives, we execute $\ceil{c \cdot m / \tau} = \cO(m/\tau)$ steps of the algorithm whenever a symbol arrives.
Hence, on arrival of $T[(b + 2) \cdot \tau]$, we have computed the new values of $\ell$ and~$i$.
\end{proof}

\subsection{Isolating the Computational Challenge: Core-Matching Queries}

For the more involved solution based on function inversion, we reduce the problem of computing a random access data structure to the following relaxed notion of pattern matching:

\defcustomproblem{Core-Matching Queries}{%
{Input}/{A positive integer $\tau>0$ and a reference $R \in \Sigma^m$.},%
{Query}/{Given an integer $n$, where $n = 2^\ell\cdot \tau$ for some $\ell \geq 0$, and constant-time random access to a text $T \in \Sigma^{3n}$, output an occurrence of $T(n\dd 2n]$ in~$R$. If $T$ does not occur in $R$, then the query algorithm may return a special failure symbol.}%
}

\lemreduceratocorematching

\begin{proof}

    \def\b#1{b_{#1}}
    We first show that, without loss of generality, we can assume that $n = 3 \cdot 2^L \cdot \tau$ for some ${L \in [0\dd \floor{\log_2(m/3\tau)}]}$. 
    If $n < 3\tau$, then a copy of~$T$ is a suitable random access data structure (and serves as a reduction to zero core-matching queries).
    If $n > m$, then it suffices to compute a data structure for $T(n - m\dd n]$.
    Hence we can assume that $n \in [3\tau\dd m]$.
    Let $L = \floor{\log_2(n/(3\tau))} \in [0\dd \floor{\log_2(m/\tau)}]$. If $n \neq 3 \cdot 2^L \cdot \tau$, then we construct the data structure for both $T[1 \dd 3 \cdot 2^L \cdot \tau]$ and $T(n - 3 \cdot 2^L \cdot \tau\dd n]$.
    It is then easy to output the support-length and answer queries without affecting the asymptotic complexities.
    From now on, we assume that $n = 3 \cdot 2^L \cdot \tau$.

\begin{figure}
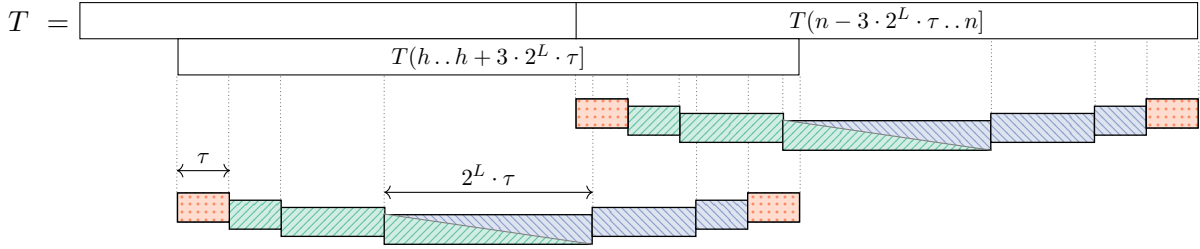

\input{figures/offline_to_core_matching_styles}
\begin{tikzpicture}[x=.765em, y=1em]

\def\numlvls{3}
\def\naivect{6}
\def\minwidth{.5\textwidth}
\def\twocol{1}

\foreach[evaluate=\loopyoff as \loopyoffplus using {\loopyoff + 1}] \loopxbase/\loopybase/\loopyoff/\looplrid in {5/0/4/1,23/1.25/2/0} {

\def\xbase{\loopxbase}
\def\ybase{\loopybase}
\def\yoff{\loopyoff}
\def\yoffplus{\loopyoffplus}
\def\lrid{\looplrid}

\input{figures/offline_to_core_matching}
}

\node[below left=0 and 0 of 0-lbox.north east, draw, minimum height=1.25em, minimum width=.9\textwidth, fill=white] (T) {};
\node[below left=0 and 0 of 1-lbox.north east, draw, minimum height=1.25em, minimum width=.5\textwidth, fill=white] {};
\draw (0-lbox.north west) to (0-lbox.south west);
\node[left=0 of T] {$T\ =\ $};
\node at (0-lbox.center) {\scalebox{.8}{$T(n - 3\cdot 2^L \cdot \tau\dd n]$}};
\node at (1-lbox.center) {\scalebox{.8}{$T(h \dd h + 3 \cdot 2^L \cdot \tau ]$}};

\draw[<->] ([yshift=.75em]1-east-1.north west) to node[midway, above=.25em] {\scalebox{.8}{$2^L \cdot \tau$}} ([yshift=.75em]1-east-1.north east);
\draw[<->] ([yshift=.75em]1-naive-west-1.north west) to node[midway, above=.25em] {\scalebox{.8}{$\tau$}} ([yshift=.75em]1-naive-west-\naivect.north east);
\end{tikzpicture}%
\input{figures/offline_to_core_matching_legend}%
\caption{%
Data structure layout from the proof of \cref{lem:reduce_offline_ra_to_core_matching} using $L = \floor{\log_2((n - h)/(3\tau))}$.
Blocks marked \insertedblock{} are substrings of~$R$ that match substrings of $T(h\dd h + 3\cdot 2^L \cdot \tau]$ and ${T(n - 3 \cdot 2^L \cdot \tau \dd n]}$. Blocks marked \underconstructionblock{} are matching substrings found by considering the reverses of~$R$ and~$T$.
Finally, blocks \deletedblock{} mark fragments of $\cO(\tau)$ naively stored symbols.%
}\label{fig:ra_offline_layout}
\end{figure}
    
    \paragraph{Data structure layout.}     
    We start by describing a random access data structure that requires that~$T$ is a substring of~$R$.
    It consists of a copy of $T[1\dd \tau]$ and, for every ${\ell \in [0\dd L]}$, a position~$i_\ell$ such that $R[i_\ell \dd i_\ell + 2^\ell \cdot \tau) = T(2^\ell \cdot \tau\dd 2^{\ell + 1} \cdot \tau]$.
    This representation requires $\cO(\tau + \log(m/\tau))$ space.
    Now we show that this simple data structure already allows constant time random access to $T[1\dd 2n/3]$. The leftmost $\tau$ queries are trivially handled by the copy of $T[1\dd \tau]$.    
    Given a query position $i \in T(\tau\dd 2n/3]$, we observe that $i \in (2^\ell \cdot \tau \dd 2^{\ell + 1}\cdot \tau]$ for $\ell = \floor{\log_2((i - 1) / \tau)}$. Hence, we can compute~$\ell$ in constant time using \cref{lem:computelog}, and then return $T[i] = R[i_\ell + i - 2^\ell\cdot \tau - 1]$.
    For the remaining queries $i \in (2n/3\dd n]$, we construct the same data structure for the reversed text and reversed reference without asymptotically increasing the time and space complexities.

    Finding the occurrence $R[i_\ell \dd i_\ell + 2^\ell \cdot \tau) = T(2^\ell \cdot \tau\dd 2^{\ell + 1} \cdot \tau]$ corresponds to a core-matching query with parameter $\tau$ on the preprocessed reference~$R$ with query text $T[1\dd 3\cdot 2^\ell \cdot \tau]$. Hence computing the data structure can be reduced to $\cO(\log(m / \tau))$ core-matching queries.

    \paragraph{Handling \boldmath$T$ that does not occur in~$R$.\unboldmath}
    We assumed that~$T$ is a substring of~$R$. If it is not, then the construction might fail due to a core-matching query returning the failure symbol. Hence, instead of running the algorithm for the entire $T$, we binary search for the minimal ${h \in (\max(0, n - m) \dd n]}$ such that the construction succeeds for text $T(h\dd n]$. This increases the time complexity and the number of core-matching queries by an  $\cO(\log m)$ multiplicative factor.

    Deciding the input to each core-matching query takes constant time, and we need $\cO(\tau)$ additional time to create a copy of the length-$\tau$ prefix and suffix of $T(h\dd n]$, resulting in the claimed time complexity. 
\end{proof}

\subsection{A Simple Core-Matching Data Structure}

\begin{restatable}{lemma}{lemnaivecorematching}
\label{lem:naive_core_matching}
    Given $R \in \Sigma^m$ and integer parameters $\ell \geq 0$ and $\tau \geq 1$, let $n = 2^\ell \cdot \tau$.
    There is a data structure of size $\cO(m/n)$ that answers any core-matching query $T \in \Sigma^{3n}$ in $\cO(n \log m)$ time and $\cO(1)$ additional space. It can be constructed in $\cO(\poly(m))$ time and $\cO(m/n + \polylog(m))$ space deterministically, or in $\cO(m + (m \log m)/n + \poly(\err))$ time and $\cO(m/n + \poly(\err))$ space with success probability at least $1 - 2^{-\err}$, for any integer parameter $\err \geq \log_2 (m\sigma)$.
\end{restatable}

\begin{proof}
    We use Karp--Rabin fingerprints such that there are no two distinct length-$2n$ substrings of~$R$ that have the same fingerprint.
    We obtain the prime and the base with \cref{fact:construct_deterministic_collision_free_fp} or \cref{fact:construct_probabilistic_collision_free_fp}, resulting either in $\cO(\poly(m))$ time and $\cO(\polylog(m))$ space, or in $\cO(\poly(\err))$ time and space with success probability at least $1- 2^{-\err}$.
    We store the fingerprints $\phi(R(bn - 2n\dd bn])$ with $b \in [2\dd \floor{\frac{m}{n}}]$ in a balanced binary search tree, each associated with the value~$b$ (i.e., the fingerprints of length-$2n$ substrings starting at every~$n$ positions). The fingerprints can be computed in $\cO(m)$ time in total using \cref{obs:rollfp}, and the tree can be constructed in $\cO((m \log m) / n)$ time and $\cO(m / n)$ space.
    
    It is easy to see that, if $T= R[i\dd i+3n)$ for some $i \in [1\dd m - 3n + 1]$, then this occurrence of~$T$ fully contains the fragment $R(bn - 2n\dd bn]$ with $b = \floor{\frac{i+3n - 1}{n}}$. Hence we can find an occurrence of $T(n\dd 2n]$ as follows.
    We compute the sequence of fingerprints $\phi(T[j \dd j + 2n))$ for $j \in [1\dd n + 1]$ in total time $\cO(n)$ using \cref{obs:rollfp}.
    We then search for each fingerprint in the binary search tree in $\cO(\log(m / n))$ time. If the fingerprint is contained and associated with some value $b$, then we know that $\phi(T[j \dd j + 2n)) = \phi(R(bn - 2n\dd bn])$. We naively verify whether the two substrings actually match, and, if they do, we return that $R(bn - n - j + 1\dd bn - j + 1] = T(n\dd 2n]$.
    If the substrings mismatch, then $T[j \dd j + 2n)$ does not occur in~$R$ because distinct length-$2n$ substrings of~$R$ have distinct fingerprints, and we immediately return the failure symbol. Hence we run the $\cO(n)$-time verification at most once. If we find none of the fingerprints in the search tree, we also return the failure symbol.    
\end{proof}

We are now ready to obtain a streaming $k$-error random access data structure with a sublinear product of space and update time:

\cornaivekerror
\begin{proof}
We fix $\tau = \Theta(\sqrt{m})$ and preprocess~$R$ according to \cref{lem:naive_core_matching} for every value $\ell$ in $[0\dd \floor{\log_2(m/(3\tau))}]$. 
Then, we construct an offline $0$-error random access data structure using \cref{lem:reduce_offline_ra_to_core_matching}. This results in construction time $\tO(n)$, space complexity $\tO(\sqrt{m})$, and constant query time. 
When using the probabilistic version of \cref{lem:naive_core_matching}, we use parameter $2\err \geq 2\log_2 (m\sigma)$, which implies that each of the $\cO(\log m)$ data structures is constructed successfully with probability at least $1 - 2^{-2\err}$, and all of them with probability at least $1 - 2^{-2\err} \cdot \cO(\log m) \geq 1 - 2^{-\err}$ by the union bound.
Finally, combining the obtained offline $0$-error random access data structure with \cref{lem:random_access_reduce_kerrorstreaming_to_noerroroffline}, and recalling that we only have to run the offline preprocessing once (see \cref{lem:preprocess_once_use_many_times}), we obtain the stated result.
\end{proof}

\section{Core-Matching via Function Inversion}

For \cref{cor:naivekerror}, we used parameter $\tau = \Theta(\sqrt{m})$ to balance the $\cO(\tau)$ space contributed by \cref{lem:reduce_offline_ra_to_core_matching} and the $\cO(m/\tau)$ space contributed by \cref{lem:naive_core_matching}.
Next, we avoid the $\cO(m/\tau)$ space by replacing the core-matching solution from \cref{lem:naive_core_matching} with a function inversion data structure of size $\cO(\tau)$.
We start by reducing core-matching to function inversion.

\begin{restatable}{lemma}{lemseminaivecorematching}
\label{lem:semi_naive_core_matching}
    Let $\ell \geq 0$ and $\tau \geq 1$ be integers, and let $n = 2^\ell \cdot \tau$. Given a reference $R \in \Sigma^m$ and an integer parameter $\err \geq \log_2 (m\sigma)$, we can compute and store a sequence of $\cO(\err)$ functions with the following properties.
    Each function has domain and co-domain $[N]$ with $N = \cO(\frac{m}{n})$, can be stored in $\cO(1)$ space, and can be evaluated in $\cO(n)$ time using $\cO(1)$ space.
    A core-matching query $T \in \Sigma^{3n}$ can be reduced to inverting each of the functions $\cO(n)$ times, and the reduction can be performed deterministically in $\cO(n\err)$ time and $\cO(\err)$ additional working space.
    The functions can be constructed in $\cO(\poly(\err))$ time with success probability at least $1-2^{-\err}$.
\end{restatable}

\begin{proof}
    The solution is similar to \cref{lem:naive_core_matching}, but instead of querying a binary search tree we invert a function.
    We still use Karp--Rabin fingerprints such that there are no two distinct length-$2n$ substrings of~$R$ that have the same fingerprint.
    We obtain the prime $p$ and the base $r$ using \cref{fact:construct_probabilistic_collision_free_fp} in $\cO(\poly(\err))$ time and space with success probability at least $1- 2^{-(\err + 1)}$ (using parameter $\err + 1$).
    
    Let $f : [\floor{m / n} - 1] \rightarrow [p]$ be defined by $f(x) = \phi(R(xn - n\dd xn + n])$.    
    It is easy to see that, if $T = R[i\dd i+3n)$ for some $i \in [1\dd m - 3n + 1]$, then this occurrence of~$T$ fully contains the fragment $R(xn - n\dd xn + x]$ with $x = \floor{\frac{i+3n - 1}{n}} - 1$. Hence we can find an occurrence of $T(n\dd 2n]$ as follows.
    We compute the sequence of fingerprints $\phi(T[j \dd j + 2n))$ for $j \in [1\dd n + 1]$ in total time $\cO(n)$ using \cref{obs:rollfp}.
    For each fingerprint $y$, we try to obtain $x \in f^{-1}(y)$ by inverting $f$.
    We then naively verify whether $T[j \dd j + 2n) = R(xn - n\dd xn + n]$, and, if this condition holds, we return that $R(xn - j + 1\dd xn + n - j + 1] = T(n\dd 2n]$.
    If the substrings mismatch, then $T[j \dd j + 2n)$ does not occur in~$T$ because distinct length-$2n$ substrings of~$R$ have distinct fingerprints. In this case, and also if all inversion queries report that $f^{-1}(y) = \emptyset$, we return the failure symbol.

    Finally, the co-domain of $f$ is too large. We apply \cref{lem:inversion_reduce_codomain} to obtain $\cO(\err)$ functions with domain and co-domain $[\floor{m / n} - 1]$ in $\cO(\poly(\err))$ time and space. Then, one inversion of $f$ can be reduced to one inversion of each of the functions. By using parameter $\err + 1$, the construction succeeds with probability at least $1 - 2^{-(\err + 1)}$. Overall, the success probability is at least $1 - 2^{-\err}$, as claimed.
\end{proof}

\subsection{Sparse String Synchronizing Sets}

If we use the inversion data structure of Fiat and Naor (\cref{fact:inversion}) to implement \cref{lem:semi_naive_core_matching}, then on the lowest level ($\ell = 0$) we can implement a core-matching query in $\tO(\tau)$ space and $\tO(m^3/ \tau^4)$ time: the domain is of size $N = \cO(m/\tau)$, but we have to multiply the $\tO(N^3 / \tau^3)$ inversion time from \cref{fact:inversion} with the number of inversion queries and the evaluation time of the function, both of which are linear in $\tau$. Next, we show how to retain the domain size and the evaluation time of the function, while decreasing the number of inversion queries so that it is almost constant.

We achieve this by using string synchronizing sets \cite{DBLP:conf/stoc/KempaK19} to consistently and regularly select synchronizing substrings in both~$R$ and~$T$.
In the definition below, we retain the original consistency and density conditions from \cite{DBLP:conf/stoc/KempaK19}, but we also add a \emph{sparsity} condition that ensures that we select only a small number of substrings within any {length-$\tau$} window of $R$.

\syncsets

Our construction enforces the sparsity condition by sampling every substring with probability around $k/\tau$. 
Our sampling of substrings is not fully independent.
By using a $k$-wise independent hash function and Chernoff--Hoeffding bounds as stated below, we are able to satisfy the sparsity condition.

\begin{lemma}[{\cite[Theorem~5(III), $\delta = 1$]{DBLP:journals/siamdm/SchmidtSS95}}]\label{lem:chernoff_hoeffding}
    For a positive integer~$d$ and a set $D$, let $h : D \rightarrow [0\dd d)$ be drawn uniformly at random from a family of $k$-wise independent hash functions. Consider a subset $D' \subseteq D$ with $\absolute{D'} = kd$, and let $D^* = {\{ s \in D' \mid h(s) = 0\}}$.
    Then $\prob[\absolute{D^*} \notin [1\dd 2k)] \leq e^{-k/3}$.
\end{lemma}

For our construction of synchronizing sets, we define a function $\ident : \Sigma^{2\tau} \rightarrow \{0,1\}$ that identifies synchronizing substrings. Crucially, the function can be stored in small space and can be evaluated efficiently in a rolling manner, similarly to Karp--Rabin fingerprints. This way, we can produce the elements of a string synchronizing set as a stream.

\sss

\begin{proof}
    Since positions in~$S$ are chosen solely based on length-$2\tau$ substrings, it is easy to see that the consistency condition will be satisfied by any function $\ident$.
    Hence we only have to focus on achieving density and sparsity.
    For $A \in \Sigma^{2\tau}$, we define $\ident(A)$ as follows.
    \begin{enumerate}[label=$(\roman*)$]
        \item\label{sss_run_middle} If $\per(A) \leq \frac\tau3$, then $\ident(A) = 0$.
        \item\label{sss_run_sentinels} If $\per(A) > \frac\tau3$ and either $\per(A[2\dd2\tau]) \leq \frac\tau3$ or $\per(A[1\dd2\tau)) \leq \frac\tau3$, then $\ident(A) = 1$. 
        \item\label{sss_regular} Otherwise, $\ident(A) = 1$ if and only if  $h(A) = 0$, where $h : \Sigma^{2\tau} \rightarrow [0\dd d)$ is drawn uniformly at random from a family of $k$-wise independent hash functions, using $k = 9\err$ and ${{\tau / (12k)} \leq d \leq {\tau / (3k)}}$. Intuitively, by sampling each length-$2\tau$ substring with probability $\approx k / \tau$, we will likely select at least one and at most $\cO(k)$ out of any $\tau$ distinct substrings.
    \end{enumerate}
    \paragraph{{\boldmath $S$ is $\cO(k)$-sparse and $\tau$-synchronizing.\unboldmath}}
    We explain how to efficiently implement the functions $\ident$ and~$h$ later. First, we show that using $\ident$ to construct~$S$ indeed leads to an $\cO(k)$-sparse $\tau$-synchronizing set.
    Fix some $i \in [1\dd m-3\tau + 2]$ and let $B = R[i\dd i + 3\tau - 2]$. 
    Let $S_i = \{ j \in [1\dd \tau] \mid \ident(B[j\dd j + 2\tau)) = 0 \}$.
    
    If $\per(B) \leq \frac\tau3$, then also $\per(B[j \dd j + 2\tau)) \leq \frac\tau3$ for every $j \in [1\dd \tau]$. It follows from \ref{sss_run_middle} that $\ident(B[j \dd j + 2\tau)) = 0$, and thus $S_i = \emptyset$, which means that~$i$ satisfies both the density and the sparsity conditions.

    \def\freq{\textnormal{\textsf{freq}}}
    \def\infreq{\textnormal{\textsf{rare}}}
    It remains to consider the case when $\per(B) > \frac\tau3$, for which we have to show $1 \leq \absolute{S_i} \in \cO(k)$.
    We first focus on the positions contributed by step \ref{sss_run_sentinels}.
    Let us assume that~$B$ contains a length-$(2\tau - 1)$ substring of period $\leq \frac\tau3$ (as otherwise step \ref{sss_run_sentinels} will not contribute to $S_i$).
    Let $\ell, r \in [1\dd \tau + 1]$ be the respectively minimal and maximal values such that $p_\ell := \per(B[\ell\dd \ell + 2\tau - 1)) \leq \frac \tau3$ and ${p_r := \per(B[r\dd r + 2\tau - 1)) \leq \frac \tau3}$.
    Periodic substrings of different minimal periods $p_{\ell}$ and $p_r$ cannot overlap by more than $p_\ell + p_r$ symbols (see, e.g., \cite[Lemma~1]{DBLP:conf/focs/KolpakovK99}), which implies that 
    the entire $B[\ell \dd r + 2\tau - 1)$ has period $p = p_\ell = p_r$. Hence, every $B[j\dd j +2\tau)$ with $j \in [\ell\dd r)$ has period~$p$ and will not contribute to $S_i$ due to step \ref{sss_run_middle}.
    Also, we have $\ell > 1$ or $r \leq \tau$ (or both), since otherwise we would have $\per(B) = p$, contradicting that $\per(B) > \frac\tau3$.
    Due to step \ref{sss_run_sentinels}, $\ell > 1$ implies that $\ident(B[\ell - 1\dd \ell + 2\tau - 1)) = 1$ and $r \leq \tau$ implies that $\ident(B[r\dd r + 2\tau)) = 1$.
    Thus, if~$B$ contains a length-$(2\tau - 1)$ substring of period $\leq \frac\tau3$, then step \ref{sss_run_sentinels} contributes either one or two positions to~$S_i$, and hence the density condition is satisfied.

    Finally, we focus on the positions contributed by step \ref{sss_regular}.
    Let $D = \{ B[j\dd j + 2\tau) \mid j \in [1\dd \tau] \}$ be the set of distinct length-$2\tau$ substrings of~$B$.
    We partition~$D$ into $D_{\freq} \subseteq D$ and $D_{\infreq} = D \setminus D_{\freq}$, where $D_{\freq}$ contains exactly the substrings that occur at least four times in~$B$.
    Since $A \in D_{\freq}$ of length $2\tau$ has four occurrences in~$B$ of length $3\tau$, a common corollary of the periodicity lemma implies that $\per(A) \leq \frac\tau3$ (see, e.g., \cite{DBLP:conf/cpm/MiyazakiST97,DBLP:journals/tcs/KidaMSTSA03}),
    and hence $\ident(A) = 0$ due to step \ref{sss_run_middle}.
    Therefore, we only need to count the contributions of substrings in $D_{\infreq}$. We have $\absolute{D_{\infreq}} \leq \absolute{D} \le \tau \leq 12kd$ (where the latter holds due to $kd \geq \frac\tau{12}$).
    We split $D_{\infreq}$ into twelve subsets of size at most $kd$. By \cref{lem:chernoff_hoeffding} and the union bound, we conclude that $\{ A \in D_{\infreq} \mid h(A) = 0 \}$ contains more than $24k$ elements with probability at most $12e^{-k/3} < 2^{-k/3} = 2^{-3\err}$ (where the former inequality holds once $\err$ exceeds a sufficiently large constant).
    Since each infrequent substring occurs at most three times in $B$, this is also an upper bound on the probability that step \ref{sss_regular} contributes more than $72k$ positions to $S_i$.
    Applying the union bound over all $i$, we know that step \ref{sss_regular} contributes less than $72k$ positions to every $S_i$ (and hence the sparsity condition is satisfied) with probability at least $1 - m \cdot 2^{-3\err} \geq 1 - 2^{-2\err}$.

    Now we only have to show that the density condition is satisfied. We already know that this is the case when~$B$ contains at least one length-$(2\tau - 1)$ substring of period at most $\tau/3$. Hence we can assume that the period of every length-$2\tau$ substring is greater than $\tau/3$. This means that all length-$2\tau$ substrings are infrequent, which implies $\absolute{D_{\infreq}} \geq \tau / 3 \geq kd$.
    Due to \cref{lem:chernoff_hoeffding}, we know that ${\{ A \in D_{\infreq} \mid h(A) = 0 \}}$ (and thus also $S_i$) is empty with probability at most $e^{-k/3} < 2^{-k/3} = 2^{-3\err}$. Again, by applying the union bound, $S_i$ is non-empty for all~$i$ (and hence the density condition is satisfied) with probability at least $1 - m \cdot 2^{-3\err} \geq 1 - 2^{-2\err}$.

    \paragraph{Constructing the hash function.}
    We construct a hash function $h : \Sigma^{2\tau} \rightarrow [0\dd d)$ that is only guaranteed to be $k$-wise independent (with high probability) when restricted to its subdomain $\{ R[i\dd i + 2\tau) \mid i \in [1\dd m - 2\tau + 1]\}$; this is sufficient for our purposes.
    We use the Karp--Rabin fingerprinting function $\phi = \phi_{q,r}$ to map {length-$2\tau$} substrings to integers from $[q]$, using a prime $q \in [2^{5\err}\dd 2^{5\err + 1})$ and a base $r \in [0\dd q)$ drawn uniformly at random.     
    Then, we draw $h' : [q] \rightarrow [0\dd d)$ uniformly at random from a family of $k$-wise independent hash functions and define $h = h' \circ \phi$.
    If every pair of distinct length-$2\tau$ substrings of~$R$ have distinct fingerprints, then this clearly results in a $k$-wise independent function~$h$.
    Due to \cref{fact:collision_probability:allsubstrings}, this happens with probability at least $1 - \frac{m^3}{q} \geq 1 - 2^{-2\err}$.
    
    We set $h'$ to be a polynomial of degree $k-1$ with random coefficients over $\mathbb{F}_d$ for a large prime~$d$. If~$d$ is given, then $h'$ can be constructed and evaluated in $\cO(k)$ time.
    If $\tau \geq 12k$, then any prime $d \in [\ceil{\frac{\tau}{12k}} \dd 2\cdot\ceil{\frac{\tau}{12k}})$ satisfies $d < 2\cdot\ceil{\frac{\tau}{12k}} \leq \tau/(3k)$.
    We can find both~$q$ and~$d$ in $\cO(\poly(\err))$ time and space and with success probability at least $1-2^{-2\err}$ using \cref{fact:findprime_in_range}, where for~$d$ we exploit the fact that $\tau \leq m \leq 2^\err$.
    If $\tau < 12k$, then we use $h' : [q] \rightarrow \{0\}$, i.e., all positions in non-periodic regions are synchronizing. Then, in any window of length $\tau$, at most $\tau = \cO(k)$ positions are synchronizing, and the sparsity condition is trivially satisfied.

    Overall, we fail with probability at most $2^{-2\err}$ for each of the following reasons: we may not achieve sufficient density, we may not achieve sufficient sparsity, the fingerprint function may cause collisions, and the computation of either~$q$ or~$d$ may fail. By the union bound, we fail with probability at most $5 \cdot 2^{-2\err} < 2^{-\err}$, as claimed.

    \paragraph{Enumerating the synchronizing set.}
    Lastly, we show how to enumerate the set for any given string $R' \in \Sigma^{m'}$. We can assume that $m' = 3\tau - 1$ as for longer strings we can run the algorithm for all length-$(3\tau - 1)$ substrings with starting position $\equiv 1\pmod{\tau}$ in left-to-right order.
    Hence, our task is to output $\{ i \in [1\dd \tau] \mid \ident(R'[i\dd i + 2\tau)) = 1 \}$ in increasing order.

    We compute $p = \per(R'[\tau\dd2\tau])$ in $\cO(\tau)$ time and constant space using~\cite[Lemma 6]{DBLP:conf/esa/0001GGK15}.
    If $p \leq \frac\tau3$, then we extend the periodicity as far as possible to either side by naive scanning in $\cO(\tau)$ time, resulting in the respectively minimal and maximal $\ell' \in [1\dd\tau]$ and $r' \in (2\tau\dd3\tau]$ such that $\per(R'[\ell'\dd r')) = p$.
    If $R'$ has a fragment $R'[j\dd j + 2\tau)$ of period $\leq \tau/3$, then this fragment fully contains $R'[\tau\dd2\tau]$, and the periodicity lemma implies $\per(R'[j\dd j + 2\tau)) = p$.
    Therefore, the fragment is fully contained in $R'[\ell'\dd r')$, i.e., step \ref{sss_run_middle} applies exactly to the length-2$\tau$ fragments within $R'[\ell'\dd r')$.
    If $r' - \ell' \geq 2\tau - 1$, let $\ell = \ell' - 1$ and $r = r' - 2\tau + 1$. Otherwise, let $\ell = r = \tau + 2$.

    Now we are ready to compute the set.
    We compute $\phi(R'[i \dd i + 2\tau))$ with $i \in [1\dd \ell)$ in left-to-right order, taking $\cO(\tau)$ time overall (see \cref{obs:rollfp}).
    For each fingerprint, we compute $h(R'[i \dd i + 2\tau)) = h'(\phi(T[i \dd i + 2\tau)))$ in $\cO(\err)$ time and report that~$i \in S$ if and only if the hash value is zero.
    Then, if $\ell > 0$, we report that $\ell \in S$. Similarly, if $r \leq \tau$, we report that $r\in S$.
    Finally, we process positions $i \in (r \dd \tau]$ analogously to the ones in $[1 \dd \ell)$.
    This handles steps \ref{sss_regular} and \ref{sss_run_sentinels}, and produces the elements of $S$ in increasing order.
\end{proof}

\subsection{Reducing Core-Matching to Function Inversion}

\begin{theorem}\label{lem:reduce_core_matching_to_fi}
    Let $\ell \geq 0$ and $\tau \geq 1$ be integers, and let $n = 2^\ell \cdot \tau$. Given a reference $R \in \Sigma^m$ and an integer parameter $\err \geq \log_2 (m\sigma)$, we can compute and store a sequence of $\cO(\err)$ functions with the following properties.
    Each function is stored in $\cO(\err)$ space, has domain and co-domain $[N]$ with $N = \cO(\frac{m\err}{n})$, and can be evaluated in $\cO(n\err)$ time and $\cO(\err)$ additional working space.
    A core-matching query $T \in \Sigma^{3n}$ can be reduced to inverting each of the functions at most once, and the reduction can be performed deterministically in $\cO(n\err)$ time and $\cO(\err)$ additional working space.
    The functions can be constructed in $\cO(\poly(\err))$ time with success probability at least $1-2^{-\err}$.
\end{theorem}

\begin{proof}
    We describe two functions $f, g : [N] \rightarrow[q]$ with $q \in [2^{5\err}\dd 2^{5\err + 1})$ such that a core-matching query can be solved by inverting either~$f$ or~$g$. 
    Then, the \lcnamecref{lem:reduce_core_matching_to_fi} follows from \cref{lem:inversion_reduce_codomain} (invoked with an appropriate error parameter).
    The co-domain stems from Karp--Rabin fingerprints with prime number $q$, which can be found in $\cO(\poly(\err))$ time with probability at least $1 - 2^{-2\err}$ (see \cref{fact:findprime_in_range}).
    We assume that, for any $i, j \in [m - 2n + 1]$ with $R[i\dd i+2n) \neq R[j\dd j+2n)$, we have $\phi(R[i\dd i+2n)) \neq \phi(R[j\dd j+2n))$. By \cref{fact:collision_probability:allsubstrings}, this assumption holds with probability at least $1 - \frac{m^3}{q} \geq 1 - 2^{-2\err}$.

    \paragraph{Simple solution for short texts.} For the sake of explanation, and also as a solution for $n = \cO(\err)$, we show a simple approach that results in a function $f : [m - 2n + 1] \rightarrow[q]$. (If $n = \cO(\err)$, then $m - 2n + 1 = \cO((m\err) / n)$.) 
    We define $f(x) = \phi(R[i\dd i + 2n))$.
    Given a query text $T \in \Sigma^{3n}$, we compute $y = \phi(T[1\dd 2n])$ and invert~$f$ to find some $x' \in f^{-1}(y)$.
    Finally, we verify that $R[x'\dd x' + 2n) = T[1\dd 2n]$, which implies an occurrence $R[x' + n\dd x' + 2n) = T(n\dd 2n]$.

    If $R[x'\dd x' + 2n) \neq T[1\dd 2n]$, then $\phi(T[1\dd 2n]) = \phi(R[x'\dd x' + 2n))$ and our assumption on $\phi$ imply that $T[1\dd 2n]$ does not occur in $R$, and we return the failure symbol. We also return the failure symbol if the inversion returns $f^{-1}(y) = \emptyset$.
    The function~$f$ can be stored in constant space and evaluated in $\cO(n)$ time, and the reduction takes $\cO(n)$ time and constant additional space.

    \paragraph{Shrinking the domain.}
    We now reduce the size of the domain by considering only positions from an $\cO(\err)$-sparse $n$-synchronizing set.
    We construct the function $\ident$ from \cref{lm:sss}, invoking \cref{lm:sss} with density parameter~$n$ and sparsity parameter $2\err$.
    Let~$S$ be the resulting synchronizing set for~$R$ (which we do not actually construct and merely use for explanation).
    Let $k = \cO(\err)$ be chosen so that the synchronizing set is $k$-sparse (the precise value of~$k$ depends on the constants hidden in \cref{lm:sss}).
    We define $f : [k \cdot \ceil{\frac{m}{n}}] \rightarrow [q]$ as follows.
    For $b \in [\ceil{\frac{m}{n}}]$ and $j \in [1\dd k]$,
    we define $f(bk - k + j) = \phi(R[i\dd i+2n))$,
    where~$i$ is the $j$-th synchronizing position in $S \cap (bn - n \dd bn]$. If there are fewer than~$j$ synchronizing positions in this range, we define $f(bk - k + j) = \phi(\texttt{0}^{2\tau})$ instead.
    Due to the $k$-sparsity guarantee, it is clear that, for every $i \in S$, there are some values $b,j$ such that $f(bk - k + j) = \phi(R[i\dd i+2n))$.

    To answer a query $T \in \Sigma^{3n}$, we first find the minimal $j' \in [1\dd n + 1]$ such that $\ident(T[j'\dd j' + 2n)) = 1$ in $\cO(n\err)$ time and $\cO(\err)$ space using the algorithm from \cref{lm:sss}. For now, we assume that there is at least one such synchronizing position. (If not, then the density condition implies that~$T$ is periodic, and we use a different approach described later.)
    Next, we compute $y = \phi(T[j'\dd j' + 2n))$ and invert~$f$ to find $x \in f^{-1}(y)$.
    This tells us that the fingerprint~$y$ occurs at the $(x - k\floor{\frac{x}k})$-th synchronizing position in $S \cap (\floor{\frac{x}k}n \dd \floor{\frac{x}k}n + n]$.
    We can find this position in $\cO(n\err)$ time by enumerating the synchronizing set of $R(\floor{\frac{x}k}n \dd \floor{\frac{x}k}n + 3n]$ with the algorithm from \cref{lm:sss}.
    Let $i'$ be the obtained position. We then verify whether $R[i'\dd i' + 2n) = T[j' \dd j' + 2n)$ in $\cO(n)$ time, in which case $R(i' + n - j'\dd i' + 2n - j'] = T(n\dd 2n]$.

    The consistency condition guarantees that we always obtain $i'$ such that $R[i'\dd i' + 2n) = T[j' \dd j' + 2n)$ unless there is a fingerprint collision. As mentioned earlier, a collision implies that~$T$ does not occur in $R$, and, since we verify that the substrings match, we correctly detect this case. As before, we also return the failure symbol if the inversion returns $f^{-1}(y) = \emptyset$.

    \paragraph{Handling the periodic case.}
    If we do not find a synchronizing position in $T$, then the density condition implies that~$T$ has period at most $\frac n3$. Let $p_{\max} = \ceil{\frac{n}3}$.
    We solve this case with a simpler function $g : [\ceil{\frac{m}{p_{\max}}}] \rightarrow [q]$ defined by \[g(x) = \min\{\ \phi(R[i'\dd i' +n+p_{\max}) \mid i' \in (x\cdot p_{\max} - p_{\max}\dd x\cdot p_{\max}]\ \}.\]

    For answering the query, we compute $y = \min\{\ \phi(T[j\dd j +n+p_{\max}) \mid j \in (n - p_{\max}\dd n] \}$ as well as the position $j'$ that achieves the minimum. 
    We obtain $x \in g^{-1}(y)$ by inverting~$g$. Then we obtain a position $i'$ that achieves the minimum when computing $g(x)$, and verify that $R[i'\dd i' +n+p_{\max}) = T[j'\dd j' +n+p_{\max})$, which implies an occurrence $R(i' + n - j'\dd i' + 2n - j'] = T(n\dd 2n]$.
    Due to the periodicity of $T$, it can be readily verified that, if~$T$ has an occurrence in $R$, then there is also~$x$ such that $g(x) = y$.
    Computing $g(x)$, $y$, $i'$, and $j'$ can each be done using \cref{obs:rollfp} in $\cO(n)$ time and constant working space.
\end{proof}

\begin{theorem}\label{thm:core_matching_queries}
    Let $x,y,z \in \mathbb R$ be non-negative constants, and let $\tau, \Delta > 0$ be integer parameters.
    Assume that there is a function inversion data structure that, for a function $f : [N] \rightarrow [N]$ that can be evaluated in constant time and space, can be constructed in $\cO(N^{1+z} \cdot \polylog(N))$ time, answers queries in $\cO((1 + N^{x} / \tau^{y}) \cdot \polylog(N))$ time, and can be stored and queried using space $\cO(\tau \cdot \polylog(N))$. Further assume that the construction succeeds with probability at least $\frac12$.
    
    Then, there is an $\tO(\tau)$-space data structure that can answer a core-matching query for any text of length~$3n$, where $n \in \{2^\ell \cdot \Delta \mid \ell \in [0 \dd \floor{\log_2 (m/(3\Delta))}] \}$, in worst-case time $\tO(n \cdot (1 + m^x / (n^x \tau^{y})))$.
    The preprocessing takes $\cO(\poly(m))$ expected time and succeeds without error.
    The preprocessing can also be implemented in $\tO(m^{1 + z} / \Delta^{z} \cdot \poly(\err))$ time with success probability at least $1 - 2^{-\err}$ at the expense of increasing the space complexity and update time by an $\cO(\poly(\err))$ multiplicative factor, for any integer parameter $\err \geq \log_2 (m\sigma)$.
\end{theorem}
\begin{proof}
Let us fix some integer $\err \geq \log_2 (m\sigma)$.
    Throughout the proof, the $\tO$-notation not only hides $\polylog(m)$ factors, but also $\poly(\err)$ factors.

    We focus on one value $\ell$ and the associated $n = 2^\ell \cdot \Delta$.
    By \cref{lem:reduce_core_matching_to_fi}, in $\tO(1)$ time and space, we can construct a sequence of $\tO(1)$ functions such that inverting each function at most once suffices for solving a core-matching query of length $3n$.
    Each function is over a domain of size $\tO(m / n)$ and can be evaluated in $\tO(n)$ time.
    We construct a function inversion data structure for each function.
    By multiplying the preprocessing time $\tO((m / n)^{1+ z})$ for a constant-time computable function with the $\tO(n)$ evaluation time and the number $\tO(1)$ of functions, the preprocessing time for function inversion is \[{\tO((m / n)^{z + 1} \cdot n)} \subseteq \tO(m^{1 + z} / {n}^z) \subseteq \tO(m^{1 + z} / {\Delta}^z).\]
    Summing over all $\ell \in [1\dd \ceil{\log_2(m/(3\Delta))}]$, the total time for constructing the functions with \cref{lem:reduce_core_matching_to_fi} and then constructing the function inversion data structures is $\tO(m^{1 + z} / {\Delta}^z)$.
        The query time for function inversion for a text of length $n$ is
        $\tO(n(1 + (m/n)^x/\tau^y))$.
        
    If all inversion data structures were constructed successfully, then we already achieve the claims of the \lcnamecref{lem:reduce_ra_to_fi}; specifically, the construction of functions with \cref{lem:reduce_core_matching_to_fi} succeeds with probability at least $1- 2^{-2\err}$ (invoked with error parameter $2\err$), and the remaining steps are deterministic.
    
    By constructing $2\err$ inversion data structures for each function, at least one data structure will succeed for each function with probability $1 - 2^{-2\err}$.
    Whenever we have to invert a function, 
    we perform an inversion query using each of the inversion data structures that we have built for it, aborting the inversion query as soon as the time exceeds $\tO(n(1 + (m/n)^x/\tau^y))$ or the space exceeds $\tO(\tau)$.
    If a query actually returns a preimage, we verify that the result is correct by evaluating the function once; the time required for this is dominated by the time required for function inversion (which, in particular, performs function evaluations).
    If all queries either exceed the time, or return no preimage, or return an incorrect preimage, then we answer the core-matching query with the special failure symbol.

    Finally, to get a preprocessing algorithm without error, we use $\err = \ceil{\log_2 m}$. We can verify that the preprocessing was successful in $\cO(\poly(m))$ time as follows. We only have to check that there are no fingerprint collisions, that the computed synchronizing set is indeed sufficiently dense and sparse, and that there is at least one inversion data structure for each function that works correctly in the claimed time complexity.
    Each of these tasks can be performed naively by brute force in $\cO(\poly(m))$ time and $\tO(1)$ additional working space. If the construction fails, we restart it with new randomness. We already know that the construction succeeds with high probability and hence, the expected number of times we have to restart the construction is close to one.
\end{proof}

The data structure from \cref{fact:inversion} lets us use \cref{lem:reduce_ra_to_fi} with $x = y = 3$ and $z = 0$, thus obtaining the following result:

\begin{theorem}\label{cor:core_matching_queries}
	Consider integers $\tau,\Delta>0$ and a reference $R \in \Sigma^m$.
    There is an $\tO(\tau)$-space data structure that can answer a core-matching query for any text of length $3n$, where $n \in \{2^\ell \cdot \Delta \mid \ell \in [0 \dd \floor{\log_2 (m/(3\Delta))}] \}$, in worst-case time $\tO(n \cdot (1 + m^3 / (n^3\tau^{3})))$.
    The preprocessing takes $\cO(\poly(m))$ expected time and succeeds without error.
    The preprocessing can also be implemented in $\tO(m \cdot \poly(\err))$ time with success probability at least $1 - 2^{-\err}$ at the expense of increasing the space complexity and update time by an $\cO(\poly(\err))$ multiplicative factor, for any integer parameter $\err \geq \log_2 (m\sigma)$.
\end{theorem}

\subsection{Obtaining Suffix Random Access Data Structures}

\begin{theorem}\label{lem:reduce_ra_to_fi}
    Let $x,y,z \in \mathbb R$ be non-negative constants, and let $\tau > 0$ be an integer parameter.
    Assume that there is a function inversion data structure that, for a function $f : [N] \rightarrow [N]$ that can be evaluated in constant time and space, can be constructed in $\cO(N^{1+z} \cdot \polylog(N))$ time, answers queries in $\cO((1 + N^{x} / \tau^{y}) \cdot \polylog(N))$ time, and can be stored and queried using space $\cO(\tau \cdot \polylog(N))$. Further assume that the construction succeeds with probability at least $\frac12$.
    
    Then, for every positive integer $k = \tO(\tau)$, there is a $(k - 1)$-error streaming random access data structure with space complexity $\tO(\tau)$, worst-case update time $\tO(1 + k^xm^x / \tau^{x + y})$, and worst-case query time $\cO(1 + \log k)$.
    The preprocessing takes $\cO(\poly(m))$ expected time and succeeds without error.
    The preprocessing can also be implemented in $\tO(k^zm^{1 + z} / \tau^{z} \cdot \poly(\err))$ time with success probability at least $1 - 2^{-\err}$ at the expense of increasing the space complexity and update time by an $\cO(\poly(\err))$ multiplicative factor, for any integer parameter $\err \geq \log_2 (m\sigma)$.
\end{theorem}

\begin{proof}
Let us fix some integer $\err \geq \log_2 (m\sigma)$.
    Throughout the proof, the $\tO$-notation not only hides $\polylog(m)$ factors, but also $\poly(\err)$ factors.
    By \cref{lem:random_access_reduce_kerrorstreaming_to_noerroroffline}, it suffices to show that there is a $0$-error offline random access data structure with space complexity $\tO(\tau / k)$ that has constant query time and, for a text of length $n$, can be constructed in $\tO(n \cdot (1 + k^xm^x / \tau^{x + y}))$ time, after an $\tO(k^zm^{1+z} / \tau^{z})$-time preprocessing of $R$ and $\ceil{\tau / k}$.
    
    Let $\Delta = \ceil{\tau / k}$.
    By \cref{lem:reduce_offline_ra_to_core_matching}, constructing the $0$-error offline random access data structure can be reduced to $\tO(1)$ core-matching queries,
    where each query text is a substring of~$T$ of length $3n'$ with $n' = 2^\ell \cdot \Delta$ for some $\ell \in [0\dd \floor{\log_2(m/(3\Delta))}]$. The reduction takes $\tO(\Delta)$ time and space.
    For now, assume that we have an oracle that answers each of the core matching queries in $\tO(n' \cdot (1 + m^x / ({n'}^x\tau^y)))$ time and constant space.
    Then, the overall time to answer all queries is $\tO(n \cdot (1 + m^x / (\Delta^x\tau^y))) = \tO(n \cdot (1 + k^xm^x / \tau^{x + y}))$ and the overall space is $\tO(\Delta)$, as required.
    
    Instead of using the oracle, we can achieve the same time complexity by answering the queries using \cref{thm:core_matching_queries} with parameters $\tau$ and $\Delta$.
    The space increases by an additive $\tO(\tau)$ factor; however, we only have to invest this space once for the preprocessing of \cref{thm:core_matching_queries}, and once for each active query procedure (accounting for the working space of the core-matching data structure at query time). The reductions used to obtain \cref{lem:random_access_reduce_kerrorstreaming_to_noerroroffline} construct at most $\tO(1)$ offline $0$-error random access data structures at a time, and hence there are at most $\tO(1)$ active query procedures.
    Therefore, the additional space is $\tO(\tau)$, as claimed.
\end{proof}

The data structure from \cref{fact:inversion} lets us use \cref{lem:reduce_ra_to_fi} with $x = y = 3$ and $z = 0$, thus obtaining the following result:

\def\contentofstreamingkerror{%
    There is an algorithm that, for any reference $R \in \Sigma^m$ and positive integer parameters~$\tau$ and $k = \tO(\tau)$, maintains a $(k - 1)$-error streaming random access data structure with space complexity $\tO(\tau)$, worst-case update time $\tO(1 + k^3m^{3} / \tau^6)$, and query time $\cO(1 + \log k)$. The preprocessing succeeds without error in $\cO(\poly(m))$ expected time.
    
    The preprocessing can be implemented in $\tO(m \cdot \poly(\err))$ time with success probability at least $1 - 2^{-\err}$ at the expense of increasing the space complexity and update time by an $\cO(\poly(\err))$ multiplicative factor, for any integer $\err \geq \log_2 (m\sigma)$.
}
\def\versionofstreamingkerror{full version}

\streamingkerror

\section{Function Inversion Protocol and Lower Bound}

In computational complexity, a fundamental difference arises between \emph{uniform} and \emph{non-uniform} computation. 
In the former case, we assume that a single algorithm is used to solve input instances of all possible sizes (i.e., the size is part of the input).
In the latter case, we are allowed to design a family of algorithms using a different algorithm for each input size. 
We say \emph{non-uniform algorithm} to refer to the entire family.
Crucially, a non-uniform algorithm may embed an arbitrary amount of \emph{advice} that depends only on the size of the input.

In this section, we introduce a non-uniform function inversion protocol based on any offline random access data structure. Its time and space complexity for function inversion will depend on the time and space complexity of the random access data structure. Using well-known lower bounds on the time-space trade-off of permutation inversion, we then derive lower bounds on the time-space trade-off of random access data structures.

In the auxiliary lemma below, we show that, for any pair of a function $f$ and a string $S$, we can either use $f$ to derive a compressed representation of $S$, or we can use $S$ to derive a function inversion data structure for $f$.
Then, in \cref{lem:how_many_strings_sets_to_invert_any_function}, we show that a small set of strings is sufficient to derive an inversion data structure for \emph{any} function; otherwise, we would have a provably impossible compression scheme.
This set is then used to obtain the non-uniform inversion protocol.

\begin{lemma}\label{lem:compressor_or_inverter}
\let\tau\tauconst
Let $\tau, \delta > 0$ with $\tau < 1$ be constant. 
Assume that there is an offline random access data structure with 
space complexity $\cO(m^\tau)$ bits, amortized query time $\tO(m^\delta)$, and construction time $\tO(m^{\tau + \delta})$ for a text of length at most $m^\tau \log m$. For a positive integer $N$, let $M = \ceil{N^{\tau/(1 - \tau)} \cdot \log_2 N}$. For every string $S \in \{0,1\}^{NM}$ and every constant-time computable function $f : [N] \rightarrow [N]$, at least one of the following statements holds.
\begin{enumerate}
\item The string~$S$ can be compressed into a representation of $NM - M + o(M)$ bits that can be decompressed if access to~$f$ is given (i.e., if we can compute $f(x)$ for any $x \in [N]$).
\item There is a function inversion data structure of size $\cO(M)$ bits that, given any query $y \in [N]$, outputs an element of $f^{-1}(y)$ (or reports that $f^{-1}(y) = \emptyset$) in $\tO(N^{(\tau + \delta) / (1- \tau)})$ time, using $\cO(M)$ working space. At query time, the data structure needs random access to~$S$ and~$f$ (i.e., it needs to be able to compute $f(x)$ for any $x \in [N]$ in constant time). 
\end{enumerate}
\end{lemma}

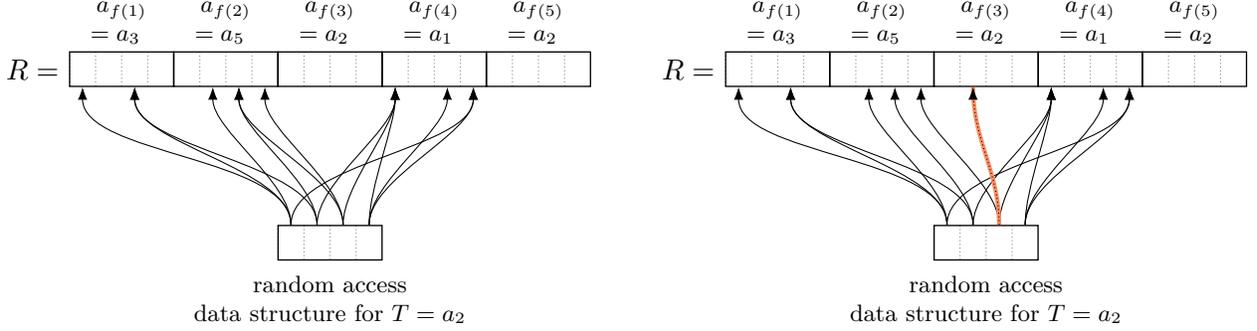
\begin{figure}
    \tikzset{
    every node/.style={inner sep=0pt},
    allblocks/.style={line width=.5pt, inner sep=-.75pt},
    }

    \def\N{5}
    \def\M{4}
    \def\f{3,5,2,1,2}
    \pgfmathsetmacro{\NM}{int(\N*\M)}

    \def\basedrawing{%
        \foreach[evaluate=\i as \im using int(\i - 1)] \i in {1,...,\NM} {
            \node (tl) at (\im, 0) {};
            \node (br) at (\i, -1) {};
            \node[fit=(tl)(br), allblocks] (p\i) {};
            \draw[gray, densely dotted] (tl.center) to (tl |- br.center);
        }

        \foreach[
            evaluate=\i as \rborder using {int(\i*\M)},
            evaluate=\rborder as \lborder using {int(\rborder-\M)},
        ] \i in {1,...,\N} {
            \node (tl) at (\lborder, 0) {};
            \node (br) at (\rborder, -1) {};
            \node[fit=(tl)(br), draw, allblocks] (b\i) {};
        }

        \foreach[count=\i from 1] \fofi in \f {
            \node[above=.25em of b\i, align=center] {\footnotesize{$a_{f(\i)}$}\\[-.2\baselineskip]\footnotesize{$ = a_{\fofi}\enskip$}};
        }

        \foreach[evaluate=\ip as \i using int(\ip + 8), evaluate=\i as \im using int(\i - 1)] \ip in {1,...,\M} {
            \node (tl) at (\im, -5) {};
            \node (br) at (\i, -6) {};
            \node[fit=(tl)(br), allblocks] (t\ip) {};
            \draw[gray, densely dotted] (tl.center) to (tl |- br.center);
        }
        \node (tl) at (8, -5) {};
        \node (br) at (12, -6) {};
        \node[fit=(tl)(br), draw, allblocks] (t) {};

        \node[below=.5em of t, align=center] {\footnotesize random access\\[-.2\baselineskip]\footnotesize data structure for $T = a_2$};
        \node[left=0 of p1] {$R=\ $};
    }

    \subcaptionbox{The random access data structure for text $a_2$ can simulate access to $a_2$ using only $a_1, a_3, a_5$ and~$f$, resulting in a compressed representation in which the random access data structure replaces $a_2$.}{%
    \begin{tikzpicture}[x=.9em, y=1.2em]
        \basedrawing
        \foreach[count=\src from 1] \dstset in {{1,3,6,16},{3,7,13},{7,8,13},{13,15,16}} {
            \foreach \dst in \dstset {
                \draw[-Latex] (t\src.north) to[out=90, in=270, looseness=.7] (p\dst.south);
            }
        }
    \end{tikzpicture}}%
    ~\hfill~%
    \subcaptionbox{The random access data structure for text~$a_2$ simulates access to $a_2$ by reading a symbol of $a_{f(3)} = a_2$. Hence we can find an element of $f^{-1}(2)$ by simulating access to $a_2$.}{%
    \begin{tikzpicture}[x=.9em, y=1.2em]
        \basedrawing
        \foreach[count=\src from 1] \dstset in {{1,3,6,16},{3,7,13},{8,13},{13,15,16}} {
            \foreach \dst in \dstset {
                \draw[-Latex] (t\src.north) to[out=90, in=270, looseness=.7] (p\dst.south);
            }
        }
        \draw[ultra thick, my-red] (t3.north) to[out=90, in=270] (p10.south);
        \draw[densely dotted] (t3.north) to[out=90, in=270] (p10.south);
        \draw[Latex-] (p10.south) to ++(0,-1pt);
    \end{tikzpicture}}
    \caption{Random access data structures for the reference from the proof of \cref{lem:compressor_or_inverter}. An arrow from a text position to a reference position indicates that the reference position is read by the random access data structure when simulating access to the text position.}
    \label{fig:enter-label}
\end{figure}

\begin{proof}
\let\tau\tauconst
Let $m = N\cdot M = \Theta(N^{1/(1 - \tau)} \cdot \log N)$,
which implies $M = o(m^\tau \log m)$, as well as $m^\tau = \Theta(N^{\tau/(1 - \tau)} \cdot \log^\tau N)$ and hence $m^\tau = o(M)$.
We split~$S$ into strings $A_1, \dots, A_N \in \{0,1\}^M$ defined by $\forall i \in [1 \dd N] : {A_i = {S((i - 1)\cdot M \dd i \cdot M]}}$.
Throughout the proof, we focus on the reference $R[1\dd m] = A_{f(1)}\cdot A_{f(2)}\cdot \ldots \cdot A_{f(N)}$.
We will use a random access data structure for~$R$ (and some for now undefined text) to derive either a compression scheme or a function inversion protocol.

\paragraph{Compression scheme.} Our goal is to find some $i \in [1 \dd N]$ such that we can effectively replace~$A_i$ with the random access data structure for reference~$R$ and text~$A_i$.
We run the following compression algorithm for each $i \in [1 \dd N]$ and terminate as soon as the algorithm is successful.
We preprocess the reference $R$, and then process $A_i$ as the text, resulting in a random access data structure of size $\cO(m^\tau) \subseteq o(M)$ bits.
The compressed representation of the string set consists of this data structure, the value~$i$ stored in $\cO(\log N)$ bits,
and the strings $A_1, \dots, A_{i - 1}, A_{i + 1}, \dots, A_N$ using $NM - M$ bits. Due to $\log_2 N = \Theta(\log M) \subset o(M)$, the overall size is $NM - M + o(M)$ bits, as claimed.

It remains to verify if the compressed representation is indeed sufficient for restoring $A_i$. 
We first create the string $R' = B_{f(1)}\cdot B_{f(2)} \cdot \ldots \cdot B_{f(N)}$, where $B_{f(j)} = 0^M$ if $f(j) = i$, and otherwise $B_{f(j)} = A_{f(j)}$. (This is straightforward using~$f$ and $A_1, \dots, A_{i - 1}, A_{i + 1}, \dots, A_N$.)
We then use the random access data structure for text~$A_i$ to simulate reading the entire~$A_i$, performing random accesses to~$R'$ rather than~$R$.
This correctly restores~$A_i$, unless one of the following holds. 
First, if~$A_i$ is not a substring of~$R$ (which implies that~$i$ is not in the image of $f$), then the random access data structure may not provide access to the entire~$A_i$.
Second, we may access some position $R'[h] = B_{f(\ceil{h/M})}[1 + h \bmod M]$ with $f(\ceil{h/M}) = i$, which can intuitively be seen as trying to access one of the occurrences of $A_i$ in~$R$.
Both cases can readily be detected.

If the decompression fails for all $i \in [1 \dd N]$, then we derive a function inversion protocol instead.

\paragraph{Function inversion protocol.} 
We preprocess the reference~$R$ (without processing any text), which 
results in a data structure of size $\cO(m^\tau) \subseteq o(M)$ bits.
On arrival of the first query $i \in [1\dd N]$ (asking for some $j \in [1\dd N]$ with $f(j) = i$), we build the random access data structure for the text~$A_i$. Then, we use the data structure to simulate read access to each position of $a_i$. Since the compression scheme failed, it is clear that (while simulating read access) we either access some position $R[h]$ with $f(\ceil{h/M}) = i$, or~$i$ is not in the image of~$f$. Hence, as soon as we access some $R[h]$ with $f(\ceil{h/M}) = i$, we answer the query with $j = \ceil{h/M}$. If no such position $R[h]$ is accessed, then we instead answer that~$i$ is not in the image of~$f$.
Processing the text $A_i$ and performing~$M$ queries takes $\tO(m^{\tau + \delta} + M \cdot m^\delta) = \tO(N^{(\tau + \delta) / (1 - \tau)})$ time, as claimed.

If $M = \tO(N^{(\tau + \delta) / (1 - \tau)})$, then we first copy the data structure before running the query.
Otherwise, we explicitly keep track of the changes made to the random access data structure during the query execution. Before returning the result, we revert all changes by either restoring the copy or undoing the explicitly tracked changes (which does not affect the asymptotic time complexity). 
Hence, at the beginning of each query, the random access data structure is in the same state as right after preprocessing the reference and before processing any text. 
\end{proof}

The compressed string representation from \cref{lem:compressor_or_inverter} consists of $NM - M + o(M)$ bits. 
Hence, for a fixed function $f$, it can distinguish only $2^{NM - M + o(M)}$ different bitstrings of length $NM$.

\begin{observation}\label{lem:how_many_strings_compressible_per_function}
Let $0 < \tau < 1$ be a constant. For $N \in \mathbb Z_{> 0}$, let $M = \ceil{N^{\tau/(1 - \tau)} \cdot \log_2 N}$ and consider a fixed function $f : [N] \rightarrow [N]$.
Less than $2^{NM - M + o(M)}$ bitstrings of length $NM$ can be compressed using~$f$ in the sense of \cref{lem:compressor_or_inverter}.
\end{observation}

\begin{lemma}\label{lem:how_many_strings_sets_to_invert_any_function}
Let $0 < \tau < 1$ be constant. For $N \in \mathbb Z_{> 0}$, let $M = \ceil{N^{\tau/(1 - \tau)} \cdot \log_2 N}$.
There is a set $\mathcal A \subseteq \{0,1\}^{NM}$ of size $\absolute{\mathcal A} = \cO(N \log N)$ such that, for every function $f : [N] \rightarrow [N]$, there is at least one string in $\mathcal A$ that implies an inversion protocol for~$f$ in the sense of \cref{lem:compressor_or_inverter}.
\end{lemma}

\begin{proof}
Let $\mathcal F$ be the set of all $[N]$ to $[N]$ functions, and let $\mathcal S = \{0,1\}^{NM}$.
We will show that, for $i \in [0 \dd \ceil{\log_2 \absolute{\mathcal F}} + 1]$, there are $X_i \subseteq \mathcal S$ and $Y_i \subseteq \mathcal F$ with $x_i = \absolute{X_i} = 2^{NM} - i$ and $y_i = \absolute{Y_i} \leq \absolute{\mathcal F} / 2^i$ such that, for every function from $\mathcal F \setminus Y_i$, there is at least one string in $\mathcal S \setminus X_i$ that implies an inversion protocol in the sense of \cref{lem:compressor_or_inverter}. Then, using $\mathcal A = \mathcal S \setminus X_{\ceil{\log_2 \absolute{\mathcal F}} + 1}$ satisfies the claim.

Trivially, the statement holds for $i = 0$ using $Y_0 = \mathcal F$.
Inductively assume that the statement holds for some $i \in [0 \dd \ceil{\log_2 \absolute{\mathcal F}}]$.
Now we show that it holds for $i + 1$.
We will show that there is a string in $X_i$ that can be used to derive an inversion protocol for at least half of the functions in~$Y_i$ in the sense of \cref{lem:compressor_or_inverter}. Then, $X_{i + 1}$ is obtained by removing this string from $X_i$, and $Y_{i + 1}$ is obtained by removing all the invertible functions from $Y_i$.

We model the input space using the complete bipartite graph $K_{x_i, y_i}$ (with strings from $X_i$ on the left side, and functions from $Y_i$ on the right side).
There are $x_i \cdot y_i$ edges, which we orient as follows. An edge goes from right to left if and only if its incident function can be used to compress its incident string using \cref{lem:compressor_or_inverter} (and conversely, if the edge goes from left to right, then the string can be used to invert the function).
By \cref{lem:how_many_strings_compressible_per_function}, there are at most $y_i \cdot 2^{NM - M + o(M)}$ edges from right to left. It holds $i \leq \ceil{\log_2 \absolute{\mathcal F}} = \ceil{N \log_2 N}$, which implies $x_i = 2^{NM} - i \geq 2^{NM - 1}$. Thus, there are at most $y_i \cdot 2^{NM - M + o(M)} < x_i \cdot y_i / 2$ edges from right to left, and at least $x_i \cdot y_i / 2$ edges from left to right.
Hence there must be a string in $X_i$ that has at least $y_i / 2$ outgoing edges, i.e., it can be used to derive an inversion protocol for at least half the functions in~$Y_i$.
\end{proof}

\randomaccessimpliesnonuniforminversion
\begin{proof}
\let\tau\tauconst
Let $M = \ceil{N^{\tau/(1 - \tau)} \cdot \log_2 N}$. 
Since the function inversion data structure will be non-uniform, we can afford to embed the $\cO(N \log N)$ strings from \cref{lem:how_many_strings_sets_to_invert_any_function} into the algorithm.
Now we describe how to obtain the inversion data structure for a given function~$f$.
We try the embedded strings one at a time. For each string, we run the 
algorithm from \cref{lem:compressor_or_inverter} to obtain a function inversion data structure. We verify the correctness of the data structure by evaluating $f(x) = y$ for every $x \in [1\dd N]$. When queried with~$y$, the data structure either reports some alleged value $x' \in f^{-1}(y)$, or it reports that~$y$ is not in the image of~$f$. In the latter case, we know that the inversion protocol is incorrect. In the former case, we verify that actually $f(x') = y$.
In this way, we verify that the entire image of~$f$ is correctly inverted.
It remains to check if all non-image elements of the co-domain are handled correctly. Hence, we try all $y \in [1\dd N]$ and try to invert them. Whenever the data structure reports some alleged $x' \in f^{-1}(y)$, we verify that $f(x') = y$.
If we discover that the data structure reports an incorrect result, we immediately abort and continue with the next embedded string. 
Once we find a string that yields a correct inversion data structure, we report this data structure and the index of the string. Since there are only $\tO(N)$ strings, we can afford the $\cO(\log N) \subset \tO(1)$ bits needed to store the index.
\end{proof}

We first derive a conditional lower bound for the suffix access random data structure.

\raofflinelowercond
\begin{proof}
\let\tau\tauconst
If $\delta < 3 - 7\tau$, then \cref{lem:random_access_implies_non_uniform_inversion} implies that there is a non-uniform function inversion protocol that, up to polylogarithmic factors, uses $S = N^{\tau / (1- \tau)}$ space and $T = N^{(\delta + \tau)/(1-\tau)}$ time. We obtain $S^3 T = N^{(4\tau + \delta)/(1-\tau)} = o(N^3)$, improving over Fiat and Naor's function inversion data structure~\cite{doi:10.1137/S0097539795280512}. We further note that Fiat and Naor's bounds only make sense for $S \ge N^{2/3}$ (otherwise, $T \ge N$, which is trivial), and therefore our lower bound holds for $0.4 \le \tau < 1$. 
\end{proof}

Secondly, we show a non-conditional lower bound.

\raofflinelower

\begin{proof}
\let\tau\tauconst
If $\delta < 1 - 3\tau$, then \cref{lem:random_access_implies_non_uniform_inversion} implies that there is a function inversion protocol that, up to polylogarithmic factors, uses $S = N^{\tau / (1- \tau)}$ space and $T = N^{(\delta + \tau)/(1-\tau)} \ll N^{1 - \tau / (1- \tau)} = N/S$ time. However, then $S \cdot T \ll N$ contradicts the well-known lower bound for inverting a permutation~\cite{DBLP:conf/stoc/Yao90,DBLP:journals/eccc/DeTT09}, which holds for non-uniform algorithms. (While the lower bounds in \cite{DBLP:conf/stoc/Yao90,DBLP:journals/eccc/DeTT09} do not explicitly mention non-uniformity, the bound in \cite[Corollary~10.3]{DBLP:journals/eccc/DeTT09} is achieved by obtaining a compression scheme from the function inversion protocol. The compressibility argument used, see \cite[Fact~10.1]{DBLP:journals/eccc/DeTT09}, is based on the number of distinct inputs distinguishable by a compressed representation of a given size. It does not depend on whether the compression algorithm is non-uniform; thus, the lower bound holds for non-uniform algorithms.
\end{proof}

The above results hold only for relatively short texts. However, in the streaming setting, we can show that similar lower bounds hold for texts of arbitrary length.

\begin{corollary}\label{lem:ra_streaming_lower}
\let\tau\tauconst
Let $\eps, \tau$ with $0 < \tau < 1$ be constant.
Assume that there is a streaming random access data structure with space complexity $\cO(m^\tau)$. For all $n\geq m^\tau \log m$ and all constant $\varepsilon > 0$, there is a length-$n$ text $T$ such that both of the following hold:
\begin{enumerate}
    \item \label{it:streaming_lb_cond} If the amortized query time is $\cO(Q/m^{\varepsilon})$, where $Q = m^{3-7\tau}$ and $0.4 \le \tau < 1$, then the worst-case update time is $\tilde\Omega(Q)$, unless there is a non-uniform function inversion data structure that improves over the time-space trade-off of Fiat and Noar.
    \item \label{it:streaming_lb} If the amortized query time is $\cO(Q / m^{\varepsilon})$, where $Q = m^{1-3\tau}$ and $0 < \tau < 1$, then the worst-case update time is $\tilde\Omega(Q)$.
\end{enumerate}
\end{corollary}

\begin{proof}
    \let\tau\tauconst
    Let $\mu = \floor{m^\tau \log m}$. We actually show a stronger claim: in both cases, every $2\mu$ updates take $\tilde\Omega(Q \mu)$ time. We construct the text $T$ in blocks of length $\mu$.
    Assume that, for some integer $i > 0$, we have already constructed $T[1\dd i\mu]$ with the claimed property.
    Now assume that there is no string $S[1\dd \mu]$ such that updating~$T$ by appending~$S$ takes $\tilde\Omega(Q \mu)$ time.
    Then everything done so far can be seen as a preprocessing for the \emph{offline} construction of a random access data structure. In particular, given any text $T'$ of maximum length $\mu$, we can now construct a random access data structure in $\tO(Q\mu)$ time. (This holds because a random access data structure for $T[1\dd i\mu] \cdot T'$ is trivially also a random access data structure for $T'$.)
    However, this contradicts \cref{lem:ra_offline_lower_cond} when $Q = m^{3-7\tau}$ and \cref{lem:ra_offline_lower} for $Q = m^{1-3\tau}$.
    Hence, we can construct an arbitrarily long string~$T$ that can be split into blocks of length $\mu$ such that every block requires $\tilde\Omega(Q \mu)$ update time.
    The claim then immediately follows because every length-$2\mu$ fragment of~$T$ contains one of the blocks.
\end{proof}

Finally, an unconditional lower bound higher than ours by a factor of $\tau$ would imply a breakthrough in the long-standing unconditional lower bound for function inversion. The known bound states that, for any inversion data structure for a constant-time computable function $f : [N] \rightarrow [N]$, the product of space and query time must be $\Omega(N)$~\cite{DBLP:conf/stoc/Yao90,DBLP:journals/eccc/DeTT09}.

\begin{restatable}{corollary}{nobetterlowercond}\label{lem:nobetterlowercond} 
Let $0 < \eps < 1$ and let $\tau > 0$ be an integer parameter. Assume that, for every 0-error streaming random access data structure with space complexity $\tO(\tau)$ bits and constant query time, the worst-case update time is $\tilde\Omega(m / \tau^{2-\epsilon})$. Then, every $\tau$-space function inversion data structure for a constant-time computable function $f : [N] \rightarrow [N]$ must have worst-case query time $\tilde\Omega(N/\tau^{1-\epsilon})$, and hence the product of space and query time is $\tilde\Omega(N \tau^{\epsilon})$.
\end{restatable}
\begin{proof}
Assume that there is a $\tau$-space function inversion data structure for $f$ with both query and update time in $\tO(N/\tau^{1-\epsilon})$. 
Applying the construction of \cref{lem:reduce_ra_to_fi} with $x = 1$ and $y = 1-\epsilon$, we obtain
a random access data structure with space complexity $\tO(\tau)$ bits, worst-case query time $\cO(1)$, and worst-case update time $\tO(m / \tau^{2-\epsilon})$.
Therefore, a $\tilde\Omega(m / \tau^{2-\epsilon})$ lower bound on the update time of the random access data structure directly translates to a $\tilde\Omega(N/\tau^{1-\epsilon})$ lower bound on the function inversion time.
\end{proof}

\section{Application: Asymmetric Streaming Pattern Matching}

In a pattern matching problem, we are given a pattern $P[1\dd m]$ and a text $T[1\dd n]$ with the task of computing the ending positions of fragments of $T$ that \emph{match} $P$ under some matching relation; we refer to these positions as the \emph{occurrences} of $P$ in $T$.\footnote{Our reduction also applies to the variant where one wishes to compute all fragments instead of their ending positions. Note, however, that the set of such fragments might be much larger than the set of ending positions; for example, this can be the case in approximate pattern matching under the edit distance. Additionally, some extra problem-specific work would be needed to avoid double-reporting in said variant.}
In the $k$-approximate variant of pattern matching under a given matching relation, we are looking for the ending positions of the fragments of $T$ that are at distance at most $k$ from a string matching~$P$; these positions are called the \emph{$k$-approximate occurrences} of $P$ in $T$ (under the corresponding matching relation).
Below, we focus on the two most popular string distances, namely the Hamming and edit distances, where the edit distance can be weighted; that is, each edit operation may have a weight dependent on the involved symbols, which can be accessed in constant time. 

In the online read-only setting, one is given read-only access to a pattern $P$ and to a text $T$ that arrives online symbol by symbol. The goal is to output each integer $j$ for which there exists an $i$ such that $T[i\dd j]$ matches $P$, and to do so before symbol $T[j+1]$ arrives.

We next show that streaming random access data structures enable black-box reductions from several fundamental pattern matching problems in the asymmetric streaming setting to their counterparts in the online read-only setting.
Namely, we introduce and consider the following class of \emph{generic pattern matching} problems:


\defgenericpatternmatching%
\noindent This class includes several classical pattern matching problems. We provide some of them below:
\begin{itemize}
\item The standard pattern matching problem is $1$-generic.
\item The circular pattern matching problem~\cite{10.1007/978-3-319-02309-0_59,FREDRIKSSON2009579,10.1093/comjnl/bxt023,Iliopoulos2017,10.1007/978-3-642-25591-5_69,10.1007/978-3-642-38905-4_15,DBLP:journals/jcss/Charalampopoulos21,DBLP:conf/esa/Charalampopoulos22,keditCPM,DBLP:conf/cpm/BathieCS24}, in which a string~$S$ matches~$P$ if $S$ equals a \emph{rotation} $P(j \dd ] P[\dd j]$ of $P$ for some $j \in [1 \dd m]$, is $2$-generic.
\item One of the variants of elastic-degenerate string matching~\cite{grossi_et_al:LIPIcs.CPM.2017.9,aoyama_et_al:LIPIcs.CPM.2018.9,doi:10.1137/20M1368033,BERNARDINI2020109,DBLP:conf/biostec/ProchazkaCK021,ILIOPOULOS2021104616,DBLP:journals/mst/BernardiniGPSSZ24,pissis_et_al:LIPIcs.CPM.2025.28,gawrychowski_et_al:LIPIcs.CPM.2025.29,GABORY2025105296}.
A symbol $S[i]$ of an elastic-degenerate string $S$ is a finite subset of $\Sigma^\ast$, where $\Sigma$ is the alphabet.
An elastic-degenerate string generates a language $L(S) = \{S_1 S_2 \cdots S_m : S_i \in S[i]\}$.
We say that a regular string $Q \in \Sigma^\ast$ and an elastic-degenerate string $S$ match if $Q \in L(S)$.
If only the pattern is elastic-degenerate~\cite{ILIOPOULOS2021104616,GABORY2025105296} and contains $z$ elastic-degenerate symbols, i.e., only $z$ symbols contain more than a single string, then the problem is $(2z+1)$-generic.
\item String matching with variable-length gaps~\cite{Lee03,DBLP:journals/jcb/MorgantePVZ05,DBLP:conf/cocoon/RahmanILMS06,DBLP:journals/ir/FredrikssonG08,DBLP:conf/lata/FredrikssonG09,10.5555/643002.643006,BILLE201225,AMIR201534,
10.1007/978-3-642-20662-7_7,DBLP:journals/algorithmica/AmirKLPPS19}, where the pattern is of the form $P = P_1 \cdot g\{a_1,b_1\} \cdot P_2 \cdot g\{a_2,b_2\} \cdots g\{a_{r-1},b_{r-1}\} \cdot P_r$. Each $P_i$ is a string over $\Sigma$ that only matches itself, whereas each $g\{a_i,b_i\}$ is a wildcard string that matches any string of length in the interval $[a_i \dd b_i]$.
This problem is $(r + \sum_i b_i)$-generic.
\end{itemize}
We note that the $k$-approximate variants of all the above problems are generic pattern problems as well:

\begin{observation}
If a pattern matching problem $\mathcal{P}$ is $z$-generic, then its $k$-approximate variant under either the Hamming distance or the (weighted) edit distance is $(z+2k)$-generic. 
\end{observation}

We now present a black-box transformation of an online read-only algorithm for a generic pattern matching problem to an asymmetric streaming algorithm.

\begin{theorem}\label{th:blackbox_ro_to_as}
Consider integers $z \geq 0$ and $\tau > 0$.
Suppose that, for a $z$-generic pattern matching problem with pattern $P[1\dd m]$, there exists an online read-only algorithm $\mathcal{A}$ that uses $\cO(\tau)$ additional space and processes a text $T[1\dd n]$ in $g(m, \tau)$ worst-case (or amortized) time per symbol, where $f$ and $g$ are functions that are non-decreasing in their parameters.
    
    Suppose that there exists an $\cO(\tau)$-space streaming $z$-error random access data structure 
    with worst-case update time $c(m,\tau)$ and worst-case query time $q(m,\tau)$.
    
    Then, there is an asymmetric streaming algorithm that, given $\tau$, random access to $P$, and streaming access to $T$, 
    solves the given $z$-generic pattern matching problem for $P$ and $T$ using $\cO(\tau)$ additional space and $\cO(c(m,\tau) + g(m, \tau) \cdot q(m, \tau))$ time per symbol of $T$;
    this time complexity is worst-case (amortized) if $\mathcal{A}$ processes each symbol in $g(m, \tau)$ worst-case (resp.~amortized) time.
\end{theorem}
\begin{proof}
	If $\tau \geq n$, we can simply store the text and run algorithm $\mathcal{A}$ within the claimed bounds.
	We thus henceforth assume that $\tau < n$.

	As the text $T$ arrives, we conceptually decompose it into blocks $T[f_i\dd f_{i + 1})$, where each block is a right-maximal fragment supported by a $z$-error random access data structure.
	Note that we can assume that each block (apart, possibly, from the last one) is of length at least $\tau$, as we can maintain random access to any length-$\tau$ fragment of $T$ explicitly within our space budget by copying its symbols instead.
	At any moment, we store $z$-error random access data structures (and at most~$\tau$ symbols) for the rightmost two blocks we have encountered so far (more precisely, for the second rightmost block if one exists and for the seen prefix of the rightmost block).
	We first explain how these data structures enable us to efficiently solve the $z$-generic pattern matching problem at hand, and then describe how to efficiently maintain them.
	
	Consider a block $T[f_i\dd f_{i + 1})$.
	When we read a symbol $T[f_i]$ and start maintaining the $z$-error random access data structure for this block, we take $\cO(\tau)$ time to create a copy of the information precomputed for $P$ by algorithm $\mathcal{A}$.
	This comes at no extra asymptotic cost due to our assumption that each block has length at least $\tau$.
	(We do this as the pattern matching algorithm for processing the text may edit this information, and we want to have the original version available when processing subsequent blocks.)
	We use this copy of precomputed information to run $\mathcal{A}$ for $P$ and $T[f_i\dd f_{i + 2})$ (or $T[f_i\dd f_{i + 1})$ if this is the last block) as the symbols of this fragment of $T$ arrive.
	Our maintained data structures enable $\cO(q(m,\tau))$-time random access to $T[f_i\dd f_{i + 2})$ and hence $\mathcal{A}$ can be implemented to work in time $\cO((f_{i-2}-f_i) \cdot g(m, \tau) \cdot q(m,\tau))$.
	In the end, we free any memory used by the pattern matching algorithm in $\cO(\tau)$ time, including the (possibly edited) copy of the information precomputed for $P$.
	Over all blocks, this takes a total of $\cO(n \cdot g(m, \tau) \cdot q(m,\tau))$ time, which is within the claimed bound.
	To prove that the occurrences of $P$ are computed correctly, it suffices to prove that any fragment of $T$ that can be decomposed to at most $z+1$ substrings of~$P$ and at most $z$ symbols spans at most two blocks.
	Indeed, if this were not the case, then there would exist a fragment $T[a \dd b)$ and a block $T[f_i\dd f_{i + 1})$ with $a < f_i < f_{i+1} < b$, contradicting the right-maximality of $T[f_i\dd f_{i + 1})$, as $T[f_i \dd b)$ can be factorized into at most $z+1$ substrings of $P$ and $z$ symbols.
	Finally, note that, at any time, there are at most two active instances of $\mathcal{A}$; e.g., before receiving symbol $T[f_{3}]$, we are running the algorithm $\mathcal{A}$ for texts $T[f_1 \dd f_3)$ and $T[f_2 \dd f_3)$.
	We can avoid reporting duplicates by using a single variable that records the most recently reported occurrence.
	
	Let us now describe how to efficiently maintain the $z$-error random data structures for the two rightmost blocks.
	We process each block $T[f_i\dd f_{i + 1})$ as follows, starting with the block $T[f_1 \dd f_2)$, where $f_1 = 1$.
	We maintain two instances of the $z$-error random access data structure for the considered block, updating one after the other for each arriving symbol.
	When a symbol $T[j]$ arrives, we first update the first instance of the data structure.
	If the support-length is incremented, we also update the second instance of the data structure.
	Otherwise, we conclude that $j = f_{i+1}$, delete the first instance of the data structure, stop updating the second instance, and start processing block $T[f_{i + 1}\dd f_{i + 2})$.
	At the same time, we also delete the $z$-error random access data structure stored for block $T[f_{i - 1}\dd f_i)$, if one exists.
	We apply the same trick as above, that is, create copies of the precomputed information for $P$ that enables the efficient construction of $z$-error random access data structures for each block, at no extra asymptotic cost.
	It is then clear that the maintenance of the random access data structures is achieved within the claimed bounds.
\end{proof}

\begin{corollary}\label{cor:blackbox_ro_to_as}
Consider integers $z \geq 0$ and $\tau > 0$.
Suppose that, for a $z$-generic pattern matching problem with pattern $P[1\dd m]$, there exists an online read-only algorithm $\mathcal{A}$ that uses $\cO(\tau)$ additional space and processes a text $T[1\dd n]$ in $g(m, \tau)$ worst-case (or amortized) time per symbol, where $f$ and $g$ are functions that are non-decreasing in their parameters.

Then, there is an asymmetric streaming algorithm that, given $\tau$, random access to $P$, and streaming access to $T$, solves the given $z$-generic pattern matching problem for $P$ and $T$ using $\cO(\tau)$ additional space and $\cO(k^3 m^3/\tau^6 + g(m, \tau) \cdot (1+\log k))$ time per symbol of $T$; this time complexity is worst-case (amortized) if $\mathcal{A}$ processes each symbol in $g(m, \tau)$ worst-case (resp.~amortized) time.
\end{corollary}
\begin{proof}
The claim follows in a straightforward manner via the black-box reduction encapsulated in \cref{th:blackbox_ro_to_as} and our random access data structure from \cref{cor:streaming_kerror}. 
\end{proof}

Our reduction applies to the entire class of $z$-generic pattern matching problems, thus providing a general framework for transforming read-only online algorithms into asymmetric streaming ones. Below, we give an example of an application of the reduction to standard and circular pattern matching problems.

\begin{corollary}\label{fact:pm}
	Suppose that we are given positive integers $\tau$ and $k$, read-only access to a pattern~$P$ of length $m$, and a streaming text $T$ of length $n$.
	After preprocessing $P$, we can deterministically
\begin{enumerate}
\item compute all exact occurrences of $P$ in $T$ using space $\cO(\tau)$ and worst-case time $\tO(1+m^3/\tau^6)$ per symbol of $T$, and
\item compute all $k$-approximate occurrences of $P$ in $T$ under the Hamming distance using space $\cO(k \log m + \tau)$ and worst-case time $\tO(k + k^3m^3/\tau^6)$ per symbol of $T$, and
\item compute all $k$-approximate occurrences of $P$ in $T$ under the edit distance using space $\tO(k^4 + \tau)$ and amortized time $\tO(k^4 + k^3m^3/\tau^6)$ per symbol of $T$.
\end{enumerate}
\end{corollary}
\begin{proof}
We apply \cref{cor:blackbox_ro_to_as} to 
(a) \cite{DBLP:journals/talg/BreslauerG14} for the exact variant, (b) \cite[Theorem 4.5]{DBLP:conf/isaac/BathieKS23} for the $k$-approximate variant under the Hamming distance, and (c) \cite[Lemma 5.8]{DBLP:conf/isaac/BathieKS23} for the $k$-approximate variant under the edit distance.
\end{proof}

\begin{table}[t]\small
\caption{Complexities of exact and approximate (standard) pattern matching for a pattern of length~$m$.
Here, $\tau$ is an integer parameter and time complexities are given per symbol of the text.
The complexities in the asymmetric streaming model (marked with $^{\dagger}$) are shown in this work and are \emph{deterministic} (apart from Las Vegas randomization in the preprocessing of the pattern).
In contrast, complexities in the streaming setting are \emph{Monte Carlo randomized}.}
\label{tab:pm}
\setlength{\tabcolsep}{.76em}
\centering
\begin{tabular}{|l|l|l|l|}
\hline
\textbf{} & \textbf{Online read-only} & \textbf{Asym.~streaming} & \textbf{Streaming} \\
\textbf{} & (deterministic, space used on  & (deterministic, space used & (randomized) \\
\textbf{} & top of read-only $P$ \emph{and} $T$) &  on top of read-only $P$) & \\ 
\hline
\hline
Exact & 
\makecell[l]{$\cO(1)$ space\\ $\cO(1)$ time~\cite{DBLP:journals/talg/BreslauerG14}} & 
\makecell[l]{$\cO(\tau)$ space\\ $\tO(1+m^3/\tau^6)$ time$^{\dagger}$}& 
\makecell[l]{$\cO(\log m)$ space\\ $\cO(\log m)$ time~\cite{porat2009optimal}} \\
\hline
$k$-mism. & 
\makecell[l]{$\cO(k \log m)$ space\\ $\cO(k)$ time~\cite{DBLP:conf/isaac/BathieKS23}} & 
\makecell[l]{$\cO(k \log m + \tau)$ space\\ $\tO(k + k^3m^3/\tau^6)$ time$^{\dagger}$} & 
\makecell[l]{$\cO(k \log m)$ space\\ $\tO(\sqrt{k})$ time~\cite{clifford2018streaming}} \\
\hline
$k$-edit & 
\makecell[l]{$\tO(k^4)$ space\\ $\tO(k^4)$ (amort.) time~\cite{DBLP:conf/isaac/BathieKS23}} & 
\makecell[l]{$\tO(k^4 + \tau)$ space\\ $\tO(k^4 + k^3m^3/\tau^6)$ (amort.) time$^{\dagger}$} & 
\makecell[l]{$\tO(k^2)$ space\\ $\tO(k^2)$ time~\cite{DBLP:conf/icalp/Bhattacharya023}}\\
\hline
\end{tabular}
\end{table}

We summarize the results for exact and approximate (standard) pattern matching in \cref{tab:pm}.\footnote{There is an unpublished work~\cite{mai2021optimalspacetimestreaming} on asymmetric streaming pattern matching. However, this work considers the problem where one is given random access to the text while the pattern arrives as a stream. Hence, this contribution is incomparable to ours.}
We note that the bounds for the fully streaming model carry straightforwardly to the asymmetric streaming one. All streaming algorithms, however, critically rely on Monte Carlo randomization.
In contrast, the output of our algorithms for the asymmetric streaming model is always correct and only the preprocessing of the pattern is Las Vegas randomized; for $\tau \ge m^{0.4}$ and $k = m^{o(1)}$, our algorithms' time-space product is sublinear.
This poses the challenge of exploring the trade-off between algorithmic efficiency and randomization in asymmetric streaming pattern matching.

We now turn our attention to circular pattern matching.
Bathie, Charalampopoulos, and Starikovskaya~\cite{DBLP:conf/cpm/BathieCS24}, showed that, for every $\tau \in [\sqrt m \log m (\log \log m)^3 \dd m]$, there exists an asymmetric streaming algorithm that solves this problem in $\cO(\tau)$ space and $\tO(m/\tau)$ time per symbol. 
We next obtain a significantly better trade-off; notably, for $\tau \ge m^{0.4}$, the time-space product of the algorithm encapsulated in \cref{fact:cpm} is sublinear.

\begin{corollary}\label{fact:cpm}
Let $\tau > 0$ be an integer parameter. There is a deterministic asymmetric streaming algorithm solving  the exact circular pattern matching problem on a read-only pattern of length $m$ and a streaming text of length $n$ in $\tO(\tau)$ space and $\tO(m^3/\tau^6+1)$ time per symbol of the text.
\end{corollary}
\begin{proof}
The claim follows via the black-box reduction (\cref{th:blackbox_ro_to_as}), using the data structure from \cref{cor:streaming_kerror} and the algorithm from~\cite[Theorem 7.1]{DBLP:conf/cpm/BathieCS24}.
\end{proof}

\section{Application: Relative Lempel--Ziv Factorization}
Finally, we show an application of suffix random access data structures to designing asymmetric streaming compression algorithms, namely, for computing the relative Lempel--Ziv factorization. Formally, the relative Lempel--Ziv factorization is defined as follows:

\begin{restatable}{definition}{rlzdef}[{\cite{10.1007/978-3-642-16321-0_20}}]
Let strings $T, R \in \Sigma^*$. We say that $T = f_1 f_2 \ldots f_z$ is the Lempel--Ziv factorization of $T$ relative to $R$, $\rlz(T,R)$, if for every $1 \le i \le z$, $f_i$ is either a symbol that does not occur in $R$, or the longest substring of $R$ that is a prefix of $f_i f_{i+1} \ldots f_z$. We call $f_1, f_2, \ldots, f_z$ the factors of the factorization. 
\end{restatable}

This method is particularly effective for compressing collections of strings that are highly similar to a reference. For this reason, the relative Lempel--Ziv factorization, when combined with simple auxiliary data structures, has become a popular text index for collections of complete genome sequences from individuals of the same species~\cite{10.1093/bioinformatics/btr505,10.1007/978-3-642-16321-0_20,DBLP:journals/bioinformatics/DeorowiczDG13,DBLP:journals/bmcgenomics/ValenzuelaNVPM18,DBLP:journals/bioinformatics/NavarroSMG19}. However, no space-efficient algorithms for its construction were known, a gap that we address in this section.

\subsection{Upper Bound}
We call a factorization of a string $T$ into phrases that are symbols and fragments of a string $R$ an \emph{$R$-factorization of $T$}. Throughout this section we work with a length-$n$ string~$T$ and a reference length-$m$ string $R$. Let $z = |\rlz(T,R)|$. We first show how to compute an $R$-factorization of~$T$ that consists of $\cO(z \log^2 m)$ phrases. In particular, this factorization satisfies the following stronger optimality property (an analogue of this definition for the LZ77 factorization was considered in~\cite{DBLP:conf/esa/0001GGK15}):

\begin{definition}[see {\cite[Definition 12]{DBLP:conf/esa/0001GGK15}}]\label{def:aoptimal}
	Consider strings $T$ and $R$.
    We call an $R$-factorization $T = f_1 f_2 \cdots f_z$
    \emph{$\alpha$-optimal} if, for all $i\in [0 \dd z-\alpha]$, the fragment $\bigodot_{j=1}^{\alpha} f_{i+j}$
    (that is the concatenation of $\alpha$ consecutive phrases) is not a substring of $R$.
\end{definition}

The following observation is immediate:

\begin{observation}\label{obs:opt_small}
An $\alpha$-optimal $R$-factorization of a string $T$ consists of at most $\alpha z$ phrases, where $z = |\rlz(T,R)|$.
\end{observation}

Let us describe an $R$-factorization of $T$ that, as we prove, is $\cO(\log^2 m$)-optimal.
We show how to efficiently compute this factorization later.
This $R$-factorization is defined for a positive integer threshold~$\Delta$.
The optimality guarantees of the factorization are insensitive to $\Delta$, but the complexities for computing this factorization in the asymmetric streaming setting depend on $\Delta$.

\paragraph*{The $\Delta$-core $R$-factorization.}
Conceptually, we factorize the string greedily, in a left-to-right fashion, into \emph{meta-phrases} and then partition each meta-phrase to phrases that comprise our factorization.
Consider a position $i$ of~$T$ that is either $1$ or such that a meta-phrase ends at a position~$i-1$. Consider the following cases:
\begin{enumerate}
\item  If a core-matching query with an input $T[i \dd i+3 \cdot \Delta)$ on the reference $T$ is undefined or unsuccessful, we designate a fragment $T[i \dd \min(i+\Delta,n+1))$ as a \emph{meta-phrase};
\item Otherwise, we designate a fragment $T[i\dd i + 2^{j + 1} \cdot \Delta)$ as a \emph{meta-phrase}, where $j$ is the maximal integer such that a core-matching query with an input $T[i \dd i+ 3 \cdot 2^j \cdot \Delta)$ on the reference $R$ is defined and successful.
\end{enumerate}
We further factorize each meta-phrase $T[i\dd e]$ into phrases as follows:
\begin{enumerate}
\item First, we factorize its prefix $P = T[i \dd \min(e,i+\Delta))$ into phrases according to $\rlz(P,R)$; note that $P$ can only have length smaller than~$\Delta$ only if $T[i \dd e)$ is the rightmost meta-phrase in the factorization; 
\item Secondly, if $e \ge i+\Delta$, we factorize the remaining suffix of $T[i\dd e]$ into phrases $T[{i + 2^h \cdot \Delta}\dd i + 2^{h + 1} \cdot \Delta)$ for $h \in [0\dd j]$.
\end{enumerate}
Note that each phase in a substring of $R$ by the definition of a meta-phrase and that their concatenation equals $T[i\dd e]$. We call the described factorization the \emph{$\Delta$-core $R$-factorization of $T$}. We next show that this factorization is $\cO(\log^2 m)$-optimal, which implies that it has $\cO(z \log^2 m)$ phrases due to \cref{obs:opt_small}.

\begin{lemma}\label{lem:bad_approx_ratio}
	Consider two strings $T$ and $R[1\dd m]$.
    For any positive integer $\Delta$, the $\Delta$-core $R$-factorization of $T$ is $\ceil{8\log_2^2 m + 2}$-optimal.
\end{lemma}

\begin{proof}
Note that we can assume that $m\geq 2$ as otherwise any $R$-factorization of $T$ is $2$-optimal. It suffices to prove that any fragment of $T$ that is a substring of $R$ overlaps at most $y=8\log_2^2 m$ phrases of the $\Delta$-core $R$-factorization of $T$, as this implies that said factorization is $(y+1)$-optimal.
	Clearly, it suffices to prove that this is true for 
	an arbitrary maximal such fragment $F := T[s'\dd e')$ of $T$
	(that is, we assume that
	(a) either $s' = 1$ or $T[s' -1 \dd e')$ does not occur in $R$ and
	(b) either $e' = \absolute{T} + 1$ or $T[s' \dd e' + 1)$ does not occur in $R$).
    
    We first upper-bound the number of meta-phrases that are fully contained in $F$. 
    Consider a meta-phrase ${T[s\dd e)}$ that is fully contained in $F$, that is, we have $s' \leq s$ and $e \leq e'$.
    Let $j = \floor{\log ((e' - s) / (3\Delta))}$, and observe that $T[s\dd s + 3 \cdot 2^h \cdot \Delta)$ is a prefix of $T[s\dd e')$.
    Since $T[s \dd e')$ has an occurrence in $R$, every core-matching query $T[s\dd s + 3 \cdot 2^h \cdot \Delta)$ with $h \in [0 \dd j]$ is successful, and therefore we have $e \geq s +  2^{j + 1} \cdot \Delta$. Finally,
    \[ e - s \geq 2^{j + 1} \cdot \Delta \geq \textstyle\frac13\cdot (e' - s),  \]
    that is, the meta-phrase covers at least one third of the remaining part of $F$.
    This implies that at most $\log_3 m \leq 2 \log_2 m$ meta-phrases are fully contained in $F$.
    Thus, at most $2 + 2\log_2 m \le 4 \log_2 m$ meta-phrases have an overlap with $F$ (as at most two meta-phrases can overlap $F$ but not be fully contained in~$F$).
    Each of these meta-phrases contains at most $\log_2 m$ phrases of the form
    $T[i + 2^h \cdot \Delta\dd i + 2^{h + 1} \cdot \Delta)$ for $h \in [0\dd j]$
    as the lengths of such phrases form a geometric progression with common ratio $2$.
    Further, for each meta-phrase $T[i\dd j)$, the part of its prefix $T[i \dd \min(j,i+\Delta))$ that overlaps~$F$ has an occurrence in $R$.
    Thus, for each meta-phrase that overlaps $F$, the number of phrases it contains that overlap~$F$ is at most $\log_2 m + 1 \le 2\log_2 m$.
    Therefore, $F$ overlaps at most $8\log_2^2 m$ phrases in total, concluding the proof.
\end{proof}

Before describing how to efficiently compute the $\Delta$-core $R$-factorization of $T$, we prove a simple lemma that enables us to compute the relative Lempel--Ziv factorizations of length-$\Delta$ prefixes of meta-phrases by treating them in batches.
In particular, in the proof of \cref{lem:rlz_compute_logapprox} below, we use \cref{lm:block_process} with $B$ being the concatenation of $\Theta(\tau/\Delta)$ such prefixes.

\begin{lemma}\label{lm:block_process}
Let $1 \le \tau \le m$ be an integer. Given strings $B[1 \dd \tau]$ and $R[1\dd m]$, we can compute, in $\tO(m)$ time using $\cO(\tau)$ space,
for each $i \in [1 \dd \tau]$,
a triple $(i,r,d)$ such that $R[r \dd r+d)$ is the longest substring of $R$ that occurs at position $i$ of $B$.
\end{lemma}
\begin{proof}
We sort the elements of $\Sigma$ that occur in $B$ using merge-sort and construct the suffix tree of~$B$ in $\tO(\tau)$ time and $\cO(\tau)$ space~\cite{Gusfield_1997}.
We may assume that the (first symbols of the) labels of outgoing edges of each node are stored in a binary search tree (ordered lexicographically).
We also preprocess the suffix tree of $B$ in $\cO(\tau)$ time and space for $\cO(1)$-time weighted ancestor queries~\cite{belazzougui_et_al:LIPIcs.CPM.2021.8}.
A weighted ancestor query can support, in particular, the following operation in $\cO(1)$ time: Given positions~$i, j$ of~$B$, find the (implicit or explicit) node $v$ of the suffix tree of $B$ such that the path-label of $v$ (that is, the concatenation of the string labels in the root-to-$v$ path) equals $B[i \dd j]$.

We then traverse the suffix tree of $B$ with $R$ in order to mark, for each position $r$ of $R$, the (implicit or explicit) node corresponding to the longest substring $R[r \dd r + \ell_r)$ which occurs in~$B$;
crucially, we mark it with ``$r$''.
In the interest of efficiency, during the marking procedure, we unmark a marked implicit node when a deeper implicit node in the same edge gets marked,
thus maintaining at most one implicit marked node per edge.
This marking procedure can be performed in $\tO(m)$ time as follows.
We start at the root and perform forward search, that is, we go down the tree while we can, in order to find and mark the node with path-label $R[1 \dd 1 + \ell_1)$.
Next, using a weighted ancestor query, we find the node of the tree corresponding to $R[2 \dd 1 + \ell_1)$.
We try to go down the tree from this node in order to find and mark the node with path-label $R[2 \dd 2 + \ell_2)$; and so on.
This procedure takes $\tO(m)$ time in total.
We attempt to go down in the tree with some symbol $\cO(m)$ times and each such attempt takes $\tO(1)$ time (using the binary search trees).
The number of performed weighted ancestor queries is also $\cO(m)$.

Finally, we perform a depth-first traversal of the suffix tree to compute, for each leaf $v$ of the tree, the depth $d$ of the deepest marked node on the root-to-$v$ path.
Suppose that this node is marked with~$r$ and that leaf $v$ has path-label $B[i \dd]$.
Then, the longest substring of $R$ that starts at position~$i$ of $B$ is $B[i \dd i+d) = R[r \dd r+d)$. We return triple $(i,r,d)$.
This takes $\cO(m)$ time.
\end{proof}

We are now ready to present an efficient algorithm for computing the $\Delta$-core $R$-factorization of~$T$, outputting the resulting phrases as a stream.

\begin{lemma}\label{lem:rlz_compute_logapprox}
	Consider a read-only reference string $R[1\dd m]$, a streaming string $T[1\dd n]$, and a positive integer $\tau \in [\ceil{m^{0.4}} \dd m]$.
	We can preprocess $R$ in $\poly(m)$ expected time to construct an $\tO(\tau)$-space data structure, so that
	we can then process $T$ in $\tO(z \cdot (m/\tau)^{5/3} + n \cdot  (1 + m^{3} / \tau^6) + m)$ total time using $\cOtilde(\tau)$ extra space,
	outputting, for some $\Delta \in \mathbb{Z}_{+}$, the $\Delta$-core $R$-factorization of $T$ as a stream.
	
		The preprocessing can be implemented in $\tO(m \cdot \poly(\err))$ time with success probability at least $1 - 2^{-\err}$ at the expense of increasing the space complexity and time required for processing $T$ by an $\cO(\poly(\err))$ multiplicative factor, for any integer $\err \geq \log_2 (m\sigma)$.
\end{lemma}
\begin{proof}
Let $\Delta$ be an integer in $[2 \dd 2\tau]$ that will be specified later.
	We use \cref{cor:streaming_kerror} to maintain an $\cOtilde(\tau)$-space $0$-error random access data structure for $T$ with reference $R$.
	This data structure has worst-case update time $\tO(1 + m^{3} / \tau^6)$ and query time
	$\cO(1)$; the preprocessing time of $R$ (and~$\tau$) is within the claimed bounds.
	Additionally, we employ \cref{cor:core_matching_queries} so that, by preprocessing $R$,~$\tau$, and $\Delta$ within the claimed preprocessing time bounds, we construct an $\cOtilde(\tau)$-space data structure that answers core-matching queries for texts of length $3n'$ where $n'\in \{2^\ell \cdot \Delta \mid \ell \in [0 \dd \floor{\log_2 (m/(3\Delta))}] \}$ in
	$\tO(n' ( 1 + m^{3} / ((n')^3 \tau^3)))$ time.
	
		As explained in detail below, we greedily compute the meta-phrases of the $\Delta$-core $R$-factorization of $T$ with respect to~$R$ using core-matching queries.
		We factorize the remaining prefixes of meta-phrases in batches using \cref{lm:block_process}.
	Let $r := \ceil {\tau / \Delta}$.
	
	Let us first describe and analyze our algorithm under the following assumption:
	
	\begin{assumption}\label{assumption}
	 At any point, if the suffix of $T$ that has not been factorized in meta-phrases is $T[i \dd n)$, we have constant-time random access to the prefix of $T[i \dd n)$ that we have received.
	 \end{assumption}
	 
	Suppose that we have just computed a meta-phrase of the $\Delta$-core $R$-factorization that ends at some position $i-1$ (or, if we are just now receiving symbol $T[1]$, set $i = 1$).
	We perform the following computations for the meta-phrase starting at a position $i$.
	First, we explicitly store a copy of the fragment $T[i \dd \min(i + \Delta,n+1))$ in the memory, which we will delete after we have computed another $2r$ meta-phrases.
	The total space required for such copies of fragments of $T$ of length at most $\Delta$ at any time is $\cO(r \cdot \Delta) = \cO(\tau)$ and the amortized time for creating and deleting them is $\cO(n)$ since each such fragment belongs to a distinct meta-phrase.
	We compute the meta-phrase starting at the position $i$ using core-matching queries.
	Suppose that we are in some stage of this computation and define $\ell'$ as follows.
	Let $\ell' = -1$ if we have not yet performed any core-matching queries for fragments of $T$ starting at the position $i$ and, otherwise,
	let $\ell'$ be the largest $\ell$ for which we have already successfully performed a core-matching query for $T[i \dd i + 3 \cdot 2^\ell \Delta)$.
	We then have the following cases:
	\begin{enumerate}
		\item If we lose random access to $T[i]$ before the symbol $T[i + 3 \cdot 2^{\ell'+1} \Delta - 1]$ arrives or we reach the end of $T$, we conclude that the sought meta-phrase is
	$F := T[i \dd \min ( i + 2 \cdot 2^{\ell'} \Delta , n+1))$.
		\item Else, when symbol $T[i+ 3 \cdot 2^\ell \Delta - 1]$ arrives, we issue a core-matching query for $T[i \dd i + 3 \cdot 2^\ell \Delta)$.
		\begin{enumerate}
			\item If it is unsuccessful, we again conclude that $F$ is the sought meta-phrase.
			\item Otherwise, we increment $\ell'$ and repeat this procedure.
		\end{enumerate}
	\end{enumerate}
			Note that in the process of computing a meta-phrase of length $x > \Delta$, we, in particular, also compute the partition of its length-$(x-\Delta)$ suffix into phrases via the performed core-matching queries.
	The total time complexity of the core-matching queries performed for computing a meta-phrase of length $f$ (and the partition of its length-$(f-\Delta)$ suffix into phrases) is
	$\tO(f + m^{3} / (\Delta^2 \tau^3))$.
	Over all meta-phrases, this yields a total running time of
	$\cO(n + z \cdot  m^{3} / (\Delta^2 \tau^3))$.
	
	It remains to explain how we compute $\rlz(T[i \dd \min(e,i+\Delta)),R)$ for each meta-phrase $T[i \dd e)$.
	We simply invoke \cref{lm:block_process} after every $r$ meta-phrases have been computed (and when we reach the end of $T$) to compute the relative Lempel--Ziv factorizations of such prefixes (that we have not already factorized) as a batch, noting that we have stored copies of them in memory.
	Specifically, we set $B$ equal to the concatenation of said prefixes, separated by a delimiter symbol $\$ \not\in \Sigma$.
	The total cost over all such invocations of \cref{lm:block_process} is $\tO(m \cdot (1 + z / r)) = \tO(m + z \cdot  m \Delta/\tau)$ as the total number of computed meta-phrases is $\cOtilde(z)$ due to the combination of \cref{lem:bad_approx_ratio} and \cref{obs:opt_small}.
	
	Let us now set
	$\Delta := 2\cdot \ceil{(m/\tau)^{2/3}}$ to balance $z \cdot m^{3} / (\Delta^2\tau^3) $ and $z \cdot m\Delta/\tau$.
	We then have that
	\[\tO(z \cdot m^{3} / (\Delta^2\tau^3) + z \cdot m\Delta/\tau + m) \subseteq \tO(z \cdot (m/\tau)^{5/3})\]
	and the complexity follows.
	Note that since $\tau \geq  m^{0.4}$, our constraint $\Delta \leq 2\tau$ is satisfied.
	
	We now waive \cref{assumption} under which the above algorithm was described.
	Let $\phi_1 = T[e_1 \dd e_2), \phi_2 = T[e_2 \dd e_3), \ldots, \phi_{\zeta}  = T[e_{\zeta} \dd e_{\zeta+1})$ be the meta-phrases of the $\Delta$-core $R$-factorization of $T$ ordered with respect to their starting positions.
	For each $i \in [1 \dd \zeta]$, let~$x_i$ be the position of~$T$ such that the described algorithm computes meta-phrase $\phi_i$ after receiving symbol $T[x_i+1]$ and before receiving symbol 
$T[x_i+2]$.

\begin{observation}\label{obs:consec}
For all $i \in [0 \dd \zeta]$, we have $x_i + 1 \geq e_{i+1}$.
\end{observation}

	Further, for each position $j$ of $T$, let $h(j)$ denote the support-length of the maintained instance of the $0$-error random access data structure of \cref{cor:streaming_kerror} after processing a symbol $T[j]$.
	The computation of meta-phrases guarantees that shortly prior to computing a meta-phrase, we have random-access to the unfactorized suffix of the received prefix of $T$:
	
	\begin{observation}\label{obs:just_before}
		For any $i\in [1 \dd \zeta]$, we have $x_i - \max(h(x_i),\Delta) \leq e_i$.
	\end{observation}
	
	We then have the following property.
	
	\begin{claim}\label{claim:union:two:copies}
		Suppose that we have just received (but not processed) a symbol $T[j]$ and we have already computed $\phi_0, \ldots, \phi_i$ for some $i\in [0 \dd \zeta)$.
		Then, we have
		\[ (x_{i} - h(x_{i}) \dd x_{i}] \cup
		[e_{i} \dd \min(j+1,e_{i} + \Delta)) \cup
		(j - h(j) \dd j]\cup
		[e_{i+1} \dd \min(j+1,e_{i+1} + \Delta))
		\supseteq [e_{i+1} \dd j].\]
	\end{claim}
	\begin{proof}
		\cref{obs:just_before} directly yields
		\[(x_{i} - h(x_{i}) \dd x_{i}] \cup
		[e_{i} \dd e_{i} + \Delta) \supseteq [e_i \dd x_i].\]

		Further, by the monotonicity of $y - h(y)$, we have
		$j - h(j) \leq x_{i+1} - h(x_{i+1})$.
		Another application of \cref{obs:just_before} implies that if $j - \Delta > e_{i+1}$, then $x_{i+1} - h(x_{i+1}) \leq e_{i+1}$.
		Hence,
		\[(j - h(j) \dd j]\cup
		[e_{i+1} \dd \min(j+1,e_{i+1} + \Delta)) \supseteq [e_{i+1} \dd j].\]
		
		A direct application of \cref{obs:consec} then yields the  required inclusion.
	\end{proof}
	
	We are now ready to explain how to enforce \cref{assumption} by carefully maintaining three copies of the $0$-error random access data structure of \cref{cor:streaming_kerror}.
	We maintain two basic data structures, namely
	\begin{itemize}
		\item $D_1$ which processes symbols of $T$ as they arrive, and
		\item $D_2$ which processes symbols from left-to-right, pausing when it reaches a symbol $T[x_i]$, for each $i$, and only resumes processing subsequent symbols when a phrase $\phi_{i+1}$ has been computed.  
	\end{itemize}
	Further, we maintain an auxiliary data structure $D_1'$ 
	that processes a symbol $T[j]$ just before a symbol $T[j+1]$ is received, that is, after all other computations that are to be performed for the symbol $T[j]$ have been completed.
	
	Suppose that we have just received (but not processed) symbol $T[j+1]$.
	$D_1$ and $D_1'$ provide constant-time random access to $(j-h(j) \dd j]$ and $D_2$ provides constant-time random access to $(j-h(j) \dd j]$ if we have not yet computed any meta-phrases and to $(x_{i} - h(x_{i}) \dd x_{i}]$ if we have computed exactly $i$ meta-phrases.
	Further, we also have access to
	$T[e_{i} \dd \min(j+1,e_{i} + \Delta))$ as well as to $T[e_{i+1} \dd \min(j+1,e_{i+1} + \Delta))$ if $e_{i+1}$ is defined; this is due to the explicitly stored copies of prefixes of meta-phrases.
	By \cref{claim:union:two:copies}, we thus have constant-time random access to $T[e_{i+1} \dd j]$.
	In particular, in the case when we compute $\phi_{i+1}$ for some $i$ when processing $T[j+1]$, we can update~$D_2$ as necessary using $D_1'$.
	The cost of maintaining $D_2$ and $D_1'$ is asymptotically dominated by the cost of maintaining $D_1$.
	This concludes the proof.
\end{proof}

Next, we show how to refine the $R$-factorization computed by \cref{lem:rlz_compute_logapprox} using ideas similar to those in a work of Fischer, Gagie, Gawrychowski, and Kociumaka~\cite{DBLP:conf/esa/0001GGK15}.
The proof of \cref{lem:refine} follows an approach described in \cite[Section 5.2]{DBLP:conf/esa/0001GGK15} that relies on the following result for dictionary matching in the read-only setting.

\begin{fact}[see Theorem 11 in the full version of {\cite{DBLP:conf/esa/0001GGK15}}]\label{fct:small_space_dict_match}
    Given a text $R[1\dd m]$ and patterns $P_1, \ldots, P_d$ of total length~$p$, we can compute, for each pattern $P_i$, the maximum length $\ell_i$ such that $P_i[1\dd \ell_i]$ occurs in $R$, and a position of $R$ where $P_i[1\dd \ell_i]$ occurs using $\cO(m \log m + p)$ total time and $\cO(d)$ space.\footnote{Here, we state a simplified version of Theorem 11 from the full version of~\cite{DBLP:conf/esa/0001GGK15} which is sufficient for our purposes.}
\end{fact}

\begin{lemma}\label{lem:refine}
    Suppose that we are given a positive integer $\alpha \geq 1$, a real $\epsilon \in (0, 1]$,
    and read-only access to a string~$R[1\dd m]$.
    Further, suppose that we receive an $\alpha$-optimal $R$-factorization of a string $T$ in a streaming fashion.
    Then, for any $\tau \in [\ceil{12\alpha / \epsilon} \dd n]$, we can output as a stream an $R$-factorization of $T$ with $(1+\epsilon)z$
    phrases, where $z = |\rlz(T,R)|$, in total time $\cO(\alpha^2 z \cdot \frac{1}{\epsilon} \cdot (m/\tau) \log m)$ using space $\cO(\tau)$.
\end{lemma}
\begin{proof}
    Let us denote the phrases of the $\alpha$-optimal $R$-factorization of $T$ that we receive by $f_1, \ldots, f_{\hat z}$.
    Let us denote the sequence of phrases we output by $o_1, \ldots, o_{\bar z}$.
    At any point in time, when our algorithm receives a phrase $f_v$, it stores a suffix $f_u, \ldots, f_v$ of the $\alpha$-optimal factorization such that
    $f_1 \cdots f_v = o_1 \cdots o_x f_u \cdots f_v$, where $o_x$ is the last phrase that our algorithm has returned.

    Set $\mu := \ceil{12 \alpha/\epsilon}$ and $\zeta =\cO(\tau)$ to be the smallest multiple of $\mu$ in $[\tau \dd \infty)$.
    Consider the arrival of a phrase $f_v$.
    We first append $f_v$ to the maintained sequence of phrases, thus obtaining sequence $S_v := f_u, \ldots, f_v$.
    It will be clear from the description below that we maintain $v - u \leq 2 \tau$ as an invariant.
    We then only perform further computations, described next, if $v = \hat z$ or $v \equiv 0 \pmod {\zeta}$.
    We partition~$S_v$ into parts consisting of $\mu$ phrases (and merge the last two parts if the last one is shorter).
    Let us denote the $i$-th part in this partition by $\Pi_i$.
    We then repeatedly apply \cref{fct:small_space_dict_match} as follows.
    First, for each $i$ we compute the longest prefix $F_i$ of the concatenation of strings in $\Pi_i$ that occurs in $R$ along with a position where it occurs using a call to the algorithm encapsulated in \cref{fct:small_space_dict_match}.
    We then repeat this process with each $\Pi_i$ replaced by $\Pi_i[1 + \absolute{F_i} \dd ]$; and so on.
    Since we know that each $\Pi_i$ admits an $R$-factorization that consists of $\cO(\alpha/\epsilon)$ phrases, it suffices to invoke the algorithm from \cref{fct:small_space_dict_match} $\cO(\alpha/\epsilon)$ times.
    Hence, the total time complexity of the described procedure is
    \[\cO((\alpha/\epsilon) \cdot (m\log m + \tau \cdot \epsilon / \alpha)) = \cO((\alpha/\epsilon) \cdot m \log m).\]
    This way, we compute an $R$-factorization $h_1 \cdots h_y$ of $f_u \cdots f_v$ and output phrases $h_1, \ldots, h_y$.
    Note that $v-u \leq \zeta-1 \leq 2\tau$ holds at all times due to the frequency of the application of the procedure described above.
    Further, note that said procedure is called at most $\ceil{\frac{\alpha z + 1}{\zeta}}$ times; once for every $\zeta$ phrases that arrive and perhaps once more when the last phrase arrives.
    Hence, the total running time is $\cO(\alpha^2 z \cdot \frac{1}{\epsilon} \cdot (m/\tau) \log m)$.

    The $\alpha$-optimality of the input $R$-factorization guarantees that each $\Pi_i$ considered throughout the course of the algorithm is not a fragment of $R$.
    Further, due to the greedy nature of our computation, all the phrases that we partition each $\Pi_i$ into, apart possibly from the last one, contain an endpoint of a factor of $\rlz(T,R)$.
    Hence, our factorization has at most one excess factor per encountered $\Pi_i$.
    The total number of excess factors in the output $R$-factorization $o_1 \cdots o_{\bar z}$ compared to the number of factors in the RLZ factorization is thus at most
        \[\left\lceil\frac{\alpha z + 1}{\zeta}\right\rceil \cdot \frac{2\tau + 1}{\mu}
        \leq \frac{4\alpha z}{\tau} \cdot \frac{3\tau}{\lceil 12\alpha / \epsilon \rceil}
        \leq \frac{12\alpha z \cdot \epsilon}{12\alpha}
        \leq \epsilon z.\qedhere\]
\end{proof}

We now put everything together.

\begin{theorem}\label{thm:approx_rlz_ub}
	Consider a read-only reference string $R[1\dd m]$, a streaming string $T[1\dd n]$, a positive integer
	$\tau \in [\ceil{m^{0.4}}/\epsilon \dd m]$, 
	 and a real $\epsilon \in (0, 1]$.
	Let $z = |\rlz(T,R)|$.
	We can preprocess $R$ in $\poly(m)$ expected time to construct an $\tO(\tau)$-space data structure, so that
	we can then process $T$ in total time $\tO(z \cdot (m/\tau)^{5/3} + z\cdot m/(\epsilon\tau) + n \cdot  (1 + m^{3} / \tau^6) + m)$ total time using $\cOtilde(\tau)$ extra space,
	outputting an $R$-factorization of $T$ of size at most $(1+\epsilon)z$ as a stream.
	
	The preprocessing can be implemented in $\tO(m \cdot \poly(\err))$ time with success probability at least $1 - 2^{-\err}$ at the expense of increasing the space complexity and time required for processing $T$ by an $\cO(\poly(\err))$ multiplicative factor, for any integer $\err \geq \log_2 (m\sigma)$.
\end{theorem}
\begin{proof}
	If $\tau < \ceil{12\ceil{8\log_2^2 m + 2}/\epsilon}$, we replace $\tau$ by $999 \tau$ noting that, for any positive integer $m$,
	we have
	$999 \ceil{m^{0.4}}/\epsilon \geq \ceil{12\ceil{8\log_2^2 m + 2}/\epsilon}$ holds.
	This ensures that all of the constraints required by the lemmas we use below are satisfied.

	We employ \cref{lem:rlz_compute_logapprox} to output (as a stream) a $\Delta$-core $R$-factorization of $T$ for a positive integer $\Delta$ within the stated bounds.
	This factorization is $\ceil{8\log_2^2 m + 2}$-optimal due to \cref{lem:bad_approx_ratio}.
	We process this output stream according to \cref{lem:refine} with $\alpha = \ceil{8 \log_2^2 m+2}$, thus obtaining an $R$-factorization of $T$ of size at most $(1+\epsilon)z$ within the stated complexities.
\end{proof}

\edef\newT{t}
\let\newS\tau

\subsection{Lower Bound}

Finally, we show a lower bound on the space-time product of a $(1+\varepsilon)$-approximation algorithm for computing a relative Lempel--Ziv factorization.

\subparagraph{Lower bounds via branching programs.}
The $R$-way branching program model of computation, introduced by Cook and Borodin
\cite{DBLP:journals/siamcomp/BorodinC82}, is asymptotically at least as powerful as the word RAM model with word-width $\cO(\log R)$.
It has been useful for proving unconditional lower bounds on time-space tradeoffs for decision problems (see, e.g., \cite{DBLP:journals/jcss/BeameJS01,DBLP:journals/cjtcs/Pagter05,DBLP:journals/tcs/Patt-ShamirP93} and references therein).

An \emph{$R$-way branching} program is defined for a set of input instances $I \subseteq [R]^n$, i.e., each instance is an ordered tuple $(x_1, \dots, x_n)$ of $n$ variables.
The program is modeled as a directed acyclic graph in which one node is the dedicated \emph{start node}.
Every node is labeled with some index $i \in [n]$, and we say that the node \emph{probes} variable $x_i$.
Each node has out-degree at most $R$, and the outgoing edges are labeled with distinct values from $[1, R]$.
Each edge can optionally be annotated with one output value from $[1, R]$.
The computation is performed by repeatedly executing the following procedure, starting from the start node.
If the current node is labeled $i$, then we follow its outgoing edge with label~$x_i$. If this edge is annotated with some output, we report this output.
If the edge does not exist, then we terminate.
The traversed edges form the \emph{computation path} of the input.

The time complexity~$\newT$ of an $R$-way branching program is the length of the longest computation path traversed by any input from $I$ (measured in edges), and its capacity is the base-two logarithm of the number of nodes reachable by any input from $I$.
For our purposes, we use a slightly different notion of time and capacity that focuses on the parts of the program that actually produce output values.
This will be useful, e.g., for splitting the preprocessing and streaming phase of asymmetric streaming algorithms. 
For any input, its \emph{output-critical computation path} is obtained by removing from its computation path the lengthwise maximal prefix that does not produce any output values.
Then, the \emph{output-critical time complexity} is the length of the longest output-critical computation path traversed by any input from $I$. The \emph{output-critical capacity} is the base-two logarithm of the number of nodes in the union of the output-critical computation paths of all inputs from $I$.

\begin{lemma}\label{lem:readonly_to_Rway_generic}
    On a word RAM of word-width $\cO(\log_2 R)$, consider an algorithm designed to work for a set of possible inputs~$I \subseteq [R]^n$.
    The algorithm receives the input as a read-only array of size~$n$, with each value stored in $\ceil{\log_2 R}$ bits.
    The algorithm is allowed to preprocess the input into an auxiliary data structure.
    Then, using the data structure and random access to the input, it produces a stream of output elements from $[R]$, each in $\ceil{\log_2 R}$-bit representation.

    If, during the output phase, the algorithm uses at most $\newS$ bits of working space (including the auxiliary data structure) and at most $\newT$ time, then there is an equivalent $R$-way branching program with output-critical time complexity $\cO(\newT)$ and output-critical capacity $\cO(\newS)$.    
\end{lemma}

\begin{proof}    
    Consider the following adaptation of the word RAM model. The algorithm may no longer access the input or write the output in a word-wise manner.
    Instead, it can only read (resp.\ write) a single input (resp.\ output) at a time, which takes constant time. 
    
    If the word-width is $w = \cO(\log_2 R)$, then, in the original model, the algorithm can (partially) read (resp.\ write) up to $1 + \ceil{w / \log_2 R} = \cO(1)$ consecutive input (resp.\ output) elements in a single operation.
    Such a multiple-read (resp.\ write) operation can be simulated in the adapted model in $\cO(1)$ time and $\cO(w)$ bits of space (merely performing multiple oracle queries or output operations and using the extra space to buffer the elements). Trivially, $w = \cO(\newS)$, as any word RAM algorithm requires at least one word of space.
    Therefore, there is an equivalent algorithm in the adapted word RAM model that still uses $\cO(\newS)$ bits of space and $\cO(\newT)$ time.
    
    We create an $R$-way branching program with $2^{\cO(\newS)}$ states, one for each memory configuration of this algorithm.
    We then add edges that are equivalent to the behavior of the algorithm, which is possible because at most one element is read or written per computation step.  
\end{proof}

We will need the following tool.

\begin{lemma}\label{lem:rway_to_tree}
    Consider an $R$-way branching program over $n$ input variables. There is an $R$-way branching program whose structure is a tree rooted in its start node, and satisfying the following property for every possible input. The computation path consists of at least $n$ edges. During the first~$n$ steps, every variable is probed exactly once.
    The program outputs value~$y$ in the $i$-th computation step if and only if the original program outputs $y$ in the $i$-th computational step.
\end{lemma}

\begin{proof}
    We obtain the program by modifying the original one. First, we create $n + 1$ nodes $v_0, \dots v_n$, where, for every $i \in [0, n)$ and every $j \in [R]$, there is an edge from $v_i$ to $v_{i + 1}$ that has label $j$ and no output annotation.
    For every node $v$ of the original program and for every $j \in [R]$, if node $v$ has no outgoing edge with label $j$, then we add such an edge with destination $v_0$, and without output annotation.
    This does not change the behavior of the program and ensures that all computation paths are of length at least $n$.
    Also, every node has out-degree either $0$ or $R$.
    
    Next, as long as there is a node that has multiple ingoing edges, we duplicate the node (including its outgoing edges) once per ingoing edge, such that each copy has exactly one ingoing edge.
    (This is indeed a finite process; if we process the nodes in topological order, then we never add ingoing edges to already processed nodes.)
    Now the program is a tree rooted in the start node, where every leaf is at depth at least $n$, and every node is either a leaf or has $R$ outgoing edges.
    
    Finally, consider a node $v$ with ancestors $a_1, \dots, a_d$ (in root-to-$v$ order), where $d < n$. Such a node cannot be a leaf, and hence it has children $u_1, \dots u_R$ (respectively reached via edges labeled $1, \dots, R$).
    Let $a_{d + 1} = v$.
    If there is some $i \in [1, d]$ such that both $a_i$ and $v$ probe the same variable, then let $x$ be the label of the edge from $a_i$ to $a_{i + 1}$.
    Every input that reaches node $v$ must assign value $x$ to the variable.
    We proceed as follows. We change the label of $v$ such that it probes some variable not probed by any of its ancestors.
    Then, for every $j \in [R] \setminus \{ x \}$, we replace child $u_j$ and the entire subtree rooted in $u_j$ with a copy of the subtree rooted in $u_x$.
    If the edge with label $x$ outgoing from $v$ has an output annotation, then we assign the same annotation to the outgoing edge from $v$ with label $j$ (and otherwise we assign no annotation to this edge).
    This again does not change the behavior of the program.
    We apply this procedure until no longer possible, in order of depth, starting from the root.
    Afterwards, the tree satisfies the claimed properties.
\end{proof}

\subparagraph{Lower bounds via set complementation.}

We obtain lower bounds via a reduction from the set complementation problem.
In this problem, we are given a set $X \subset [R]$ of size $\absolute{X} = n$ (as a read-only array) and have to report all elements of $[R] \setminus X$.
For $R = 2n$, it is known that any algorithm using $\cO(\newS)$ additional working space requires $\Omega(n^2 / \newS)$ time \cite{DBLP:journals/tcs/Patt-ShamirP93}. 
We show a similar result for a relaxed version of the problem, where a subset of the reported elements is allowed to be from $X$, i.e., the algorithm is allowed to produce false positives. We also allow an arbitrary time and space preprocessing of $X$ before producing the first output value. Using this version of the problem, we then obtain lower bounds for batched function inversion and relative Lempel-Ziv compression.

\defdetailedcustomproblem{$(\alpha,\beta,\gamma,n)$-Set-Complementation}%
{Let $\alpha, \beta, \gamma \in \mathbb R_{> 0}$ and $n \in \mathbb N_{> 0}$. Let $R = n + {\ceil{\alpha n}}$.}%
{%
Input/{ordered tuple $(x_1, \dots, x_n) \in [R]^n$ in which all values are distinct},%
Output/{list of at most $\beta n$ not necessarily distinct elements of $[R]$, of which at least $\gamma n$ distinct elements are in the set complement $[R] \setminus \{x_1, \dots, x_n\}$}}%

\begin{theorem}\label{thm:prep_set_comp}
    Let $\alpha, \gamma \in \mathbb R_{> 0}$ and $n \in \mathbb N_{> 0}$ such that $\alpha n, \gamma n \in \mathbb N_{>0}$ and $\gamma \leq \alpha \leq \frac{1}{3}$.
    Let $R = n + {\alpha n}$ and $\beta = \frac\gamma{3\alpha}$.
    Let $I \subseteq [R]^n$ be a set of $\absolute{I} \geq {R!}/{((R - n)! \cdot 2^{32})}$ input instances, where each instance consists of $n$ distinct values. Consider an $R$-way branching program that solves $(\alpha,\beta,\gamma,n)$-set-complementation for all instances in $I$. If the program has output-critical capacity $\newS < \gamma n/ 2^{32}$ and output-critical time complexity~$\newT$, then $\newS \newT \geq \gamma n^2 / 2^{32}$.
\end{theorem}

\begin{proof}
For some input $(x_1, \dots, x_n) \in I$, we say that an output value $y \in [R]$ is correct if and only if $y \notin \{ x_i \mid i \in [n] \}$.
Consider any $R$-way branching program for $I$ that has output-critical capacity~$\newS$ and output-critical time complexity $\newT$.
Every output-critical computation path must output $\gamma n$ correct values, and hence it consists of at least $\gamma n + 1$ nodes, i.e., $\newT \geq \gamma n$ and $\newS \geq \log_2 (\gamma n + 1) \geq 1$.
If $\gamma n \leq 2^{32}$ then $\newS \geq 1 \geq \gamma n / 2^{32}$ and the claim trivially holds.
Hence, we can assume $\gamma n > 2^{32}$, which implies $\newS \geq \log_2 (\gamma n + 1) > 32$.
If $\newT \geq \gamma n^2 / 2^{32}$, then also $\newS\newT \geq \newT \geq \gamma n^2 / 2^{32}$ and the claim trivially holds.
From now on, we can assume $2^{32} < \gamma n \leq \newT < \gamma n^2 / 2^{32}$.

We start by showing that every output-critical computation path contains a short piece on which a large number of correct outputs and a relatively small number of incorrect outputs is produced.

\begin{claim}\label{claim:dense_piece}
    Let $r := \min(\floor{\gamma n / 1000}, \floor{\gamma n^2 / (20000\newT)})>0$.
    For every input from $I$, its output-critical computation path contains a piece of at most $\ceil{n / 20}$ consecutive edges that produce at least~$r$ correct output values and at most $\floor{(22/21) \cdot r\beta/\gamma}$ output values overall.
\end{claim}

\begin{proof}
    We cut the output-critical computation path into edge-disjoint pieces such that each piece produces exactly $r$ correct output values. 
    We discard the final piece if it produces fewer correct values.
    Because of $r \leq \gamma n / 1000$, the number of pieces is $\floor{\gamma n   / r} > \gamma n / r - 1 \geq 999/1000 \cdot \gamma n / r$.
    If $\newT > \ceil{n/20}$, then necessarily $r = \floor{\gamma n^2 / (20000\newT)}$.
    Hence, if at least a $(1/999)$ fraction of the pieces are of length over $\ceil{n/20}$, 
    then their total length is strictly more than 
    \[1/999 \cdot 999/1000 \cdot \gamma n / r \cdot n/20 = \gamma n^2 / (20000r) \geq \newT,\]
    which is a contradiction.
    If at least a $(998/999)$ fraction of the pieces produce over $(22/21) \cdot r\beta/\gamma$ output values, then the overall number of outputs exceeds 
    \[(998/999) \cdot (999/1000) \cdot \gamma n / r \cdot (22/21) \cdot r\beta/\gamma > \beta n,\] 
    which is another contradiction.
    Hence, we can find a $(1/999)$ fraction of pieces such that each of them produces at most $(22/21) \cdot r\beta/\gamma$ output values.
    Then, within this $(1/999)$ fraction of pieces, there must be at least one piece of length at most $\ceil{n/20}$.
\end{proof}

We assign every input to the initial node of a piece satisfying \cref{claim:dense_piece}.
There is a node with at least $\absolute{I} / 2^\newS$ inputs assigned to it.
We focus on this node $v$ and the subset $I'$ of inputs.
We apply \cref{lem:rway_to_tree} to the subprogram obtained by defining the start node to be $v$ and including all edges and nodes reachable from $v$.
From now on, we say \emph{modified program} to refer to the resulting program.
For every input from $I'$, we run the first $\ceil{n / 20}$ computational steps of the modified program and assign the input to the node reached in this way.
Since all $\ceil{n/20}$ variables probed during this procedure are distinct and have distinct values, there are $R! / (R - \ceil{n / 20})!$ nodes reachable in this way.
Consequently, there must be a final node $u$ with the following number of inputs assigned to it:
\[\left.\absolute{I'} \middle/ \frac{R!}{(R - \ceil{n / 20})!} \right. \geq \frac{(R - \ceil{n / 20})!}{((R - n)! \cdot 2^{\newS + 32})} \geq \frac{(R - \ceil{n / 20})!}{((R - n)! \cdot 2^{2\newS})}.\]

Our goal is to provide an upper bound on this number, resulting in a lower bound on $2^{2\newS}$.
Hence, we count the number of possible inputs that could be compatible with $u$, i.e., inputs for which $v$ is the initial node of a piece satisfying \cref{claim:dense_piece}, and for which running the modified program for $\ceil{n / 20}$ steps results in the final node $u$.
We focus on the unique computation path from $v$ to $u$ (recalling that the modified program is a tree).
Due to \cref{claim:dense_piece}, we know that, while running the modified program, the first at most $(22/21) \cdot r\beta/\gamma$ output values must contain at least $r$ correct output values (for every compatible input), and there are indeed at least $r$ correct outputs produced during the initial $\ceil{n / 20}$ steps.
There are at most $\binom{\floor{(22/21) \cdot r\beta/\gamma}}{r} \leq ({(22/21) \cdot e \cdot \beta}/{\gamma})^r$ ways of fixing $r$ correct output values among the first at most $(22/21) \cdot r\beta/\gamma$ output values produced by the path.
The path from $v$ to $u$ fixes the values of $\ceil{n/20}$ variables. The remaining variables must avoid the already used values, as well as the $r$ correct outputs.
Hence, there are $(R - \ceil{n/20} - r)! / (R - n - r)!$ ways of assigning the remaining variables.
We have shown that the number of inputs that are compatible with $u$ cannot exceed 
\[\left(\frac{22 \cdot e \cdot \beta}{21\gamma}\right)^r \cdot \frac{(R - \ceil{n/20} - r)!}{(R - n - r)!}.\]

Below, we use our lower bound on the number of inputs assigned to $u$ as well as our upper bound on the number of inputs compatible with $u$ to obtain the lower bound on $\newS$.
We exploit that $R - n = \alpha n$ and $R - \ceil{n / 20} - r + 1 \geq R - n/20 - \gamma n \geq R - n/20 - \alpha n = 19n/20$.
\newcommand{\xdot}[1][0]{\mathrlap{\hspace*{#1em}{}\enskip\cdot\enskip{}}\hspace*{2em}}%
\begin{alignat*}{2}
    & \mathrlap{\frac{(R - \ceil{n/20})!}{(R - n)! \cdot 2^{2\newS}} \leq \left(\frac{22 \cdot e \cdot \beta}{21\gamma}\right)^r \cdot \frac{(R - \ceil{n/20} - r)!}{(R - n - r)!}}&&\\[.5\baselineskip]
    \mathllap{\implies \enskip 2^{2\newS}}
    \enskip \geq \enskip {}& 
    \left(\frac{21\gamma}{22 \cdot e \cdot \beta}\right)^r \xdot & \frac{(R - n - r)!}{(R - n)!} \xdot & \frac{(R - \ceil{n/20})!}{(R - \ceil{n/20} - r)!}\\[.25\baselineskip]
    \enskip \geq \enskip {}&
    \left(\frac{21\gamma}{22 \cdot e \cdot \beta}\right)^r  \xdot[.75] & (R - n)^{-r} \xdot & (R - \ceil{n/20} - r + 1)^{r}\\[.25\baselineskip]
    \enskip \geq \enskip {}&
    \left(\frac{21\gamma}{22 \cdot e \cdot \beta}\right)^r  \xdot[1.5] & (\alpha n)^{-r} \xdot & (19n/20)^{r}\\[.25\baselineskip]
    \enskip \geq \enskip {}&
    \mathrlap{\left(\frac{21 \cdot 19 \gamma}{22 \cdot 20 \cdot e \cdot \alpha\beta}\right)^r = 
    \left(\frac{21 \cdot 19 \cdot 3}{22 \cdot 20 \cdot e} \cdot \frac{\gamma}{3\alpha} \cdot \frac{1}{\beta} \right)^r = 
    \left(\frac{21 \cdot 19 \cdot 3}{22 \cdot 20 \cdot e} \right)^r > (1.0007)^r
    } &&
\end{alignat*}
Hence, we have shown $\newS \geq r/2 \cdot  \log_2 (1.0007) > r / 2^{11}$.
Finally, if $r = \floor{\gamma n / 1000}$ then $r > \gamma n / 2^{11}$ and $\newS \geq \gamma n / 2^{22}$.
If $r = \floor{\gamma n^2 / (20000\newT)}$ then $r > \gamma n^2 / (2^{16} \cdot \newT)$. Therefore, $\newS \geq \gamma n^2 / (2^{28} \cdot \newT)$, which implies $\newS\newT \geq \gamma n^2 / 2^{32}$, concluding the proof.
\end{proof}

\begin{corollary}\label{lem:ram_set_comp}
    Let $\alpha, \beta, \gamma, n, R$ be as in \cref{thm:prep_set_comp}.
    Consider a randomized word RAM algorithm for $(\alpha, \beta, \gamma, n)$-set-complementation over all possible inputs, where the word-width is $\cO(\log R)$.
    The input is given read-only, with each variable stored in $\ceil{\log_2 R}$ bits.
    The algorithm is allowed to preprocess the input into an auxiliary data structure.
    Then, using the data structure and random access to the input, it produces the output elements as a stream.
    
    If the algorithm succeeds with probability $\geq 1/100$ and, during the output phase, uses worst-case space $\newS$ bits (including the data structure) and worst-case time $\newT$, then $\newS = \Omega(\gamma n)$ or $\newS\newT = \Omega(\gamma n^2)$.
\end{corollary}

\begin{proof} 
    Let $I$ of size $R! / (R - n)!$ be the set of all possible inputs. Let $r$ be the outcome of the randomness of the algorithm.
    For input $x \in I$, we define indicator variable $X_{x, r} = 1$ if and only if the algorithm with randomness $r$ succeeds for input $x$ (and otherwise $X_{x,r} = 0$).
    We have $\mathbb E_r[X_{x, r}] = P_r[X_{x, r} = 1] \geq 1/100$.
    By the linearity of expectation, $\mathbb E_{r}[\sum_{x \in I} X_{x, r}] = \sum_{x_\in I} \mathbb E_r[X_{x, r}] \geq \absolute{I} / 100$.
    Therefore, there must be a fixed~$r$ such that the algorithm succeeds for $\absolute{I} / 100$ inputs, which implies the existence of a deterministic algorithm that succeeds for $\absolute{I} / 100$ inputs.
    By \cref{lem:readonly_to_Rway_generic}, there is an $R$-way branching program that solves the problem on this subset of inputs with output-critical capacity $\cO(\newS)$ and output-critical time $\cO(\newT)$. Due to \cref{thm:prep_set_comp}, this implies $\newS = \Omega(\gamma n)$ or $\newS\newT = \Omega(\gamma n^2)$.
\end{proof}

\begin{corollary}
    On a word RAM of word-width $\cO(\log N)$, consider a data structure that, when constructed for a function ${f : [N] \rightarrow [N]}$ with parameter $b \in [N]$, allows answering the following type of query. Given a batch of $b$ elements $\{x_1, \dots, x_b\} \subseteq[N]$, report for each $x_i$ either a preimage element $y_i \in f^{-1}(x_i)$, or report that $f^{-1}(x_i) = \emptyset$. The data structure may evaluate $f$ at query time.
    
    Assume that the construction succeeds with probability $\geq 1/100$, and, after successful construction, all queries can be answered correctly.
    If the size of the data structure including the working space at query time is $\newS = o(N)$ bits, then a query takes $\Omega(Nb/\newS)$ time in the worst case.
\end{corollary}
\begin{proof}
    Without loss of generality, assume that $N - 1$ is a multiple of $4$. Let $n = {3\cdot(N - 1)}/{4}$ and $R = n + n / 3 = N - 1$.
    We reduce $(\frac{1}{3},\frac{1}{3},\frac{1}{3}, n)$-set-complementation to constructing and querying the data structure.
    Consider an input $(x_1, \dots, x_n)$ for this problem.
    We define the function $f : [R + 1] \rightarrow [R + 1]$ as $f(i) = x_i$ if $i \in [n]$ and otherwise $f(i) = R + 1$.
    Given random access to $(x_1, \dots, x_n)$, we can evaluate the function in constant time (without storing additional information).
    We preprocess $f$ into a function inversion data structure of size $\newS = o(N)$ bits.

    Now we show how to compute the set complement.
    Let $\newT$ be the worst-case time to answer a batch of $b$ queries. 
    For every $i \in [1, \ceil{R / b}]$, we invert the batch $(B - b, B]$ with $B = \min(R,ib)$. For every $j \in (B - b, B]$, we report that $j$ is in the set complement if and only if the inversion query did not report a preimage in $[n]$.
    Hence, we report the entire set complement in $\cO(R\newT / b + R)$ time.
    By \cref{lem:ram_set_comp}, and due to $\newS = o(N) = o(n)$, we have $R\newT / b + R = \Omega(N^2 / \newS)$.
    Because of $R = N - 1 \notin \Omega(N^2 / \newS)$, this implies $\newT = \Omega(Nb / \newS)$. 
\end{proof}

\begingroup
    \let\olddollar\textdollar
    \let\oldhash\#
    \def\${{\texttt\olddollar}}%
    \def\#{{\texttt\oldhash}}%
    \let\phi\unvarphi

\begin{theorem}\label{lem:rlz_lower_fixed_text_length}\let\phi\unvarphi
    Fix positive integers $m,z,\phi$.
    On a word RAM of word-width $\cO(\log m)$, consider an algorithm that preprocesses a read-only reference $R[1\dd m]$ into a data structure of size $\newS = o(m / \phi^2)$ bits. 
    Then, given text $T[1\dd 48z\phi]$ satisfying $\absolute{\rlz(R, T)} \geq z$, it outputs an $R$-factorization of $T$ 
    of size at most $\phi \cdot \absolute{\rlz(R, T)}$ 
    in left-to-right order (each phrase as a fragment of $T$ without providing an occurrence in~$R$), using $\newS$ bits of additional space. 
    
    If the preprocessing succeeds with probability $\geq 1/100$ such that afterwards the text processing succeeds for every text, then the text processing takes $\Omega(zm / (\phi^3\newS))$ time in the worst-case.
\end{theorem}

\begin{proof}
    Let $a = 48\phi$, $\alpha = 1/(48\phi)$, and $k = \floor{m/(2a^2)} \cdot a$.
    If $m < 2a^2$ then the claimed lower bound is trivial. Hence, assume $m \geq 2a^2$ and thus $k \geq a$.
    We will reduce $(\alpha, 1/3, \alpha, k)$-set-complementation to the problem of computing, for some reference $R$ and some text $T$, an $R$-factorization of $T$ of size at most $\phi \cdot \absolute{\rlz(R, T)}$. The input is an array of read-only variables $x_1, \dots, x_k$. Let $X = \{x_1, \dots, x_k\}$, then we have to report a superset of $[K] \setminus X$, where $K = k + \alpha k$.
    The total number of reported elements may not exceed $k / 3 = 16k\phi/a$.
    We define the reference 
    \[R[1\dd m] = x_1^{a} \cdot \$ \cdot x_2^{a} \cdot \$ \cdots \$ \cdot x_k^{a} \cdot \$^{m - (a + 1) \cdot k + 1}.\]
    Note that $(a + 1) \cdot k \leq 2ak = 96\phi k \leq m$, i.e.,~$R$~is indeed well-defined.
    For now, we focus on the case in which $z = K$. We generalize the construction for arbitrary $z$ later.
    We define the text
    \[ T[1\dd aK] = \texttt{1}^{a} \cdot \texttt{2}^{a} \cdots {K}^{a}. \]

    Note that, given random access to $x_1, \dots, x_k$, we can simulate constant time random access to~$R$ and $T$.
    The size of $\rlz(T,R)$ is exactly $2k \geq K$: a gadget $\sigma^{a}$ becomes a reference if and only if $\sigma \in X$, resulting in $k$ references; if $\sigma$ is one of the $\alpha k$ elements of $[K] \setminus X$, then the gadget is factorized into $a$ literals, resulting in $a \cdot \alpha k = k$ literals.
    Therefore, in any $\phi$-approximate factorization, there are at most $2k\phi$ phrases.
    
    Now assume that we receive some $\phi$-approximate factorization as a stream.
    If some gadget $\sigma^{a}$ gets factorized into fewer than~$a$ phrases (i.e., at least two consecutive $\sigma$s are in the same reference), then we know that $\sigma \in X$, and we do not report $\sigma$. If, however, the gadget gets factorized into $a$ literals, then we do not know with certainty if $\sigma \in X$, and we do report~$\sigma$.
    It is clear that we indeed report the $\alpha k$ elements of $[K] \setminus X$.
    To count the number of false positives, we recall that there are at most $2k\phi$ phrases. This implies that at most $2k\phi / a < k/3$ elements are reported overall.
    
    Hence, if we compute an $R$-factorization of $T$ of size at most $\phi \cdot \absolute{\rlz(R, T)}$ in $\newS$ bits of space, then we also solve $(\alpha, 1/3, \alpha, k)$-set-complementation in $\newS$ bits of space, and \cref{lem:ram_set_comp} implies that the time is $\Omega(\alpha k^2 / \newS)$ in the worst case. (Note that we can indeed apply \cref{lem:ram_set_comp} because of $\log k = \Theta(\log(m))$.)
    Substituting $\alpha = \Omega(1/\phi)$ and $k = \Omega(m/\phi)$, the lower bound is $\Omega(km / (\phi^2 \newS))$.
    If $z = K = \Theta(k)$, then the bound is $\Omega(zm / (\phi^2 \newS))$.

    \subparagraph{\boldmath Lower bound for $z < K$.\unboldmath} 
    
    We cut $T$ into $\ceil{K/z}$ pieces, each of length exactly $az$, with the final piece possibly shorter. If the final piece is shorter, then we pad it with a prefix of $T$ such that it is also of length $az$.
    Hence, each piece consists of exactly $z$ substrings of the form $\sigma^a$, and $\absolute{\rlz(T,R)} = z$.
    Over all the pieces, the number of phrases is at most $4k$: in the entire $T$, we had exactly $2k$ phrases, and padding the final piece at most doubles the number of phrases.
    Given random access to $x_1, \dots, x_k$, we can simulate constant time random access to each piece.

    Assume that some algorithm produces a $\phi$-approximate factorization of a piece as a stream, using~$\newS$ bits of space. In a moment, we will show that we can assume without loss of generality that the algorithm does not modify the precomputed data structure during the text processing phase. 
    Hence, without rerunning the preprocessing, we can run the algorithm for each piece.
    We use the same strategy as before, outputting $\sigma$ if and only if $\sigma^a$ got factorized into $a$ literals.
    The number of reported elements over all pieces is at most $4k\phi / a < k / 3$, and thus we still solve set-complementation.
    Therefore, the total time for processing all pieces must be $\Omega(km / (\phi^2 \newS))$, and a worst-case piece requires $\Omega(z/K \cdot km / (\phi^2 \newS)) = \Omega(zm / (\phi^2 \newS))$ time.

    It remains to be shown that we indeed do not need to rerun the preprocessing.
    During preprocessing, the algorithm computes a data structure of size $\newS$ bits.
    We show that, without significantly increasing the space or time usage of the algorithm for computing an $R$-factorization of $T$, we can restore this data structure after the text processing phase.
    During the first $\ceil{\newS / \log m}$ steps of the text processing phase, we explicitly keep track of all the changes made to the data structure.
    This history of operations can be stored in $\cO(\newS)$ bits of space.
    After $\ceil{\newS / \log m}$ steps, we pause the text processing phase.
    We create a copy of the current state of the data structure in $\cO(\newS / \log m)$ time and $\cO(\newS)$ bits of space.
    Then, we continue the text processing phase.
    Once the entire text has been processed, we can restore the state of the data structure by reverting the history of changes on either the copy created after $\ceil{\newS / \log m}$ steps or the original data structure.
    In either case, we increase both time and space by at most a constant factor.

    \subparagraph{\boldmath Lower bound for $z > K$.\unboldmath}
    We use the text $T'[1\dd az] = T^\infty[1\dd az]$, consisting of $z$ gadgets of the form $\sigma^a$.
    The size of $\rlz(T',R)$ is at least $z$ and at most $2k \cdot \ceil{z/K} \leq 2k \cdot \ceil{z/k} \leq 4z$.
    Given random access to $x_1, \dots, x_k$, we can simulate constant-time random access to $T'$.

    We analyze an $R$-factorization of $T'$ which has at most $\phi \cdot \absolute{\rlz(T',R)}$ phrases. 
    There are $\floor{z / K}$ occurrences of $T$ in $T'$.
    Since $T'$ gets factorized into at most $4z\phi$ phrases, there are at most $z/(2K)$ occurrences of $T$ that get factorized into more than $8K\phi \leq 16k\phi$ phrases.
    We focus on the at least $\ceil{z/(2K)}$ occurrences of $T$ that are factorized into at most $16k\phi$ phrases.
    For each such occurrence $T'[i\dd i + aK) = T$, we can define an algorithm for solving set-complementation as follows.
    We run the preprocessing for computing the $R$-factorization of $T'$, and then run the text processing phase until we output the phrase ending at position $T[i - 1]$. (Indeed, no phrase can stretch across two occurrences of $T$.)
    This concludes the preprocessing of the set-complementation algorithm.

    For the output phase of set-complementation, we continue running the text processing phase of the algorithm for computing the $R$-factorization of $T'$ until we receive the phrase that ends at position $T'[i + aK - 1]$, i.e., we obtain a factorization of $T$ that consists of at most $16k\phi$ phrases as a stream.
    Using this stream, we can solve set-complementation as before, reporting some element $\sigma$ if and only if $\sigma^a$ gets factorized into $a$ literals.
    It follows that producing this section of the text processing phase takes $\Omega(km / (\phi^2 \newS))$ time. 
    We stress that this lower bound holds for \emph{each} occurrence of $T$ that gets factorized into at most $16k\phi$ phrases; if there is a single occurrence that takes less time, then we also solve set-complementation in less time.
    We multiply the lower bound with the number $\ceil{z/(2K)}$ of considered occurrences, resulting in the lower bound $\Omega(z / K \cdot km / (\phi^2 \newS)) = \Omega(zm / (\phi^2 \newS))$.
\end{proof}

\begin{corollary}\label{cor:rlz_lb}
    Fix positive integers $n,m,z,\phi$ such that $100z\phi \leq n \leq zm$.
    On a word RAM of word-width $\cO(\log m)$, consider an algorithm that preprocesses a read-only reference $R[1\dd m]$ into a data structure of size $\newS = o(m / \phi^2)$ bits. 
    Then, given text $T[1\dd n]$ satisfying $\absolute{\rlz(R, T)} \in [z, 100z\phi]$, it outputs an $R$-factorization of $T$ of size at most $\phi \cdot \absolute{\rlz(R, T)}$ as a stream in left-to-right order (each phrase as a fragment of $T$ without providing an occurrence in~$R$), using $\newS$ bits of additional space. 
    
    If the preprocessing succeeds with probability $\geq 1/100$ such that afterwards the text processing succeeds for every text, then the text processing takes $\Omega(zm / (\phi^2\newS))$ time in the worst-case.
\end{corollary}

\begin{proof}
    Let $\phi' = 2\phi$. Consider a reference $R$ of length $m$ and a text $T'$ of length $48z\phi'$. 
    By \cref{lem:rlz_lower_fixed_text_length}, computing an $R$-factorization of $T'$ of size at most  $\phi' \cdot \absolute{\rlz(R, T)}$ takes $\Omega(zm / (\phi^2\newS))$ text processing time in the worst case (with the same assumptions on preprocessing and space as in the claim).
    Particularly, the lower bound is obtained using a text $T'[1\dd 48z\phi']$ for which the size of $\rlz(T',R)$ is $z' \in [z, 48z\phi']$.
    We define $T[1\dd n] = T'[1\dd 48z\phi'] \cdot R^\infty[1\dd n - 48z\phi']$.
    Trivially, $R^\infty[1\dd n - 48z\phi']$ can be factorized into at most $n/m \leq z$ phrases, and the exact factorization of~$T$ consists of $z' + z \leq 48z\phi' + z < 100z\phi$ phrases. If we have an $R$-factorization of $T$ of size at most $\phi \cdot (z' + z) \leq 2\phi z' = \phi'z'$, it immediately reveals a $\phi'$-approximate factorization of $T'$.
    It follows that computing an $R$-factorization of $T$ of size at most $\phi \cdot \absolute{\rlz(T,R)}$ takes at least $\Omega(zm / (\phi^2\newS))$ text processing time in the worst case.
\end{proof}

\endgroup

\bibliographystyle{plainurl}
\bibliography{main}

@inproceedings{10.1145/3357713.3384266,
author = {Chan, Timothy M. and Golan, Shay and Kociumaka, Tomasz and Kopelowitz, Tsvi and Porat, Ely},
title = {Approximating text-to-pattern Hamming distances},
year = {2020},
isbn = {9781450369794},
doi = {10.1145/3357713.3384266},
booktitle = {Proceedings of STOC 2020},
pages = {643–656},
numpages = {14}
}

@article{DBLP:journals/algorithmica/GolanKP19,
  author       = {Shay Golan and
                  Tsvi Kopelowitz and
                  Ely Porat},
  title        = {Streaming Pattern Matching with d Wildcards},
  journal      = {Algorithmica},
  volume       = {81},
  number       = {5},
  pages        = {1988--2015},
  year         = {2019},
  url          = {https://doi.org/10.1007/s00453-018-0521-7},
  doi          = {10.1007/S00453-018-0521-7},
  bibsource    = {dblp computer science bibliography, https://dblp.org}
}

@inproceedings{DBLP:conf/latin/AyadLPV24,
  author       = {Lorraine A. K. Ayad and
                  Grigorios Loukides and
                  Solon P. Pissis and
                  Hilde Verbeek},
  title        = {Sparse Suffix and {LCP} Array: Simple, Direct, Small, and Fast},
  booktitle    = {Proc.~of {LATIN} {(1)}},
  pages        = {162--177},
  year         = {2024},
  doi          = {10.1007/978-3-031-55598-5_11},
}

@article{GABORY2025105296,
title = {Elastic-degenerate string comparison},
journal = {Inf. Comput.},
volume = {304},
pages = {105296},
year = {2025},
issn = {0890-5401},
doi = {https://doi.org/10.1016/j.ic.2025.105296},
author = {Estéban Gabory and Moses Njagi Mwaniki and Nadia Pisanti and Solon P. Pissis and Jakub Radoszewski and Michelle Sweering and Wiktor Zuba}
}

@misc{mai2021optimalspacetimestreaming,
      title={Optimal Space and Time for Streaming Pattern Matching}, 
      author={Tung Mai and Anup Rao and Ryan A. Rossi and Saeed Seddighin},
      year={2021},
      eprint={2107.04660},
      archivePrefix={arXiv},
      primaryClass={cs.DS},
      url={https://arxiv.org/abs/2107.04660}, 
}

@article{AMIR201534,
title = {Dictionary matching with a few gaps},
journal = {Theor. Comput. Sci},
volume = {589},
pages = {34-46},
year = {2015},
issn = {0304-3975},
doi = {https://doi.org/10.1016/j.tcs.2015.04.011},
author = {Amihood Amir and Avivit Levy and Ely Porat and B. Riva Shalom}
}

@InProceedings{10.1007/978-3-642-20662-7_7,
  author       = {Tuukka Haapasalo and
                  Panu Silvasti and
                  Seppo Sippu and
                  Eljas Soisalon{-}Soininen},
  title        = {Online Dictionary Matching with Variable-Length Gaps},
  booktitle    = {Proc.~of {SEA}},
  pages        = {76--87},
  year         = {2011},
  doi          = {10.1007/978-3-642-20662-7_7},
}

@article{DBLP:journals/algorithmica/AmirKLPPS19,
  author       = {Amihood Amir and
                  Tsvi Kopelowitz and
                  Avivit Levy and
                  Seth Pettie and
                  Ely Porat and
                  B. Riva Shalom},
  title        = {Mind the Gap! - Online Dictionary Matching with One Gap},
  journal      = {Algorithmica},
  volume       = {81},
  number       = {6},
  pages        = {2123--2157},
  year         = {2019},
  doi          = {10.1007/S00453-018-0526-2},
  bibsource    = {dblp computer science bibliography, https://dblp.org}
}

@article{DBLP:journals/ir/FredrikssonG08,
  author       = {Kimmo Fredriksson and
                  Szymon Grabowski},
  title        = {Efficient algorithms for pattern matching with general gaps, character
                  classes, and transposition invariance},
  journal      = {Inf. Retr.},
  volume       = {11},
  number       = {4},
  pages        = {335--357},
  year         = {2008},
  doi          = {10.1007/S10791-008-9054-Z},
  bibsource    = {dblp computer science bibliography, https://dblp.org}
}

@inproceedings{DBLP:conf/lata/FredrikssonG09,
  author       = {Kimmo Fredriksson and
                  Szymon Grabowski},
  title        = {Nested Counters in Bit-Parallel String Matching},
  booktitle    = {Proc.~of {LATA}},
  volume       = {5457},
  pages        = {338--349},
  publisher    = {Springer},
  year         = {2009},
  doi          = {10.1007/978-3-642-00982-2_29},
  bibsource    = {dblp computer science bibliography, https://dblp.org}
}

@inproceedings{Lee03,
author = {I.Lee and A.Apostolico and C.S.Iliopoulos and K.Park},
title = {Finding approximate occurrences of a pattern that contains gaps},
booktitle = {Proc.~of AWOCA},
year = {2003},
pages = {89–-100}
}

@article{DBLP:journals/jcb/MorgantePVZ05,
  author       = {Michele Morgante and
                  Alberto Policriti and
                  Nicola Vitacolonna and
                  Andrea Zuccolo},
  title        = {Structured Motifs Search},
  journal      = {J. Comput. Biol.},
  volume       = {12},
  number       = {8},
  pages        = {1065--1082},
  year         = {2005},
  doi          = {10.1089/CMB.2005.12.1065},
  bibsource    = {dblp computer science bibliography, https://dblp.org}
}

@inproceedings{DBLP:conf/cocoon/RahmanILMS06,
  author       = {M. Sohel Rahman and
                  Costas S. Iliopoulos and
                  Inbok Lee and
                  Manal Mohamed and
                  William F. Smyth},
  title        = {Finding Patterns with Variable Length Gaps or Don't Cares},
  booktitle    = {Proc.~of {COCOON}},
  pages        = {146--155},
  year         = {2006},
  doi          = {10.1007/11809678_17},
}

@article{10.5555/643002.643006,
author = {Maxime Crochemore and Costas Iliopoulos and Christos Makris and Wojciech Rytter and  Athanasios Tsakalidis and Kostas Tsichlas},
title = {Approximate string matching with gaps},
year = {2002},
issue_date = {Spring 2002},
publisher = {Publishing Association Nordic Journal of Computing},
address = {FIN},
volume = {9},
number = {1},
issn = {1236-6064},
journal = {Nordic J. of Computing},
month = mar,
pages = {54–65},
numpages = {12}
}

@article{BILLE201225,
title = {String matching with variable length gaps},
journal = {Theor. Comput. Sci},
volume = {443},
pages = {25-34},
year = {2012},
issn = {0304-3975},
doi = {https://doi.org/10.1016/j.tcs.2012.03.029},
author = {Philip Bille and Inge Li Gørtz and Hjalte Wedel Vildhøj and David Kofoed Wind}
}

@article{ILIOPOULOS2021104616,
title = {Efficient pattern matching in elastic-degenerate strings},
journal = {Inf. Comput.},
volume = {279},
pages = {104616},
year = {2021},
issn = {0890-5401},
doi = {https://doi.org/10.1016/j.ic.2020.104616},
author = {Costas S. Iliopoulos and Ritu Kundu and Solon P. Pissis},
}

@article{doi:10.1137/20M1368033,
author = {Bernardini, Giulia and Gawrychowski, Pawe\l{} and Pisanti, Nadia and Pissis, Solon P. and Rosone, Giovanna},
title = {Elastic-Degenerate String Matching via Fast Matrix Multiplication},
journal = {SIAM J. Comput},
volume = {51},
number = {3},
pages = {549-576},
year = {2022},
doi = {10.1137/20M1368033}
}

@inproceedings{DBLP:conf/biostec/ProchazkaCK021,
  author       = {Petr Proch{\'{a}}zka and
                  Ondrej Cvacho and
                  Lubos Krc{\'{a}}l and
                  Jan Holub},
  title        = {Backward Pattern Matching on Elastic Degenerate Strings},
  booktitle    = {Proc.~of {BIOSTEC}, Volume 3: BIOINFORMATICS,
                  Online Streaming},
  pages        = {50--59},
  publisher    = {{SCITEPRESS}},
  year         = {2021},
  doi          = {10.5220/0010243600500059},
  bibsource    = {dblp computer science bibliography, https://dblp.org}
}

@article{DBLP:journals/mst/BernardiniGPSSZ24,
  author       = {Giulia Bernardini and
                  Est{\'{e}}ban Gabory and
                  Solon P. Pissis and
                  Leen Stougie and
                  Michelle Sweering and
                  Wiktor Zuba},
  title        = {Elastic-Degenerate String Matching with 1 Error or Mismatch},
  journal      = {Theory Comput. Syst.},
  volume       = {68},
  number       = {5},
  pages        = {1442--1467},
  year         = {2024},
  doi          = {10.1007/S00224-024-10194-8},
}

@InProceedings{gawrychowski_et_al:LIPIcs.CPM.2025.29,
  author =	{Pawe{\l} Gawrychowski and Adam G\'{o}rkiewicz and Pola Marciniak and Solon P. Pissis and Karol Pokorski},
  title =	{{Faster Approximate Elastic-Degenerate String Matching - Part B}},
  booktitle =	{Proc.~of CPM},
  pages =	{29:1--29:21},
  ISBN =	{978-3-95977-369-0},
  ISSN =	{1868-8969},
  year =	{2025},
  URN =		{urn:nbn:de:0030-drops-231236},
  doi =		{10.4230/LIPIcs.CPM.2025.29}
}

@InProceedings{pissis_et_al:LIPIcs.CPM.2025.28,
  author =	{Solon P. Pissis and Jakub Radoszewski and Wiktor Zuba},
  title =	{{Faster Approximate Elastic-Degenerate String Matching - Part A}},
  booktitle =	{Proc.~of CPM},
  pages =	{28:1--28:19},
  ISBN =	{978-3-95977-369-0},
  ISSN =	{1868-8969},
  year =	{2025},
  doi =		{10.4230/LIPIcs.CPM.2025.28}
}

@InProceedings{grossi_et_al:LIPIcs.CPM.2017.9,
  author =	{Roberto Grossi and Costas S. Iliopoulos and Chang Liu and Nadia Pisanti and Solon P. Pissis and Ahmad Retha and Giovanna Rosone and Fatima Vayani and Luca Versari},
  title =	{{On-Line Pattern Matching on Similar Texts}},
  booktitle =	{Proc.~of CPM},
  pages =	{9:1--9:14},
  ISBN =	{978-3-95977-039-2},
  ISSN =	{1868-8969},
  year =	{2017},
  volume =	{78},
  doi =		{10.4230/LIPIcs.CPM.2017.9}
}

@article{BERNARDINI2020109,
title = {Approximate pattern matching on elastic-degenerate text},
journal = {Theor. Comput. Sci},
volume = {812},
pages = {109-122},
year = {2020},
note = {In memoriam Danny Breslauer (1968-2017)},
issn = {0304-3975},
doi = {https://doi.org/10.1016/j.tcs.2019.08.012},
author = {Giulia Bernardini and Nadia Pisanti and Solon P. Pissis and Giovanna Rosone}
}

@InProceedings{aoyama_et_al:LIPIcs.CPM.2018.9,
  author =	{Aoyama, Kotaro and Nakashima, Yuto and I, Tomohiro and Inenaga, Shunsuke and Bannai, Hideo and Takeda, Masayuki},
  title =	{{Faster Online Elastic Degenerate String Matching}},
  booktitle =	{Proc.~of CPM},
  pages =	{9:1--9:10},
  ISBN =	{978-3-95977-074-3},
  ISSN =	{1868-8969},
  year =	{2018},
  address =	{Dagstuhl, Germany},
  URN =		{urn:nbn:de:0030-drops-87016},
  doi =		{10.4230/LIPIcs.CPM.2018.9}
}

@article{DBLP:journals/tcs/GalM07,
  author       = {Anna G{\'{a}}l and
                  Peter Bro Miltersen},
  title        = {The cell probe complexity of succinct data structures},
  journal      = {Theor. Comput. Sci.},
  volume       = {379},
  number       = {3},
  pages        = {405--417},
  doi          = {10.1016/J.TCS.2007.02.047},
  year         = {2007}
}

@inproceedings{DBLP:conf/tcc/Corrigan-GibbsK19,
  author       = {Henry Corrigan{-}Gibbs and
                  Dmitry Kogan},
  title        = {The Function-Inversion Problem: Barriers and Opportunities},
  booktitle    = {Proc.~of {TCC} {(1)}},
  volume       = {11891},
  pages        = {393--421},
  publisher    = {Springer},
  doi          = {10.1007/978-3-030-36030-6_16},
  year         = {2019}
}

@Article{DBLP:journals/bioinformatics/NavarroSMG19,
  author    = {Gonzalo Navarro and Victor Sepulveda and Mauricio Mar{\'{\i}}n and Sen{\'{e}}n Gonz{\'{a}}lez},
  title     = {Compressed filesystem for managing large genome collections},
  journal   = {Bioinformatics},
  year      = {2019},
  volume    = {35},
  number    = {20},
  pages     = {4120--4128},
  bibsource = {dblp computer science bibliography, https://dblp.org},
  doi       = {10.1093/BIOINFORMATICS/BTZ192},
}

@Article{DBLP:journals/bmcgenomics/ValenzuelaNVPM18,
  author    = {Daniel Valenzuela and Tuukka Norri and Niko V{\"{a}}lim{\"{a}}ki and Esa Pitk{\"{a}}nen and Veli M{\"{a}}kinen},
  title     = {Towards pan-genome read alignment to improve variation calling},
  journal   = {{BMC} Genom.},
  year      = {2018},
  volume    = {19},
  number    = {{S2}},
  bibsource = {dblp computer science bibliography, https://dblp.org},
  doi       = {10.1186/S12864-018-4465-8},
}

@Article{DBLP:journals/bioinformatics/DeorowiczDG13,
  author    = {Sebastian Deorowicz and Agnieszka Danek and Szymon Grabowski},
  title     = {Genome compression: a novel approach for large collections},
  journal   = {Bioinformatics},
  year      = {2013},
  volume    = {29},
  number    = {20},
  pages     = {2572--2578},
  bibsource = {dblp computer science bibliography, https://dblp.org},
  doi       = {10.1093/BIOINFORMATICS/BTT460},
}

@Article{10.1093/bioinformatics/btr505,
  author     = {Sebastian Deorowicz and Szymon Grabowski},
  title      = {Robust relative compression of genomes with random access},
  journal    = {Bioinformatics},
  year       = {2011},
  volume     = {27},
  number     = {21},
  pages      = {2979–2986},
  month      = nov,
  issn       = {1367-4803},
  address    = {USA},
  doi        = {10.1093/bioinformatics/btr505},
  issue_date = {November 2011},
  numpages   = {8},
  publisher  = {Oxford University Press, Inc.},
}

@InProceedings{10.1007/978-3-642-16321-0_20,
  author    = {Shanika Kuruppu and Simon J. Puglisi and Justin Zobel},
  title     = {Relative Lempel-Ziv Compression of Genomes for Large-Scale Storage and Retrieval},
  booktitle = {Proc.~of {SPIRE}},
  year      = {2010},
  pages     = {201--206},
  bibsource = {dblp computer science bibliography, https://dblp.org},
  doi       = {10.1007/978-3-642-16321-0_20},
}

@InCollection{Iliopoulos2017,
  author    = {Iliopoulos, Costas S. and Pissis, Solon P. and Rahman, M. Sohel},
  title     = {Searching and Indexing Circular Patterns},
  booktitle = {Algorithms for Next-Generation Sequencing Data: Techniques, Approaches, and Applications},
  publisher = {Springer},
  year      = {2017},
  pages     = {77--90},
  isbn      = {978-3-319-59826-0},
  doi       = {10.1007/978-3-319-59826-0_3},
}

@Article{FREDRIKSSON2009579,
  author  = {Kimmo Fredriksson and Szymon Grabowski},
  title   = {Average-optimal string matching},
  journal = {J. Discrete Algorithms},
  year    = {2009},
  volume  = {7},
  number  = {4},
  pages   = {579-594},
  issn    = {1570-8667},
  doi     = {https://doi.org/10.1016/j.jda.2008.09.001},
}

@InProceedings{10.1007/978-3-319-02309-0_59,
  author    = {Robert Susik and Szymon Grabowski and Sebastian Deorowicz},
  title     = {Fast and Simple Circular Pattern Matching},
  booktitle = {Proc.~of {ICMMI} (Man-Machine Interactions 3)},
  year      = {2013},
  pages     = {537--544},
  bibsource = {dblp computer science bibliography, https://dblp.org},
  doi       = {10.1007/978-3-319-02309-0_59},
}

@Article{10.1093/comjnl/bxt023,
  author   = {Kuei-Hao Chen and Guan-Shieng Huang and Richard Chia-Tung Lee},
  title    = {{Bit-Parallel Algorithms for Exact Circular String Matching}},
  journal  = {Comput. J.},
  year     = {2013},
  volume   = {57},
  number   = {5},
  pages    = {731-743},
  month    = {03},
  issn     = {0010-4620},
  doi      = {10.1093/comjnl/bxt023},
}

@InProceedings{10.1007/978-3-642-25591-5_69,
  author    = {Hon, Wing-Kai and Lu, Chen-Hua and Shah, Rahul and Thankachan, Sharma V.},
  title     = {Succinct Indexes for Circular Patterns},
  booktitle = {Proc.~of ISAAC},
  year      = {2011},
  pages     = {673--682},
  doi       = {10.1007/978-3-642-25591-5_69},
  isbn      = {978-3-642-25591-5},
}

@InProceedings{10.1007/978-3-642-38905-4_15,
  author    = {Hon, Wing-Kai and Ku, Tsung-Han and Shah, Rahul and Thankachan, Sharma V.},
  title     = {Space-Efficient Construction Algorithm for the Circular Suffix Tree},
  booktitle = {Proc.~of CPM},
  year      = {2013},
  pages     = {142--152},
  doi       = {10.1007/978-3-642-38905-4_15},
  isbn      = {978-3-642-38905-4},
}

@Article{DBLP:journals/jcss/Charalampopoulos21,
  author    = {Panagiotis Charalampopoulos and Tomasz Kociumaka and Solon P. Pissis and Jakub Radoszewski and Wojciech Rytter and Juliusz Straszynski and Tomasz Walen and Wiktor Zuba},
  title     = {Circular pattern matching with \emph{k} mismatches},
  journal   = {J. Comput. Syst. Sci.},
  year      = {2021},
  volume    = {115},
  pages     = {73--85},
  bibsource = {dblp computer science bibliography, https://dblp.org},
  doi       = {10.1016/J.JCSS.2020.07.003},
}

@InProceedings{DBLP:conf/stoc/Yao90,
  author      = {Andrew Chi{-}Chih Yao},
  title       = {Coherent Functions and Program Checkers (Extended Abstract)},
  booktitle   = {Proc.~of {STOC}},
  year        = {1990},
  pages       = {84--94},
  bibsource   = {dblp computer science bibliography, https://dblp.org},
  doi         = {10.1145/100216.100226},
  rmpublisher = {{ACM}},
}

@InProceedings{asymED,
  author      = {Andoni, Alexandr and Krauthgamer, Robert and Onak, Krzysztof},
  title       = {Polylogarithmic Approximation for Edit Distance and the Asymmetric Query Complexity},
  booktitle   = {Proc.~of {FOCS}},
  year        = {2010},
  pages       = {377-386},
  doi         = {10.1109/FOCS.2010.43},
  issn        = {0272-5428},
  rmpublisher = {{IEEE} Computer Society},
}

@InProceedings{DBLP:conf/stoc/KempaK19,
  author      = {Dominik Kempa and Tomasz Kociumaka},
  title       = {String synchronizing sets: sublinear-time {BWT} construction and optimal {LCE} data structure},
  booktitle   = {Proc.~of {STOC}},
  year        = {2019},
  pages       = {756--767},
  doi         = {10.1145/3313276.3316368},
  rmpublisher = {{ACM}},
}

@Article{DBLP:journals/siamdm/SchmidtSS95,
  author    = {Jeanette P. Schmidt and Alan Siegel and Aravind Srinivasan},
  title     = {Chernoff--{H}oeffding Bounds for Applications with Limited Independence},
  journal   = {{SIAM} J. Discrete Math.},
  year      = {1995},
  volume    = {8},
  number    = {2},
  pages     = {223--250},
  bibsource = {dblp computer science bibliography, https://dblp.org},
  doi       = {10.1137/S089548019223872X},
}

@InProceedings{keditCPM,
  author    = {Panagiotis Charalampopoulos and Solon P. Pissis and Jakub Radoszewski and Wojciech Rytter and Tomasz Walen and Wiktor Zuba},
  title     = {Approximate Circular Pattern Matching Under Edit Distance},
  booktitle = {Proc.~of {STACS}},
  year      = {2024},
  pages     = {24:1--24:22},
  doi       = {10.4230/LIPICS.STACS.2024.24},
}

@Book{Gusfield_1997,
  title     = {Algorithms on Strings, Trees, and Sequences: Computer Science and Computational Biology},
  publisher = {Cambridge University Press},
  year      = {1997},
  author    = {Gusfield, Dan},
  doi       = {10.1017/CBO9780511574931},
  place     = {Cambridge},
}

@Article{AKS,
  author    = {Manindra Agrawal and Neeraj Kayal and Nitin Saxena},
  title     = {{PRIMES} is in {P}},
  journal   = {Ann. Math.},
  year      = {2004},
  volume    = {160},
  number    = {2},
  pages     = {781--793},
  issn      = {0003486X},
  doi       = {10.4007/annals.2004.160.781},
  publisher = {Annals of Mathematics},
}

@Article{DBLP:journals/tcs/KidaMSTSA03,
  author    = {Takuya Kida and Tetsuya Matsumoto and Yusuke Shibata and Masayuki Takeda and Ayumi Shinohara and Setsuo Arikawa},
  title     = {Collage system: a unifying framework for compressed pattern matching},
  journal   = {Theor. Comput. Sci.},
  year      = {2003},
  volume    = {298},
  number    = {1},
  pages     = {253--272},
  bibsource = {dblp computer science bibliography, https://dblp.org},
  doi       = {10.1016/S0304-3975(02)00426-7},
}

@InProceedings{DBLP:conf/cpm/MiyazakiST97,
  author    = {Masamichi Miyazaki and Ayumi Shinohara and Masayuki Takeda},
  title     = {An Improved Pattern Matching Algorithm for Strings in Terms of Straight-Line Programs},
  booktitle = {Proc.~of {CPM}},
  year      = {1997},
  pages     = {1--11},
  doi       = {10.1007/3-540-63220-4_45},
}

@Article{CROCHEMORE199233,
  author  = {Maxime Crochemore},
  title   = {String-matching on ordered alphabets},
  journal = {Theor. Comput. Sci.},
  year    = {1992},
  volume  = {92},
  number  = {1},
  pages   = {33-47},
  issn    = {0304-3975},
  doi     = {10.1016/0304-3975(92)90134-2},
}

@InProceedings{DBLP:conf/esa/0001GGK15,
  author      = {Johannes Fischer and Travis Gagie and Pawe\l{} Gawrychowski and Tomasz Kociumaka},
  title       = {Approximating {LZ77} via Small-Space Multiple-Pattern Matching},
  booktitle   = {Proc.~of {ESA}},
  year        = {2015},
  pages       = {533--544},
  doi         = {10.1007/978-3-662-48350-3_45},
  rmpublisher = {Springer},
  rmvolume    = {9294},
note  = {Full version available at arXiv:1504.06647}
}

@InProceedings{DBLP:conf/focs/KolpakovK99,
  author      = {Roman M. Kolpakov and Gregory Kucherov},
  title       = {Finding Maximal Repetitions in a Word in Linear Time},
  booktitle   = {Proc.~of {FOCS}},
  year        = {1999},
  pages       = {596--604},
  doi         = {10.1109/SFFCS.1999.814634},
  rmpublisher = {{IEEE} Computer Society},
}

@InProceedings{porat2009optimal,
  author      = {Benny Porat and Ely Porat},
  title       = {Exact and Approximate Pattern Matching in the Streaming Model},
  booktitle   = {Proc.~of {FOCS}},
  year        = {2009},
  pages       = {315--323},
  bibsource   = {dblp computer science bibliography, https://dblp.org},
  doi         = {10.1109/FOCS.2009.11},
  rmpublisher = {{IEEE} Computer Society},
}

@InProceedings{DBLP:conf/esa/Charalampopoulos22,
  author    = {Panagiotis Charalampopoulos and Tomasz Kociumaka and Jakub Radoszewski and Solon P. Pissis and Wojciech Rytter and Tomasz Walen and Wiktor Zuba},
  title     = {Approximate Circular Pattern Matching},
  booktitle = {Proc.~of {ESA}},
  year      = {2022},
  pages     = {35:1--35:19},
  note      = {Full version: \url{https://doi.org/10.48550/arXiv.2208.08915}},
  doi       = {10.4230/LIPICS.ESA.2022.35},
  rmseries  = {LIPIcs},
  rmvolume  = {244},
}

@InProceedings{belazzougui_et_al:LIPIcs.CPM.2021.8,
  author    = {Belazzougui, Djamal and Kosolobov, Dmitry and Puglisi, Simon J. and Raman, Rajeev},
  title     = {Weighted Ancestors in Suffix Trees Revisited},
  booktitle = {Proc.~of {CPM}},
  year      = {2021},
  pages     = {8:1--8:15},
  doi       = {10.4230/LIPIcs.CPM.2021.8},
  isbn      = {978-3-95977-186-3},
  issn      = {1868-8969},
  rmseries  = {LIPIcs},
  rmvolume  = {191},
  urn       = {urn:nbn:de:0030-drops-139594},
}

@InProceedings{DBLP:conf/isaac/BathieKS23,
  author    = {Gabriel Bathie and Tomasz Kociumaka and Tatiana Starikovskaya},
  title     = {Small-Space Algorithms for the Online Language Distance Problem for Palindromes and Squares},
  booktitle = {Proc.~of {ISAAC}},
  year      = {2023},
  pages     = {10:1--10:17},
  doi       = {10.4230/LIPICS.ISAAC.2023.10},
  rmseries  = {LIPIcs},
  rmvolume  = {283},
}

@Article{DBLP:journals/eccc/DeTT09,
  author    = {Anindya De and Luca Trevisan and Madhur Tulsiani},
  title     = {Non-uniform attacks against one-way functions and {PRG}s},
  journal   = {Electron. Colloquium Comput. Complex.},
  year      = {2009},
  volume    = {{TR09-113}},
  bibsource = {dblp computer science bibliography, https://dblp.org},
  url       = {https://eccc.weizmann.ac.il/report/2009/113},
}

@Article{fine1965uniqueness,
  author  = {Nathan J. Fine and Herbert S. Wilf},
  title   = {Uniqueness Theorems for Periodic Functions},
  journal = {Proc.~Am. Math. Soc.},
  year    = {1965},
  volume  = {16},
  pages   = {109-114},
  doi     = {10.1090/S0002-9939-1965-0174934-9},
}

@article{DBLP:journals/talg/BreslauerG14,
  author       = {Dany Breslauer and
                  Roberto Grossi and
                  Filippo Mignosi},
  title        = {Simple real-time constant-space string matching},
  journal      = {Theor. Comput. Sci.},
  volume       = {483},
  pages        = {2--9},
  year         = {2013},
  doi          = {10.1016/J.TCS.2012.11.040},
}

@Article{doi:10.1137/S0097539795280512,
  author  = {Fiat, Amos and Naor, Moni},
  title   = {Rigorous Time/Space Trade-offs for Inverting Functions},
  journal = {{SIAM} J. Comput.},
  year    = {2000},
  volume  = {29},
  number  = {3},
  pages   = {790-803},
  doi     = {10.1137/S0097539795280512},
}

@Article{karp1987efficient,
  author    = {Richard M. Karp and Michael O. Rabin},
  title     = {Efficient Randomized Pattern-Matching Algorithms},
  journal   = {{IBM} J. Res. Dev.},
  year      = {1987},
  volume    = {31},
  number    = {2},
  pages     = {249--260},
  bibsource = {dblp computer science bibliography, https://dblp.org},
  doi       = {10.1147/RD.312.0249},
}

@InProceedings{DBLP:conf/esa/CliffordFPSS15,
  author    = {Rapha\"{e}l Clifford and Allyx Fontaine and Ely Porat and Benjamin Sach and Tatiana Starikovskaya},
  title     = {Dictionary Matching in a Stream},
  booktitle = {Proc.~of {ESA}},
  year      = {2015},
  pages     = {361--372},
  doi       = {10.1007/978-3-662-48350-3_31},
  rmseries  = {LNCS},
  rmvolume  = {9294},
}

@InProceedings{DBLP:conf/esa/GolanP17,
  author    = {Shay Golan and Ely Porat},
  title     = {Real-Time Streaming Multi-Pattern Search for Constant Alphabet},
  booktitle = {Proc.~of {ESA}},
  year      = {2017},
  pages     = {41:1--41:15},
  doi       = {10.4230/LIPIcs.ESA.2017.41},
  rmseries  = {LIPIcs},
  rmvolume  = {107},
}

@InProceedings{DBLP:conf/soda/CliffordFPSS16,
  author    = {Rapha\"{e}l Clifford and Allyx Fontaine and Ely Porat and Benjamin Sach and Tatiana Starikovskaya},
  title     = {The \emph{k}-mismatch problem revisited},
  booktitle = {Proc.~of {SODA}},
  year      = {2016},
  pages     = {2039--2052},
  doi       = {10.1137/1.9781611974331.ch142},
}

@InProceedings{DBLP:conf/icalp/GolanKP18,
  author    = {Shay Golan and Tsvi Kopelowitz and Ely Porat},
  title     = {Towards Optimal Approximate Streaming Pattern Matching by Matching Multiple Patterns in Multiple Streams},
  booktitle = {Proc.~of {ICALP}},
  year      = {2018},
  pages     = {65:1--65:16},
  doi       = {10.4230/LIPIcs.ICALP.2018.65},
  rmseries  = {LIPIcs},
  rmvolume  = {107},
}

@InProceedings{starikovskaya:LIPIcs:2017:7320,
  author    = {Tatiana Starikovskaya},
  title     = {Communication and Streaming Complexity of Approximate Pattern Matching},
  booktitle = {Proc.~of {CPM}},
  year      = {2017},
  pages     = {13:1--13:11},
  doi       = {10.4230/LIPIcs.CPM.2017.13},
  isbn      = {978-3-95977-039-2},
  issn      = {1868-8969},
  rmseries  = {LIPIcs},
  rmvolume  = {78},
}

@InProceedings{clifford2018streaming,
  author      = {Rapha{\"{e}}l Clifford and Tomasz Kociumaka and Ely Porat},
  title       = {The streaming k-mismatch problem},
  booktitle   = {Proc.~of {SODA}},
  year        = {2019},
  pages       = {1106--1125},
  bibsource   = {dblp computer science bibliography, https://dblp.org},
  doi         = {10.1137/1.9781611975482.68},
  rmpublisher = {{SIAM}},
}

@InProceedings{DBLP:conf/soda/DudekGGS22,
  author      = {Bartlomiej Dudek and Pawe\l{} Gawrychowski and Garance Gourdel and Tatiana Starikovskaya},
  title       = {Streaming Regular Expression Membership and Pattern Matching},
  booktitle   = {Proc.~of {SODA}},
  year        = {2022},
  pages       = {670--694},
  bibsource   = {dblp computer science bibliography, https://dblp.org},
  doi         = {10.1137/1.9781611977073.30},
  rmpublisher = {{SIAM}},
}

@InProceedings{DBLP:conf/icalp/Bhattacharya023,
  author      = {Sudatta Bhattacharya and Michal Kouck{\'{y}}},
  title       = {Streaming k-Edit Approximate Pattern Matching via String Decomposition},
  booktitle   = {Proc.~of {ICALP}},
  year        = {2023},
  pages       = {22:1--22:14},
  bibsource   = {dblp computer science bibliography, https://dblp.org},
  doi         = {10.4230/LIPICS.ICALP.2023.22},
  rmpublisher = {Schloss Dagstuhl - Leibniz-Zentrum f{\"{u}}r Informatik},
  rmseries    = {LIPIcs},
  rmvolume    = {261},
}

@Article{DBLP:journals/iandc/RadoszewskiS20,
  author    = {Jakub Radoszewski and Tatiana Starikovskaya},
  title     = {Streaming \emph{k}-mismatch with error correcting and applications},
  journal   = {Inf. Comput.},
  year      = {2020},
  volume    = {271},
  pages     = {104513},
  bibsource = {dblp computer science bibliography, https://dblp.org},
  doi       = {10.1016/j.ic.2019.104513},
}

@InProceedings{DBLP:conf/cpm/GolanKKP20,
  author    = {Shay Golan and Tomasz Kociumaka and Tsvi Kopelowitz and Ely Porat},
  title     = {The Streaming k-Mismatch Problem: Tradeoffs Between Space and Total Time},
  booktitle = {Proc.~of {CPM}},
  year      = {2020},
  pages     = {15:1--15:15},
  bibsource = {dblp computer science bibliography, https://dblp.org},
  doi       = {10.4230/LIPIcs.CPM.2020.15},
  rmseries  = {LIPIcs},
  rmvolume  = {161},
}

@InProceedings{DBLP:conf/focs/KociumakaPS21,
  author      = {Tomasz Kociumaka and Ely Porat and Tatiana Starikovskaya},
  title       = {Small-space and streaming pattern matching with k edits},
  booktitle   = {Proc.~of {FOCS}},
  year        = {2021},
  pages       = {885--896},
  bibsource   = {dblp computer science bibliography, https://dblp.org},
  doi         = {10.1109/FOCS52979.2021.00090},
  rmpublisher = {{IEEE}},
}

@Article{DBLP:journals/algorithmica/GawrychowskiS22,
  author    = {Pawe\l{} Gawrychowski and Tatiana Starikovskaya},
  title     = {Streaming Dictionary Matching with Mismatches},
  journal   = {Algorithmica},
  year      = {2022},
  volume    = {84},
  number    = {4},
  pages     = {896--916},
  bibsource = {dblp computer science bibliography, https://dblp.org},
  doi       = {10.1007/S00453-021-00876-X},
}

@InProceedings{Ergun:10,
  author    = {Funda Erg{\"{u}}n and Hossein Jowhari and Mert Saglam},
  title     = {Periodicity in Streams},
  booktitle = {Proc.~of {APPROX-RANDOM}},
  year      = {2010},
  volume    = {6302},
  pages     = {545--559},
  publisher = {Springer},
  bibsource = {dblp computer science bibliography, https://dblp.org},
  doi       = {10.1007/978-3-642-15369-3_41},
}

@InProceedings{stream-periodicity-mismatches,
  author      = {Funda Erg{\"{u}}n and Elena Grigorescu and Erfan Sadeqi Azer and Samson Zhou},
  title       = {Streaming Periodicity with Mismatches},
  booktitle   = {Proc.~of {APPROX-RANDOM}},
  year        = {2017},
  pages       = {42:1--42:21},
  bibsource   = {dblp computer science bibliography, https://dblp.org},
  doi         = {10.4230/LIPICS.APPROX-RANDOM.2017.42},
  rmpublisher = {Schloss Dagstuhl - Leibniz-Zentrum f{\"{u}}r Informatik},
  rmseries    = {LIPIcs},
  rmvolume    = {81},
}

@InProceedings{stream-periodicity-wildcards,
  author    = {Funda Erg{\"{u}}n and Elena Grigorescu and Erfan Sadeqi Azer and Samson Zhou},
  title     = {Periodicity in Data Streams with Wildcards},
  booktitle = {Proc.~of {CSR}},
  year      = {2018},
  pages     = {90--105},
  doi       = {10.1007/978-3-319-90530-3_9},
}

@Article{DBLP:journals/algorithmica/GawrychowskiMSU19,
  author  = {Pawe\l{} Gawrychowski and Oleg Merkurev and Arseny M. Shur and Przemyslaw Uzna\'{n}ski},
  title   = {Tight Tradeoffs for Real-Time Approximation of Longest Palindromes in Streams},
  journal = {Algorithmica},
  year    = {2019},
  volume  = {81},
  number  = {9},
  pages   = {3630--3654},
  doi     = {10.1007/s00453-019-00591-8},
}

@InProceedings{DBLP:conf/cpm/MerkurevS19,
  author      = {Oleg Merkurev and Arseny M. Shur},
  title       = {Searching Long Repeats in Streams},
  booktitle   = {Proc.~of {CPM}},
  year        = {2019},
  pages       = {31:1--31:14},
  bibsource   = {dblp computer science bibliography, https://dblp.org},
  doi         = {10.4230/LIPICS.CPM.2019.31},
  rmpublisher = {Schloss Dagstuhl - Leibniz-Zentrum f{\"{u}}r Informatik},
  rmseries    = {LIPIcs},
  rmvolume    = {128},
}

@InProceedings{doi:10.1137/1.9781611973105.122,
  author      = {Michael E. Saks and C. Seshadhri},
  title       = {Space efficient streaming algorithms for the distance to monotonicity and asymmetric edit distance},
  booktitle   = {Proc.~of {SODA}},
  year        = {2013},
  pages       = {1698--1709},
  bibsource   = {dblp computer science bibliography, https://dblp.org},
  doi         = {10.1137/1.9781611973105.122},
  rmpublisher = {{SIAM}},
}

@InProceedings{DBLP:conf/esa/KociumakaSV14,
  author      = {Tomasz Kociumaka and Tatiana Starikovskaya and Hjalte Wedel Vildh{\o}j},
  title       = {Sublinear Space Algorithms for the Longest Common Substring Problem},
  booktitle   = {Proc.~of {ESA}},
  year        = {2014},
  pages       = {605--617},
  bibsource   = {dblp computer science bibliography, https://dblp.org},
  doi         = {10.1007/978-3-662-44777-2_50},
  rmpublisher = {Springer},
  rmvolume    = {8737},
}

@InProceedings{7354391,
  author    = {Saha, Barna},
  title     = {Language Edit Distance and Maximum Likelihood Parsing of Stochastic Grammars: Faster Algorithms and Connection to Fundamental Graph Problems},
  booktitle = {Proc.~of {FOCS}},
  year      = {2015},
  pages     = {118-135},
  doi       = {10.1109/FOCS.2015.17},
}

@InProceedings{DBLP:conf/cpm/BathieCS24,
  author      = {Gabriel Bathie and Panagiotis Charalampopoulos and Tatiana Starikovskaya},
  title       = {Internal Pattern Matching in Small Space and Applications},
  booktitle   = {Proc.~of {CPM}},
  year        = {2024},
  pages       = {4:1--4:20},
  bibsource   = {dblp computer science bibliography, https://dblp.org},
  doi         = {10.4230/LIPICS.CPM.2024.4},
  rmpublisher = {Schloss Dagstuhl - Leibniz-Zentrum f{\"{u}}r Informatik},
  rmseries    = {LIPIcs},
  rmvolume    = {296},
}

@Article{DBLP:journals/theoretics/GanardiHLMS25,
  author    = {Moses Ganardi and Danny Hucke and Markus Lohrey and Konstantinos Mamouras and Tatiana Starikovskaya},
  title     = {Regular Languages in the Sliding Window Model},
  journal   = {TheoretiCS},
  year      = {2025},
  volume    = {4},
  bibsource = {dblp computer science bibliography, https://dblp.org},
  doi       = {10.46298/THEORETICS.25.8},
}

@InProceedings{ganardi_et_al:LIPIcs:2016:6853,
  author    = {Moses Ganardi and Danny Hucke and Markus Lohrey},
  title     = {Querying Regular Languages over Sliding Windows},
  booktitle = {Proc.~of {FST\,TCS}},
  year      = {2016},
  pages     = {18:1--18:14},
  doi       = {10.4230/LIPIcs.FSTTCS.2016.18},
  isbn      = {978-3-95977-027-9},
  issn      = {1868-8969},
  rmseries  = {LIPIcs},
  rmvolume  = {65},
}

@InProceedings{DBLP:conf/mfcs/GanardiJL18,
  author    = {Moses Ganardi and Artur Je\.{z} and Markus Lohrey},
  title     = {Sliding Windows over Context-Free Languages},
  booktitle = {Proc.~of {MFCS}},
  year      = {2018},
  pages     = {15:1--15:15},
  doi       = {10.4230/LIPIcs.MFCS.2018.15},
  rmseries  = {LIPIcs},
  rmvolume  = {117},
}

@Article{DBLP:journals/tcs/BabuLRV13,
  author  = {Ajesh Babu and Nutan Limaye and Jaikumar Radhakrishnan and Girish Varma},
  title   = {Streaming algorithms for language recognition problems},
  journal = {Theor. Comput. Sci.},
  year    = {2013},
  volume  = {494},
  pages   = {13--23},
  doi     = {10.1016/J.TCS.2012.12.028},
}

@InProceedings{DBLP:conf/cpm/GawrychowskiRS19,
  author    = {Pawe\l{} Gawrychowski and Jakub Radoszewski and Tatiana Starikovskaya},
  title     = {Quasi-Periodicity in Streams},
  booktitle = {Proc.~of {CPM}},
  year      = {2019},
  pages     = {22:1--22:14},
  bibsource = {dblp computer science bibliography, https://dblp.org},
  doi       = {10.4230/LIPIcs.CPM.2019.22},
  rmseries  = {LIPIcs},
  rmvolume  = {128},
}

@Article{DBLP:journals/siamcomp/MagniezMN14,
  author  = {Fr\'{e}d\'{e}ric Magniez and Claire Mathieu and Ashwin Nayak},
  title   = {Recognizing Well-Parenthesized Expressions in the Streaming Model},
  journal = {{SIAM} J. Comput.},
  year    = {2014},
  volume  = {43},
  number  = {6},
  pages   = {1880--1905},
  doi     = {10.1137/130926122},
}

@InProceedings{ganardi_et_al:LIPIcs:2018:9131,
  author    = {Moses Ganardi and Danny Hucke and Markus Lohrey},
  title     = {Randomized Sliding Window Algorithms for Regular Languages},
  booktitle = {Proc.~of {ICALP}},
  year      = {2018},
  pages     = {127:1--127:13},
  doi       = {10.4230/LIPIcs.ICALP.2018.127},
  isbn      = {978-3-95977-076-7},
  issn      = {1868-8969},
  rmseries  = {LIPIcs},
  rmvolume  = {107},
}

@InProceedings{franois_et_al:LIPIcs:2016:6355,
  author      = {Nathana{\"{e}}l Fran{\c{c}}ois and Fr{\'{e}}d{\'{e}}ric Magniez and Michel de Rougemont and Olivier Serre},
  title       = {Streaming Property Testing of Visibly Pushdown Languages},
  booktitle   = {Proc.~of {ESA}},
  year        = {2016},
  pages       = {43:1--43:17},
  bibsource   = {dblp computer science bibliography, https://dblp.org},
  doi         = {10.4230/LIPICS.ESA.2016.43},
  rmpublisher = {Schloss Dagstuhl - Leibniz-Zentrum f{\"{u}}r Informatik},
  rmseries    = {LIPIcs},
  rmvolume    = {57},
}

@InProceedings{DBLP:conf/icalp/BathieS21,
  author      = {Gabriel Bathie and Tatiana Starikovskaya},
  title       = {Property Testing of Regular Languages with Applications to Streaming Property Testing of Visibly Pushdown Languages},
  booktitle   = {Proc.~of {ICALP}},
  year        = {2021},
  pages       = {119:1--119:17},
  bibsource   = {dblp computer science bibliography, https://dblp.org},
  doi         = {10.4230/LIPICS.ICALP.2021.119},
  rmpublisher = {Schloss Dagstuhl - Leibniz-Zentrum f{\"{u}}r Informatik},
  rmseries    = {LIPIcs},
  rmvolume    = {198},
}

@InProceedings{ganardi_et_al:LIPIcs:2019:11502,
  author    = {Moses Ganardi and Danny Hucke and Markus Lohrey and Tatiana Starikovskaya},
  title     = {Sliding Window Property Testing for Regular Languages},
  booktitle = {Proc.~of {ISAAC}},
  year      = {2019},
  pages     = {6:1--6:13},
  doi       = {10.4230/LIPIcs.ISAAC.2019.6},
  isbn      = {978-3-95977-130-6},
  issn      = {1868-8969},
  rmseries  = {LIPIcs},
  rmvolume  = {149},
}

@InProceedings{li_et_al:LIPIcs.FSTTCS.2021.27,
  author      = {Xin Li and Yu Zheng},
  title       = {Lower Bounds and Improved Algorithms for Asymmetric Streaming Edit Distance and Longest Common Subsequence},
  booktitle   = {Proc.~of {FST\,TCS}},
  year        = {2021},
  pages       = {27:1--27:23},
  bibsource   = {dblp computer science bibliography, https://dblp.org},
  doi         = {10.4230/LIPICS.FSTTCS.2021.27},
  rmpublisher = {Schloss Dagstuhl - Leibniz-Zentrum f{\"{u}}r Informatik},
  rmseries    = {LIPIcs},
  rmvolume    = {213},
}

@InProceedings{DBLP:conf/lata/GanardiHL18,
  author    = {Moses Ganardi and Danny Hucke and Markus Lohrey},
  title     = {Sliding Window Algorithms for Regular Languages},
  booktitle = {Proc.~of {LATA}},
  year      = {2018},
  pages     = {26--35},
  doi       = {10.1007/978-3-319-77313-1_2},
  rmvolume  = {10792},
}

@InProceedings{ganardi_et_al:LIPIcs:2018:8485,
  author    = {Moses Ganardi and Danny Hucke and Daniel K{\"o}nig and Markus Lohrey and Konstantinos Mamouras},
  title     = {Automata Theory on Sliding Windows},
  booktitle = {Proc.~of {STACS}},
  year      = {2018},
  pages     = {31:1--31:14},
  doi       = {10.4230/LIPIcs.STACS.2018.31},
  isbn      = {978-3-95977-062-0},
  issn      = {1868-8969},
  rmseries  = {LIPIcs},
  rmvolume  = {96},
}

@InProceedings{cheng_et_al:LIPIcs.ICALP.2021.54,
  author      = {Kuan Cheng and Alireza Farhadi and MohammadTaghi Hajiaghayi and Zhengzhong Jin and Xin Li and Aviad Rubinstein and Saeed Seddighin and Yu Zheng},
  title       = {Streaming and Small Space Approximation Algorithms for Edit Distance and Longest Common Subsequence},
  booktitle   = {Proc.~of {ICALP}},
  year        = {2021},
  pages       = {54:1--54:20},
  bibsource   = {dblp computer science bibliography, https://dblp.org},
  doi         = {10.4230/LIPICS.ICALP.2021.54},
  rmpublisher = {Schloss Dagstuhl - Leibniz-Zentrum f{\"{u}}r Informatik},
  rmseries    = {LIPIcs},
  rmvolume    = {198},
}

@InProceedings{DBLP:conf/spire/MerkurevS19,
  author      = {Oleg Merkurev and Arseny M. Shur},
  title       = {Searching Runs in Streams},
  booktitle   = {Proc.~of {SPIRE}},
  year        = {2019},
  pages       = {203--220},
  bibsource   = {dblp computer science bibliography, https://dblp.org},
  doi         = {10.1007/978-3-030-32686-9_15},
  rmpublisher = {Springer},
  rmvolume    = {11811},
}

@Article{DBLP:journals/jcss/FredmanW93,
  author    = {Michael L. Fredman and Dan E. Willard},
  title     = {Surpassing the Information Theoretic Bound with Fusion Trees},
  journal   = {J. Comput. Syst. Sci.},
  year      = {1993},
  volume    = {47},
  number    = {3},
  pages     = {424--436},
  bibsource = {dblp computer science bibliography, https://dblp.org},
  doi       = {10.1016/0022-0000(93)90040-4},
}

@Article{DBLP:journals/iandc/CliffordEPP11,
  author    = {Rapha{\"{e}}l Clifford and Klim Efremenko and Benny Porat and Ely Porat},
  title     = {A black box for online approximate pattern matching},
  journal   = {Inf. Comput.},
  year      = {2011},
  volume    = {209},
  number    = {4},
  pages     = {731--736},
  doi       = {10.1016/J.IC.2010.12.007},
}

@article{DBLP:journals/siamcomp/BorodinC82,
  author       = {Allan Borodin and
                  Stephen A. Cook},
  title        = {A Time-Space Tradeoff for Sorting on a General Sequential Model of
                  Computation},
  journal      = {{SIAM} J. Comput.},
  volume       = {11},
  number       = {2},
  pages        = {287--297},
  year         = {1982},
  doi          = {10.1137/0211022},
}

@InProceedings{DBLP:conf/spire/Ellert23,
  author    = {Jonas Ellert},
  title     = {Sublinear Time {L}empel-{Z}iv {(LZ77)} Factorization},
  booktitle = {Proc.~of {SPIRE}},
  year      = {2023},
  pages     = {171--187},
  bibsource = {dblp computer science bibliography, https://dblp.org},
  doi       = {10.1007/978-3-031-43980-3_14},
}

@InProceedings{DBLP:conf/focs/KempaK24,
  author    = {Dominik Kempa and Tomasz Kociumaka},
  title     = {Lempel-{Z}iv {(LZ77)} Factorization in Sublinear Time},
  booktitle = {Proc.~of {FOCS}},
  year      = {2024},
  pages     = {2045--2055},
  bibsource = {dblp computer science bibliography, https://dblp.org},
  doi       = {10.1109/FOCS61266.2024.00122},
}

@InProceedings{DBLP:conf/esa/Charalampopoulos21,
  author    = {Panagiotis Charalampopoulos and Tomasz Kociumaka and Solon P. Pissis and Jakub Radoszewski},
  title     = {Faster Algorithms for Longest Common Substring},
  booktitle = {Proc.~of {ESA}},
  year      = {2021},
  pages     = {30:1--30:17},
  bibsource = {dblp computer science bibliography, https://dblp.org},
  doi       = {10.4230/LIPICS.ESA.2021.30},
}

@InProceedings{DBLP:conf/soda/KempaK23,
  author    = {Dominik Kempa and Tomasz Kociumaka},
  booktitle = {Proc.~of {SODA}},
  title     = {Breaking the ${O}(n)$-Barrier in the Construction of Compressed Suffix Arrays and Suffix Trees},
  year      = {2023},
  pages     = {5122--5202},
  bibsource = {dblp computer science bibliography, https://dblp.org},
  doi       = {10.1137/1.9781611977554.CH187},
}

@InProceedings{DBLP:conf/mfcs/Charalampopoulos25,
  author    = {Panagiotis Charalampopoulos and Manal Mohamed and Jakub Radoszewski and Wojciech Rytter and Tomasz Walen and Wiktor Zuba},
  title     = {Counting Distinct Square Substrings in Sublinear Time},
  booktitle = {Proc.~of {MFCS}},
  year      = {2025},
  pages     = {36:1--36:19},
  bibsource = {dblp computer science bibliography, https://dblp.org},
  doi       = {10.4230/LIPICS.MFCS.2025.36},
}

@InProceedings{DBLP:conf/cpm/Charalampopoulos22,
  author    = {Panagiotis Charalampopoulos and Solon P. Pissis and Jakub Radoszewski},
  title     = {Longest Palindromic Substring in Sublinear Time},
  booktitle = {Proc.~of {CPM}},
  year      = {2022},
  pages     = {20:1--20:9},
  bibsource = {dblp computer science bibliography, https://dblp.org},
  doi       = {10.4230/LIPICS.CPM.2022.20},
}

@InProceedings{DBLP:conf/spire/RadoszewskiZ24,
  author    = {Jakub Radoszewski and Wiktor Zuba},
  title     = {Computing String Covers in Sublinear Time},
  booktitle = {Proc.~of {SPIRE}},
  year      = {2024},
  pages     = {272--288},
  bibsource = {dblp computer science bibliography, https://dblp.org},
  doi       = {10.1007/978-3-031-72200-4_21},
}

@Article{DBLP:journals/siamcomp/KociumakaRRW24,
  author    = {Tomasz Kociumaka and Jakub Radoszewski and Wojciech Rytter and Tomasz Walen},
  title     = {Internal Pattern Matching Queries in a Text and Applications},
  journal   = {{SIAM} J. Comput.},
  year      = {2024},
  volume    = {53},
  number    = {5},
  pages     = {1524--1577},
  bibsource = {dblp computer science bibliography, https://dblp.org},
  doi       = {10.1137/23M1567618},
}

@article{DBLP:journals/jcss/BeameJS01,
  author       = {Paul Beame and
                  T. S. Jayram and
                  Michael E. Saks},
  title        = {Time-Space Tradeoffs for Branching Programs},
  journal      = {J. Comput. Syst. Sci.},
  volume       = {63},
  number       = {4},
  pages        = {542--572},
  year         = {2001},
  doi          = {10.1006/JCSS.2001.1778},
  bibsource    = {dblp computer science bibliography, https://dblp.org}
}

@article{DBLP:journals/cjtcs/Pagter05,
  author       = {Jakob Pagter},
  title        = {On {A}jtai's Lower Bound Technique for R-way Branching Programs and
                  the {H}amming Distance Problem},
  journal      = {Chic. J. Theor. Comput. Sci.},
  volume       = {2005},
  year         = {2005},
  url          = {http://cjtcs.cs.uchicago.edu/articles/2005/1/contents.html},
  bibsource    = {dblp computer science bibliography, https://dblp.org}
}

@article{DBLP:journals/tcs/Patt-ShamirP93,
  author       = {Boaz Patt{-}Shamir and
                  David Peleg},
  title        = {Time-Space Tradeoffs for Set Operations},
  journal      = {Theor. Comput. Sci.},
  volume       = {110},
  number       = {1},
  pages        = {99--129},
  year         = {1993},
  doi          = {10.1016/0304-3975(93)90352-T},
  bibsource    = {dblp computer science bibliography, https://dblp.org}
}

@inproceedings{DBLP:conf/isaac/GhaziS25,
  author       = {Taha El Ghazi and
                  Tatiana Starikovskaya},
  title        = {Streaming Periodicity with Mismatches, Wildcards, and Edits},
  booktitle    = {Proc.~of {ISAAC}},
  pages        = {36:1--36:20},
  year         = {2025},
  doi          = {10.4230/LIPICS.ISAAC.2025.36},
}

@inproceedings{DBLP:conf/stoc/KempaK22,
  author       = {Dominik Kempa and
                  Tomasz Kociumaka},
  title        = {Dynamic suffix array with polylogarithmic queries and updates},
  booktitle    = {Proc.~of {STOC}},
  pages        = {1657--1670},
  year         = {2022},
  doi          = {10.1145/3519935.3520061},
}
\end{document}